\documentclass{article}
\usepackage{amssymb,bbm}
\usepackage{amsmath}
\usepackage[margin=1in]{geometry}
\usepackage{amsthm}
\usepackage{xcolor}
\usepackage{enumitem}
\usepackage{mathtools}
\usepackage{caption}
\usepackage{graphicx}
\usepackage{subcaption,dsfont,multirow,array,makecell,ragged2e}
\usepackage{hyperref}
\usepackage{bm}
\usepackage{authblk}

\definecolor{qidong}{HTML}{000080}

\definecolor{ian}{HTML}{006EB8}

\newcommand{\N}{\mathbb N}
\newcommand{\Z}{\mathbb Z}
\newcommand{\Q}{\mathbb Q}
\newcommand{\C}{\mathbb C}
\newcommand{\R}{\mathbb R}
\newcommand{\indicator}[1]{\mathds{1}_{#1}}
\DeclareMathOperator{\interior}{int}
\DeclareMathOperator{\Int}{Int}
\DeclareMathOperator{\loc}{loc}
\DeclareMathOperator{\RHS}{RHS}
\DeclareMathOperator{\LHS}{LHS}
\DeclareMathOperator{\eff}{eff}
\DeclareMathOperator{\ex}{ex}
\newcommand{\norm}[1]{\left\lVert#1\right\rVert}
\newcommand{\abs}[1]{\left\lvert#1\right\rvert}
\newcommand{\set}[1]{\left\{#1\right\}}
\newcommand{\gen}[1]{\left\langle#1\right\rangle}
\newcommand{\evaluate}[2]{\left.#1\right\rvert_{#2}}

\newcolumntype{M}[1]{>{\centering\arraybackslash}m{#1}}

\newcommand\restr[2]{{
  \left.\kern-\nulldelimiterspace
  #1
  \mathchoice{\vphantom{\big|}}{}{}{}
  \right|_{#2}
  }}

\allowdisplaybreaks

\numberwithin{equation}{section}

\theoremstyle{plain}
\newtheorem{theorem}{Theorem}[section]
\newtheorem{lemma}[theorem]{Lemma}
\newtheorem{proposition}[theorem]{Proposition}

\newtheorem{assumption}{Assumption}

\theoremstyle{definition}
\newtheorem{definition}[theorem]{Definition}

\theoremstyle{remark}
\newtheorem{remark}[theorem]{Remark}

\let\vec\mathbf
\let\emptyset\varnothing

\title{High-fugacity expansion and crystallization in non-sliding hard-core lattice particle models without a tiling constraint}

\author{Qidong He}
\author{Ian Jauslin}
\affil{Department of Mathematics, Rutgers University}

\date{}

\begin{document}

\maketitle

\abstract{
  In this paper, we prove the existence of a crystallization transition for a family of hard-core particle models on periodic graphs in arbitrary dimensions.
  We establish a criterion under which crystallization occurs at sufficiently high densities.
  The criterion is more general than that in [Jauslin, Lebowitz, Comm. Math. Phys. {\bf 364}:2, 2018], as it allows models in which particles do not tile the space in the close-packing configurations, such as discrete hard-disk models.
  To prove crystallization, we prove that the pressure is analytic in the inverse of the fugacity for large enough complex fugacities, using Pirogov-Sinai theory.
  One of the main tools used for this result is the definition of a local density, based on a discrete generalization of Voronoi cells.
  We illustrate the criterion by proving that it applies to two examples: staircase models and the radius 2.5 hard-disk model on the square lattice.
}

\tableofcontents

\section{Introduction}

Crystallization is a very well-known phenomenon.
From a physical point of view, it has been studied extensively, and much of what has been observed can be understood using a combination of effective models and numerical simulations \cite{WJ57,AW57,PM86,St88,Mc10,BK11,Isobe2015}.
However, from a mathematical point of view, there is still much to do.
Even proving the existence of crystalline phases in (somewhat) realistic particle models can pose significant challenges \cite{Gaunt1965,Do68,HP74,Ba82,EB05,Hales2005,Theil2006,Mainini2014,Jauslin2018,Mazel2018,Mazel2019,Mazel2020,Mazel2021}.
One model which has received a considerable amount of attention for being both simple to state and somewhat realistic (as well as having applications to other fields such as coding theory \cite{Cohn2017,Viazovska2017}) is the hard-sphere model \cite{Adams1972,Berryman1983,Isobe2015}, in which particles are represented as identical spheres that interact via the constraint that no two spheres can overlap.
However, proving thermodynamic properties of this model has been a huge challenge: even proving that there is crystallization at zero temperature remained open for centuries, from the formulation of the problem by Kepler to the computer-assisted proof by Hales \cite{Hales2005}.
To this day, sphere packing problems are still the subject of active research, with recent breakthroughs in eight and twenty-four dimensions \cite{Cohn2017,Viazovska2017}.
At positive temperature, the problem of crystallization in the hard-sphere model is still wide open.
The main difficulty is that crystallization in a continuum model involves the breaking of the continuous translation symmetry of the system.
The intuition behind this difficulty is that very small defects of a crystal can destroy long-range order.

In this paper, we will focus on a simpler setup: lattice models (throughout this paper, we use the word \emph{lattice} to mean \emph{periodic graph}: we do not restrict ourselves to lattices in the algebraic sense), in which the particles occupy sites of a discrete periodic graph instead of the continuum.
This is a significant simplification, as the breaking of a continuous symmetry required for crystallization in the continuum is reduced to the breaking of a discrete symmetry, which is much easier to accomplish \cite{Mermin1966}.
Nevertheless, proving crystallization in lattice models is still challenging, and developing tools to overcome these difficulties may lead to advances in continuum models as well \cite{Ruelle1971}.
For simplicity, we will further restrict our attention to models in which particles interact solely via a hard-core repulsion, though the results presented below could be adapted to more general potentials, provided that they are short-ranged and sufficiently weak.
We will call these Hard-Core Lattice Particle (HCLP) models.
\bigskip

This paper builds upon \cite{Jauslin2017,Jauslin2018}, in which a class of HCLP models was considered, which satisfy a \emph{non-sliding} condition as well as a \emph{tiling} condition.
The non-sliding condition roughly means that, at high densities, neighboring particles are locked into place, and cannot \emph{slide} with respect to one another.
The tiling condition states that it is possible to tile the lattice with the supports of the particles.
(Technically, the condition in \cite{Jauslin2017,Jauslin2018} is stronger than just non-sliding and tiling, but those are the important aspects of the condition.)
Whereas the former condition is necessary for crystallization, the latter is purely technical.
In fact, Mazel, Stuhl, and Suhov, in their extensive review of lattice regularizations of hard-disk models \cite{Mazel2018,Mazel2019,Mazel2020,Mazel2021}, have constructed an infinite class of lattice regularizations of the hard-disk model that do not slide, but that do not tile the plane because of the presence of interstitial space in between the disks.

In the present paper, we will extend \cite{Jauslin2017,Jauslin2018} by relaxing the tiling condition and prove, using Pirogov-Sinai theory \cite{Borgs1989,Pirogov1975,Zahradnk1984}, that a much larger class of non-sliding HCLP models crystallize.
In particular, all the models studied by \cite{Jauslin2017,Jauslin2018} fall within our general framework.
In addition, we can treat models that do not tile the space, such as a model of hard disks of radius 2.5 on the square lattice; see Figure \ref{fig:octagon}.
We will discuss some explicit examples in Section \ref{sec:examples} of this paper (see Figures \ref{fig:staircase} and \ref{fig:octagon}), but the framework is rather general, and also applies to models in more than two dimensions.
In addition, we have simplified the condition in \cite{Jauslin2017,Jauslin2018}, so the current work makes that criterion for proving crystallization in a HCLP model more easily usable.
\bigskip

Let us be more specific on the condition under which we will prove crystallization.
In this introduction, we will not give formal definitions, which can be found in Section \ref{sec:model} below (see in particular Assumption \ref{assumption}).
We treat HCLP systems on periodic graphs in any dimension.
The crux of the condition concerns the close-packing configurations, which are configurations of particles that maximize the density.
The most important part of the condition is that we require the number of close-packing configurations to be finite; see Figure \ref{fig:ground_state} for an example.
This excludes \emph{sliding} models: if one can slide particles around without lowering the density, then the number of close-packings will be infinite, as is the case for instance in the $2\times2$-square model studied in \cite{Hadas2022}.

The argument we will use is based on controlling {\it defects} in close-packing configurations: if the density is sufficiently high (but not maximal), then the typical configurations will look similar to the close-packing ones.
In \cite{Jauslin2017,Jauslin2018}, the models considered have close-packing configurations that tile the space, so defects could be defined using sites in the lattice that are not covered by particles.
Since we allow non-tiling models, defining defects is more involved.
To do so, we decompose the lattice into {\it generalized Voronoi cells} (see Definition \ref{def:voronoi} below for a formal definition) which assign each and every point in the lattice to its nearest particles.
Thus, whereas the particles do not tile the space, the Voronoi cells cover it (with the caveat that we define Voronoi cells in such a way that they may overlap).
Defects are then defined from the size of Voronoi cells: in close-packing configurations, the cells are all the same size, so when the configuration deviates from a close-packing, cells will expand.
To quantify this, we introduce a {\it local density} for every particle in the configuration; see Definition \ref{def:local_density}.
We then identify defects by finding particles whose local density is lower than in the close-packing configurations.

To make this work without too many complications, we must impose some extra restrictions on the system.
For one, we require that the close-packing configurations be distinct enough, in the sense that two different close-packings cannot merge seamlessly.
In addition, we assume that whenever a particle does not belong to a close-packing configuration, it must impose a dip in the local density, and this dip cannot occur {\it arbitrarily far} from the particle.
These two are important assumptions, without which the discussion would fail dramatically.
In addition to these, we impose additional constraints, which make our arguments easier, but could, in principle, be relaxed in future work, without changing the method too much.
One of these is that we impose that different close-packing configurations are all related to each other by isometries, which ensures that the local density will be the same in different close-packings.
In addition, we exclude the possibility that the local density could exceed the total density.
This can happen in certain models \cite{Hales2010}, and this would break a number of arguments made in our proof.
\bigskip

Under this condition, we prove that crystallization occurs at sufficiently high densities.
To do so, we follow the same philosophy as in \cite{Jauslin2018}, and prove that the model has a convergent {\it high-fugacity expansion}, that is, an analytic expansion in the inverse of the fugacity (an expansion in $e^{-\mu}$ where $\mu$ us the chemical potential).
The idea of a high fugacity expansion for HCLP models dates back, at least, to Gaunt and Fisher \cite{Gaunt1965} (see also \cite{Do68,HP74}), and was systematized in \cite{Jauslin2018}.
The present work is a continuation of \cite{Jauslin2018}, and we extend the treatment of such expansions to a much wider class of models.
In particular, we prove that the Lee-Yang zeros \cite{YL52,LY52} are all located inside a finite-radius disk in the complex fugacity plane.
Combining this with a classical Mayer expansion argument \cite{Ma37,Ur27,Ru63}, we thus prove that the Lee-Yang zeros lie in a finite annulus in the complex fugacity plane.

To prove the convergence of the high-fugacity expansion, we use Pirogov-Sinai theory \cite{Borgs1989,Pirogov1975,Zahradnk1984}, which allows us to balance the costs coming from the drops in the density caused by defects with the entropy gains the defects produce.
\bigskip

The rest of this paper is structured as follows.
In Section \ref{sec:model}, we define the model more precisely, state the condition under which we will prove crystallization (see Assumption \ref{assumption}), and state our main results.
These are the staircase models, and the 12th nearest neighbor exclusion.
In Section \ref{sec:GFc}, we map the HCLP particle model to a {\it contour} model.
Following \cite{Jauslin2017,Jauslin2018}, we call these contours \emph{Gaunt-Fisher configurations}.
These formalize the notion of \emph{defect} mentioned above, which really should be understood as \emph{Gaunt-Fisher configurations}.
In Section \ref{sec:peierls}, we prove the crucial estimate that will allow Pirogov-Sinai theory to work for our system: the {\it Peierls condition}.
Roughly, that states that the cost of a Gaunt-Fisher configuration is exponentially large in its size, which will allow us to control the entropy of contours.
In Section \ref{sec:Pirogov_Sinai}, we carry out the Pirogov-Sinai analysis.
Our approach is similar to that of Zahradn\'ik \cite{Zahradnk1984}, and readers unfamiliar with Pirogov-Sinai theory may want to study that reference to understand the philosophy behind the method (see also the textbook \cite{Friedli2017}).
Finally, in Section \ref{sec:examples}, we discuss some explicit examples of models for which we prove Assumption \ref{assumption}.

\section{Model and main result}\label{sec:model}

Let $\Lambda_{\infty}$ be a periodic graph embedded in $\R^{d}$.
For example, $\Lambda_{\infty}$ could be $\mathbb Z^d$, the triangular lattice, or the honeycomb lattice (which is not, strictly speaking, a lattice, but rather a periodic graph).
Denote by $\mathrm{d}_{\Lambda_{\infty}}$ the (usual) graph distance on $\Lambda_{\infty}$.
Our interest is in Hard-Core Lattice Particle (HCLP) systems on $\Lambda_{\infty}$, which we formalize as follows.

Each particle has a \emph{shape} denoted by $\omega$, which is a bounded subset of $\R^{d}$ and, for convenience, is assumed to contain $\vec{0}$. 
Hence, a particle at $x\in\Lambda_{\infty}$ occupies the volume $\omega_{x}:= x+\omega$.
We require that each $\sigma_{x}:= \omega_{x}\cap\Lambda_{\infty}$ induce a connected subgraph of $\Lambda_{\infty}$. 
Pairs of particles interact via a hard-core repulsion, that is, their supports may not overlap.
Formally, given any $\Lambda\subseteq\Lambda_{\infty}$, we define the set of particle configurations on $\Lambda$ as
\begin{equation}
\Omega(\Lambda):=\set{X\subseteq\Lambda\mid\omega_{x}\cap\omega_{x'}=\emptyset\text{ for all }x\ne x'\in X}.
\end{equation}
We will study this system in the grand canonical ensemble:
if $\Lambda$ is finite, we define the partition function at fugacity $z$ ($:= e^{\beta\mu}$, where $\mu$ is the chemical potential and $\beta$ the inverse temperature) as
\begin{equation}
\Xi_{z}(\Lambda):=\sum_{X\in\Omega(\Lambda)}z^{\abs{X}},
\end{equation}
where $\abs X$ is the number of elements in $X$.
Let 
\begin{equation}
\rho_{\max}(\Lambda):=\frac{1}{\abs{\Lambda}}\max_{X\in\Omega(\Lambda)}\abs{X}
,\quad
\rho_{\max}:=\lim_{\Lambda\Uparrow\Lambda_{\infty}}\rho_{\max}(\Lambda)
\label{rhomax}
\end{equation} 
be the maximal density and its infinite-volume limit.
Finally, define the finite-volume pressure of the system as
\begin{equation}
p_{z}(\Lambda):=\frac{1}{\abs{\Lambda}}\log\Xi_{z}(\Lambda)
\end{equation} 
and its infinite-volume limit
\begin{equation}
p(z):=\lim_{\Lambda\Uparrow\Lambda_{\infty}}p_{z}(\Lambda).
\end{equation}

Our main result is that, provided the model satisfies a \emph{non-sliding} condition along with other geometric constraints (see Assumption \ref{assumption} below), the system crystallizes at high densities in the sense that there is long-range order in the positions of particles that breaks the translation symmetry of $\Lambda_{\infty}$.

\subsection{Assumption on the model}
To specify the assumption on the model, we will need a few definitions.

First, we will assume that $\Lambda_\infty$ is such that the boundary of any connected set is connected in a coarse-grained sense, which we will now define.

\begin{definition}\label{def:rconnected}
Two points $x,y\in\Lambda_\infty$ are \emph{neighbors} if and only if $d_{\Lambda_\infty}(x,y)\le1$, which gives us a natural notion of connectedness in $\Lambda_\infty$.
In addition, a set $S\subset\Lambda_\infty$ is said to be $r$-connected if $\forall x,y\in S$, there exists a path $x\equiv x_0,x_1,\cdots,x_N\equiv y\in S$ such that $d_{\Lambda_\infty}(x_i,x_j)\le r$.
\end{definition}

We will assume that $\Lambda_\infty$ is such that there exists $\mathcal R_0\in\mathbb N$ such that the interior and exterior boundaries (see Definition \ref{def:boundaries}) of any simply connected set are $\mathcal R_0$-connected (see Assumption \ref{assumption} below).
This is a very weak assumption that was shown to hold for a very large class of graphs \cite{Timar2013} including $\mathbb Z^d$ (for which $\mathcal R_0=d$), the triangular lattice (for which $\mathcal R_0=1$), and the honeycomb lattice (for which $\mathcal R_0=3$).
See Figure \ref{fig:boundary_rconnected} for an example.

\begin{definition}\label{def:boundaries}
  Given a connected set $\Lambda\subset\Lambda_{\infty}$, we define its interior boundary as
  \begin{equation}
    \partial^{\mathrm{in}}\Lambda:=\set{\lambda\in\Lambda:\mid d_{\Lambda_\infty}(\lambda,\Lambda^c)=1}
  \end{equation}
  and its exterior boundary as
  \begin{equation}
    \partial^{\mathrm{ex}}\Lambda:=\set{\lambda\in\Lambda^c\mid d_{\Lambda_\infty}(\lambda,\Lambda)=1}
    .
  \end{equation}
\end{definition}

\begin{figure}
  \hfil\includegraphics[width=4cm]{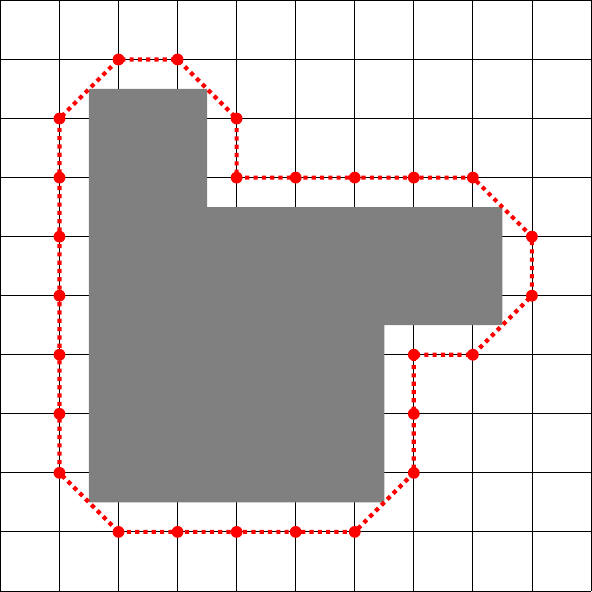}

  \caption{An example of a subset of $\mathbb Z^2$ and its boundary: the boundary is 2-connected.}
  \label{fig:boundary_rconnected}
\end{figure}

Now, let us define the notion of ground states, which could also be called \emph{close-packing states}.

\begin{definition}[ground state]
A \emph{ground state} in $\Lambda$ is a configuration $X\in\Omega(\Lambda)$ that maximizes the density: $\abs{X}=\abs{\Lambda}\rho_{\max}(\Lambda)$.
Taking the limit $\Lambda\Uparrow\Lambda_{\infty}$ in the sense of van Hove, the ground states tend to limiting configurations in $\Omega(\Lambda_{\infty})$.
An (infinite-volume) ground state is denoted by $\mathcal L^\#$ where $\#$ takes values in $\mathcal G$.
In other words, $\mathcal G$ is a set of indices, each of which specifies a ground state.
\end{definition}

See Figure \ref{fig:ground_state} for an example.

\begin{figure}
  \hfil\includegraphics[width=4.8cm]{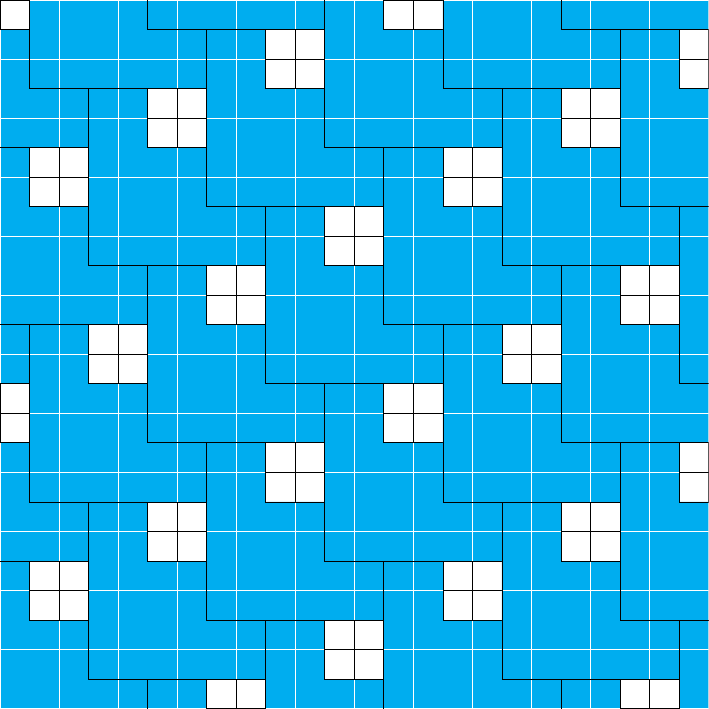}
  \caption{A section of one of the ground states for the 3-staircase model; see Section \ref{subsec:n_staircases}.}
  \label{fig:ground_state}
\end{figure}

We will assume that $\mathcal{G}$ is finite and that the ground states are periodic.
Moreover, we will assume that the different ground states are related to each other by \emph{species-preserving isometries}, which are invertible transformations of $\R^{d}$ that preserve the shapes of particles:

\begin{definition}[species-preserving isometry]
A species-preserving isometry is a Euclidean transformation $\psi$ satisfying the following properties:
\begin{enumerate}
\item the restriction $\restr{\psi}{\Lambda_{\infty}}$ induces a graph automorphism of $\Lambda_{\infty}$ (in particular, it is an isometry with respect to the graph distance $\mathrm{d}_{\Lambda_{\infty}}$);
\item it leaves the supports of the particles invariant: for all $\lambda\in\Lambda_{\infty}$, $\psi(\omega_{\lambda})=\omega_{\psi(\lambda)}$.
\end{enumerate}
\end{definition}

Whereas ground states maximize the density, it is a priori possible for them not to maximize the density \emph{locally}.
This can happen in certain models in which clusters that have a large local density may form at the expense of lowering the local density elsewhere, so that these clusters do not form ground states \cite{Hales2010}.
We will assume that this is not the case.
To state this more precisely, we need to define the notion of a \emph{local density}, which is in turn based on discrete Voronoi cells \cite{Erwig2000,Mehlhorn1988}.

\begin{definition}[discrete Voronoi cell]\label{def:voronoi}
For $X\in\Omega(\Lambda_{\infty})$ and $x\in X$, the (discrete) Voronoi cell of $\sigma_{x}$ ($:=\omega_{x}\cap\Lambda_{\infty}$) with respect to $X$ is defined as the set of points that are closer (inclusively) to $\sigma_{x}$ than to any other particle:
\begin{equation}
V_{X}(\sigma_{x}):=\set{\lambda\in\Lambda_{\infty}\mid\mathrm{d}_{\Lambda_{\infty}}(\lambda,\sigma_{x})=\min_{y\in X}\mathrm{d}_{\Lambda_{\infty}}(\lambda,\sigma_{y})}.
\end{equation}
\end{definition}

Examples of Voronoi cells are provided in Figures \ref{fig:staircase_close_packings} and \ref{fig:disk_close_packings} below.

\begin{remark} \label{rem:no_choice_voronoi_cell}
Usually, the Voronoi cells are constructed to form a partition of the whole space, which necessitates a choice on the cell boundaries. 
In contrast, the definition we take in this paper allows distinct Voronoi cells to overlap, with the benefit of enabling us to implement Pirogov-Sinai theory in a natural (i.e., choice-free) way, as we will see later.
\end{remark}

We can now define the local density, which is the inverse of the size of the Voronoi cell surrounding a particle, but adjusted for the fact that the Voronoi cells can overlap. 
Formally:

\begin{definition}[local density]\label{def:local_density}
Given $X\in\Omega(\Lambda_{\infty})$, define the local density at $x\in X$ as 
\begin{equation} \label{eqn:local_density}
\rho_{X}(x):=\left(\sum_{\lambda\in V_{X}(\sigma_{x})}\frac{1}{\abs{\set{z\in X\mid\lambda\in V_{X}(\sigma_{z})}}}\right)^{-1}.
\end{equation}
Accordingly, the maximum local density of the model is
\begin{equation}
\rho_{\max}^{\loc}:=\sup_{\substack{X\in\Omega(\Lambda_{\infty})\\X\ni\vec{0}}}\rho_{X}(\vec{0}).
\end{equation}
\end{definition}

Examples of the computation of $\rho_{\mathrm{max}}^{\mathrm{loc}}$ is provided in Figures \ref{fig:staircase_close_packings} and \ref{fig:disk_close_packings} below.

\begin{remark}
Notice that $\rho_{\max}^{\loc}=\rho_{\max}$ in the case of a tiling model, in which all sites are covered in the ground states.
For general models, however, this is not always the case.
For instance, in the hard-sphere model, it is possible to have localized configurations in which the local density (defined as the inverse of the volume of the standard Voronoi cell in the continuum) exceeds the close-packing density \cite{Hales2010}.
In this paper, we will only consider models for which the equality $\rho_{\max}^{\loc}=\rho_{\max}$ holds, but without requiring the tiling property.
\end{remark}

Next, we introduce the notion of $\#$-correctness, which is inspired by the construction of the (graph) Voronoi dual \cite{Honiden2010} and will later be used to identify defects in a configuration.

\begin{definition}[neighbor and $\#$-correctness]\label{def:correct}
For a given configuration $X$, we define the \emph{neighbors} of a particle to be those in Voronoi cells adjacent to that of the particle: given $x\in X$,
\begin{equation}
\mathcal{N}_{X}(x):=\set{z\in X\mid\mathrm{d}_{\Lambda_{\infty}}(V_{X}(\sigma_{x}),V_{X}(\sigma_{z}))\leq 1}.
\end{equation}
(Note that this definition implies that $x\in\mathcal N_X(x)$.)
A particle $x\in X$ is then said to be $\#$-correct if its neighbors in $X$ are exactly its neighbors in $\mathcal{L}^{\#}$, that is, $x\in\mathcal{L}^{\#}$ and
\begin{equation}
\mathcal{N}_{X}(x)=\mathcal{N}_{\mathcal{L}^{\#}}(x).
\end{equation}
Finally, $x\in X$ is said to be incorrect if $x$ is not $\#$-correct for any ground state $\#\in\mathcal{G}$.
\end{definition}

In addition, we define a coarse-grained notion of $\#$-correctness, which will come in useful throughout the discussion.

\begin{definition}[$\mathcal{R}_2$-neighbor and $(\#,\mathcal{R}_2)$-correctness]
\label{def:R_correct}
Given $\mathcal R_2\in\mathbb N$ (we use the subscript ${}_{2}$ for reasons that will become apparent later), for $X\in\Omega(\Lambda_{\infty})$ and $x\in X$, we define the set of $\mathcal{R}_2$-neighbors of $x$ as
\begin{equation}
  \mathcal{N}^{(\mathcal{R}_2)}_{X}(x):=\set{z\in X\mid\mathrm{d}_{\Lambda_{\infty}}(V_{X}(\sigma_{x}),V_{X}(\sigma_{z}))\leq\mathcal{R}_2}
  .
\end{equation}
A particle $x\in X$ is
is said to be $(\#,\mathcal{R}_2)$-correct if
\begin{equation}
  \textrm{for all } y\in\mathcal N_{X}^{(\mathcal R)}(x),\quad y\mathrm{\ is\ }\#\mathrm{-correct}.
\end{equation}
In addition, $x\in X$ is said to be $\mathcal{R}_2$-incorrect if $x$ is not $(\#,\mathcal{R}_2)$-correct for any $\#\in\mathcal G$.
Finally, given a ground state $\#$, let $\mathcal{C}^{(\mathcal{R}_2)}_{\#}(X)$ denote the set of $(\#,\mathcal{R}_2)$-correct particles in $X$, and $\mathcal{I}^{(\mathcal{R}_2)}(X):=X\setminus\bigcup_{\#\in\mathcal{G}}\mathcal{C}^{(\mathcal{R}_2)}_{\#}(X)$ the set of $\mathcal{R}_2$-incorrect particles.
\end{definition}

We can now state the conditions we need to require of the models we consider here.

\begin{assumption}\label{assumption}
We require that the model satisfy the following properties:
\begin{enumerate}
  \item $\Lambda_\infty$ is a periodic graph embedded in $\mathbb R^d$ with finite maximal coordination number (the number of neighbors is bounded), and is such that there exists $\mathcal R_0\in\mathbb N$ such that the interior and exterior boundaries (see Definition \ref{def:boundaries}) of any simply connected set are $\mathcal R_0$-connected; see Definition \ref{def:rconnected}.\label{asm:lattice}

  \item There exist only finitely many ground states: $\mathcal{G}$ is finite and nonempty. \label{asm:finitely_many_ground_states}

  \item
  The ground states are related by species-preserving isometries: given $\#,\#'\in\mathcal{G}$, there exists a species-preserving isometry $\psi$ such that $\psi(\mathcal{L}^{\#})=\mathcal{L}^{\#'}$.\label{asm:isometry}

\item The maximum density is equal to the maximum local density: $\rho_{\mathrm{max}}^{\mathrm{loc}}=\rho_{\mathrm{max}}$. \label{asm:density_max_local}

\item Different ground states cannot merge seamlessly: for any $X\in\Omega(\Lambda_{\infty})$, if $x\in X$ is $\#$-correct and $x'\in\mathcal{N}_{X}(x)$ is $\#'$-correct, then $\#=\#'$. \label{asm:imperfect_transition}

\item A particle that is incorrect leads to a localized dip in the local density: there exist $\mathcal R_1,\mathcal S_1\in\mathbb N$ and $\epsilon>0$ such that, for all $X\in\Omega(\Lambda)$ and $x\in X$, if $x$ is $\mathcal R_1$-incorrect, then there exists $y\in X$ such that $d_{\Lambda_{\infty}}(x,y)\le\mathcal S_1$ and $\rho_{X}^{-1}(y)\ge\rho_{\mathrm{max}}^{-1}+\epsilon$.\label{asm:density_local_density}
\end{enumerate}
\end{assumption}

\begin{remark}
  In the simplest cases, Item \ref{asm:density_local_density} holds for $\mathcal R_1=\mathcal S_1=0$.
  These are cases in which the local density is maximal if and only if a particle is $(\#,0)$-correct, which includes all of the tiling examples discussed in \cite{Jauslin2018}, as well as the staircase models discussed in \S\ref{subsec:n_staircases}.
  For these models, it is relatively straightforward to verify Item \ref{asm:density_local_density}.
  However, there also may be situations in which the local density may only dip at a finite, but large distance $\mathcal S_1$, in which case proving the assumption may be more difficult.
  To deal with such cases, we prove the following lemma, which provides an equivalent assumption that may be easier to verify.
\end{remark}

\begin{lemma}\label{lem:equiv_assum}
  Assuming Items \ref{asm:density_max_local} and \ref{asm:imperfect_transition} of Assumption \ref{assumption} hold, Item \ref{asm:density_local_density} holds if and only if
  \begin{equation}
    \set{\mathcal L^\#\mid\#\in\mathcal G}=
    g_m
    :=\set{X\in\Omega(\Lambda_{\infty})\mid\rho_{X}(x)=\rho_{\max}\text{ for all } x\in X}
    \label{Ggm}
  \end{equation}
  (that is, the set of configurations that have a constant local density which is maximal is equal to the set of ground states).
\end{lemma}

This lemma is proved in \S\ref{sec:equiv_assumption}.
It allows one to use our main result in situations where it is easier to prove that $\set{\mathcal L^\#\mid\#\in\mathcal G}=g_m$ than to prove Item \ref{asm:density_local_density} of Assumption \ref{assumption}.

We briefly describe how Assumption \ref{assumption} will enter our analysis.
\begin{itemize}
  \item
  Item \ref{asm:finitely_many_ground_states} allows us to control the number of ways a ground state can be perturbed, thus controlling the entropy of defects: if the number of ground states were infinite, there is the risk that there are too many ways to create defects, which could lead to defect-formation being likely; see Lemma \ref{lem:clusterexpansion}.

  \item
  Item \ref{asm:isometry} allows us to control the ratio of partition functions of defects of different ground states, which is a standard step in Pirogov-Sinai theory; see Proposition \ref{prop:central_estimates}, specifically \eqref{bulkterm}.

  \item
  Item \ref{asm:density_max_local} excludes situations in which the local density can be smaller than $\rho_{\mathrm{max}}$, which would break the proof of Proposition \ref{prop:central_estimates}; see Lemma \ref{lemma:rholoc_cst}.

  \item
  Items \ref{asm:lattice} and \ref{asm:imperfect_transition} allow us to map the particle model to a \emph{contour} model in a one-to-one way: using this assumption, we can fully specify a defect-free region by only looking at the particles neighboring this region; see Lemma \ref{lemma:unique_Rcorrect} and Proposition \ref{prop:GFc}.

  \item
  Item \ref{asm:density_local_density} ensures that the presence of a defect causes the local density to dip below its maximum value and that this dip occurs {\it close} to the defect, from which we can conclude the validity of the {\it Peierls condition}; see Lemma \ref{lem:Peierls}.
\end{itemize}

\begin{remark}[necessity of Assumption \ref{assumption}]
Assumption \ref{assumption} excludes models that exhibit the sliding phenomenon.
For example, consider the model of $2\times 2$ squares on $\Z^{2}$, in which one can shift entire columns of particles without disturbing a ground state \cite{Hadas2022}.
Such moves generate an infinite number of ground states, which violates Item \ref{asm:finitely_many_ground_states}.
Due to the abundance of ground states that differ from each other only by a number of columns, Item \ref{asm:imperfect_transition} is also violated.

Assumption \ref{assumption}, however, requires more of the model than simply that it does not permit sliding.
In principle, this suggests that the assumption can be relaxed.
For instance, the requirement in Item \ref{asm:lattice} that the graph $\Lambda_{\infty}$ is periodic is presumably not necessary, though this would require some changes in the argument, and it is not clear that the rest of the assumptions could be satisfied for aperiodic graphs.
More interestingly, there exist non-sliding models that violate Item \ref{asm:density_max_local}, but for which one nevertheless expects crystallization to occur at high densities.
In particular, the hard-sphere model in $\R^{3}$ is known to allow for local configurations to exceed the close-packing density but at the expense of the overall density \cite{Hales2010} and so $\rho_{\mathrm{max}}<\rho_{\mathrm{max}}^{\mathrm{loc}}$.
Thus, it may also be possible to weaken Item \ref{asm:density_local_density}, although our treatment of the Peierls condition in Lemma \ref{lem:Peierls} would need to be adapted accordingly.
\end{remark}

\begin{remark}[comparison of Assumption \ref{assumption} with \cite{Jauslin2018}]
  The condition in Assumption \ref{assumption}, under which we prove crystallization, is more general than the condition in \cite{Jauslin2018}.
  Obviously, we do not require the particle to tile $\Lambda_\infty$, but the generalization goes a little further.
  For one thing, we do not require that the ground states be periodic, as this can be proved from Items \ref{asm:lattice} and \ref{asm:finitely_many_ground_states}; see Lemma \ref{lemma:periodic_ground_state} below.
  In addition, the analog of Item \ref{asm:density_local_density} in \cite{Jauslin2018} requires the drop in the density to occur {\it right next} to the incorrect particle.
  Here, we are more general and allow for the drop in density to occur farther away, which allows us to treat models such as the hard-disk model in Section \ref{sec:examples}.
\end{remark}

\subsection{Main result}

Our main result is that, under Assumption \ref{assumption}, there is more than one extremal Gibbs state for large enough fugacities.
In each of these states, the translation invariance is broken, which shows that these high-fugacity states exhibit crystalline order.
This proves the existence of a phase transition: indeed, a Mayer expansion easily shows that the Gibbs state is unique at low fugacity.
Note, however, that our result does not give a value for the critical fugacity, or even on the number of phase transitions.
We will also derive, as a byproduct of our analysis, a convergent \emph{high-fugacity expansion} for the pressure:
\begin{equation}
p(z):=\lim_{\Lambda\Uparrow\Lambda_{\infty}}p_{z}(\Lambda)=\rho_{\max}\log z+f(z^{-1}),
\end{equation}
where $\rho_{\max}:=\limsup_{\Lambda\Uparrow\Lambda_{\infty}}\rho_{\max}(\Lambda)$, and $f$ is analytic in $z^{-1}$ for all sufficiently large $\abs{z}$.

To show that there are multiple extremal Gibbs states, we will consider a family of boundary conditions, each corresponding to a different instance of symmetry breaking. 
The precise definition of the boundary condition is of little importance for the techniques used here.
For the sake of convenience, we choose a boundary condition that is well-adapted to the Pirogov-Sinai construction detailed below, though a similar argument would allow us to treat more general situations.
To introduce this boundary condition, we must first define a coarse-graining, parametrized by a radius $\mathcal R_2\in\N$.
We will choose $\mathcal R_2$ to be large enough so as to satisfy \eqref{ineq_R2_1}, \eqref{ineq_R2_2}, and \eqref{ineq_R2_3} below.
The utility of the coarse-graining parameter $\mathcal R_2$ is technical and will become clear later.

\begin{definition}[boundary condition] \label{def:boundary_condition_correctness}
Given a ground state $\#$ and a finite region $\Lambda\Subset\Lambda_{\infty}$, we define the $\#$ boundary condition as follows: all the particles outside $\Lambda$ are in $\mathcal{L}^{\#}$, and all the particles along the boundary of $\Lambda$ are $(\#,\mathcal R_2)$-correct.
Formally, we define the set of configurations in $\Lambda$ subject to the boundary condition $\#$ as:
\begin{equation}
\begin{multlined}
\Omega^{\#}(\Lambda):=\left\{X\in\Omega(\Lambda_{\infty})\mid X\setminus\Lambda=\mathcal{L}^{\#}\setminus\Lambda,\right.
\\
\left.\text{for all } x\in X\text{ such\ that }\mathrm{d}_{\Lambda_{\infty}}(V_{X}(\sigma_{x}),\Lambda^{c})\le 1,\ x\ \mathrm{is\ }(\#,\mathcal R_2)\mathrm{-correct}\vphantom{\mathcal{L}^{\#}}\right\}.
\end{multlined}
\end{equation}
The (grand canonical) partition function on $\Lambda$ at fugacity $z$ with boundary condition $\#$ is
\begin{equation}\label{eqn:partitionfunctiondefinition}
\Xi^{\#}_{z}(\Lambda):=\sum_{X\in\Omega^{\#}(\Lambda)}z^{\abs{X\cap\Lambda}},
\end{equation}
and the probability of a configuration $X\in\Omega^{\#}(\Lambda)$ is defined as
\begin{equation}
  \frac{z^{\abs{X\cap\Lambda}}}{\Xi_{z}^{\#}(\Lambda)}
  .
\end{equation}
We denote the corresponding expectation by $\left<\cdot\right>_{z,\Lambda}^{\#}$.
\end{definition}

Our main result can be formally stated as follows.

\begin{theorem}[crystallization] \label{thm:crystallization}
  Under Assumption \ref{assumption}, there exists a constant $z_{0}>0$, independent of $\Lambda$, such that, if $z>z_{0}$, then
  \begin{equation}
    \lim_{\Lambda\Uparrow\Lambda_{\infty}}\left<\indicator{x}\right>_{z,\Lambda}^{\#}
    =\begin{cases}
    \rho_{\max}+O(z^{-1})&\mathrm{if\ }x\in \mathcal{L}^{\#}\\
      O(z^{-1})&\mathrm{otherwise}
    \end{cases},
  \end{equation}
where $\indicator{x}$ denotes the characteristic function that $x\in X$.
\end{theorem}

\begin{remark}
We have computed a value for $z_0$, which is deferred to the appendix: see \eqref{estimate_z}.
This bound is quite far from optimal, but it is instructive to see how the parameters appearing in Assumption \ref{assumption} affect the radius of convergence.
\end{remark}

As a consequence of Theorem \ref{thm:crystallization}, there are at least as many extremal Gibbs distributions as there are close-packing configurations ($\abs{\mathcal{G}}$), in all of which the translational symmetry of $\Lambda_{\infty}$ is broken.
Also, as we have noted, an intermediate result in the proof of Theorem \ref{thm:crystallization} is the construction of an expansion of $p(z)-\rho_{\max}\log z$ in powers of $z^{-1}$, which is shown to be absolutely convergent when $\abs{z}$ is sufficiently large.
We summarize the latter as a standalone theorem.

\begin{theorem}[analyticity] \label{thm:analyticity}
Under Assumption \ref{assumption}, $p(z)-\rho_{\max}\log z$ is analytic in $z^{-1}$ on $\set{z\in\C\mid\abs{z}>z_{0}}$.
\end{theorem}

Let us give a brief outline of the proof of Theorem \ref{thm:analyticity}, from which Theorem \ref{thm:crystallization} is proved.
The first step is to map the model to a contour model.
Here, a \emph{contour} will be called a \emph{Gaunt-Fisher configuration}, in honor of \cite{Gaunt1965}, and abbreviated GFc.
The GFc's are constructed from the incorrect particles, and are chosen to be thick enough so that, in the effective GFc model, pairs of GFc's only interact via a hard-core repulsion.
In addition, GFc's will retain information on the particles inside them, and on the index of the close-packing outside the GFc.
This will allow the mapping between configurations and GFc's to be one-to-one.
This is carried out in Section \ref{sec:GFc}.

Next, we prove that the weight of a GFc in the effective GFc model is {\it exponentially small} in its size.
This is called the {\it Peierls condition}.
To prove it, we use Item \ref{asm:density_local_density} to show that a certain proportion of sites in the support of the GFc has a local density that is $<\rho_{\mathrm{max}}$.
This is done in Section \ref{sec:peierls}; see Lemma \ref{lem:Peierls} in particular.

We then unravel the Pirogov-Sinai machinery.
The idea is to prove that the GFc model is dilute enough that we can compute its observables using a convergent cluster expansion.
We do so using the Peierls condition.
There is a difficulty we have to contend with: GFc's actually interact with each other through the fact that the close-packing outside two neighboring GFc's has to be the same.
To eliminate this long-range interaction, we use the Minlos-Sinai trick, which consists in \emph{flipping} contours in such a way that they all have the same external close-packing.
Doing so comes at a cost, which we can estimate using the cluster expansion inductively.
This is carried out in Section \ref{sec:Pirogov_Sinai}, and leads to the proof of Theorem \ref{thm:analyticity}.

To prove Theorem \ref{thm:crystallization}, we allow the fugacity to vary infinitesimally and locally, and compute derivatives with respect to local fluctuations of the fugacity.
Using the convergent expansion proved in Theorem \ref{thm:analyticity}, this allows us to compute the required observables to the required order in $z^{-1}$.

\section{Gaunt-Fisher configurations}\label{sec:GFc}

The main step toward deriving the analyticity in Theorem \ref{thm:analyticity} is to map the particle model to a contour model.
Intuitively, the contours associated to a particle configuration form a localized, complete record of defects, that is, deviations of the configuration from the ground states.
Following \cite{Gaunt1965,Jauslin2018}, we will refer to contours as \emph{Gaunt-Fisher configurations} (abbreviated as \emph{GFc's}).

\subsection{Some useful lemmas}

We will consider $\mathcal R_2$-incorrect particles as responsible for the formation of defects and accordingly construct GFc's from particle configurations using the $\mathcal R_2$-incorrect particles.
Before delving into the study of defects, however, let us first investigate the properties of particles living in defect-free configurations, that is, the ground states.

We first prove that ground states must be periodic, which will be useful in proving Lemmas \ref{lemma:rholoc_cst} and \ref{lem:finiteeffectiveparticle} below.

\begin{lemma}\label{lemma:periodic_ground_state}
  The ground states are periodic: for any $\#\in\mathcal G$, there exist linearly independent vectors $k_1,\cdots,k_d\in\mathbb R^d$ such that, for all $i=1,\cdots,d$,
  \begin{equation}
    \mathcal L^\#=\mathcal L^\#+k_i
    .
  \end{equation}
\end{lemma}

\begin{proof}
  We prove this by contradiction: suppose that for every linearly independent family $k_1,\cdots,k_d$, $\mathcal L^\#$ is not invariant under $k_i$ translations.
  Choose an infinite family of $k_i$ that are translations of $\lambda_\infty$.
  In this case, there would be an infinite number of ground states, obtained from $\mathcal L^\#$ by translating by $k_i$, which contradicts Item \ref{asm:finitely_many_ground_states} of Assumption \ref{assumption}.
\end{proof}

We will refer to the Voronoi cell of a particle in a ground state as a \emph{reference} Voronoi cell.
To each point in a reference Voronoi cell (of a ground state $\#$), we will also assign a fractional weight that is the reciprocal of the number of particles in $\#$ to which it is equidistant (recall that the Voronoi cells overlap; see Definition \ref{def:voronoi}).
Later on, the weights will enable us to quantify the drop in density due to the presence of defects. 

\begin{definition}[reference Voronoi cell]
\label{def:effective_particle}
Given a ground state $\#$ and $x\in \mathcal{L}^{\#}$, we denote the reference Voronoi cell of a particle $x$ in the ground state $\#$ by 
\begin{equation}
\sigma^{\#}_{x}:= V_{\mathcal{L}^{\#}}(\sigma_{x}).
\label{sigmasharp}
\end{equation} 
Moreover, we define a weight function $v^{\#}:\sigma^{\#}_{x}\rightarrow\Q$ on the reference Voronoi cell, 
\begin{equation}
v^{\#}(\lambda):=\frac{1}{\abs{\set{z\in \mathcal{L}^{\#}\mid\lambda\in\sigma^{\#}_{z}}}}.
\label{vsharp}
\end{equation}
Note that, by Definition \ref{def:local_density},
\begin{equation}
  \sum_{\lambda\in\sigma^\#_x}v^\#(\lambda)
  =\rho_{\mathcal L^\#}(x)^{-1}
  .
  \label{sumv_rholoc}
\end{equation}
\end{definition}

We note a few basic properties of the reference Voronoi cells.

\begin{lemma}\label{lem:finiteeffectiveparticle}
The support of a reference Voronoi cell in any ground state $\#$ is bounded by a radius $r_{\eff}$ that is independent of $\#$.
\end{lemma}

\begin{proof}
By contradiction, suppose that the reference Voronoi cell of $x\in\mathcal{L}^{\#}$ has infinite size.
Then, for all $r\in\N$, there exists $z\in\sigma_{x}^{\#}$ such that $\mathrm{d}_{\Lambda_{\infty}}(x,z)>r$.
By Definition \ref{def:voronoi}, the ball of radius $r$ around $z$ must be devoid of particles: for all $y\in\mathcal{L}^{\#}$, $\mathrm{d}_{\Lambda_{\infty}}(z,\sigma_{y})>r$.
Since the support $\omega$ of a particle is bounded, one can always choose $r$ large enough that a particle fits well inside the ball of radius $r$ centered at $z$, which contradicts the fact that $\mathcal{L}^{\#}$ is a ground state.
Thus, the support of the reference Voronoi cells in the ground state $\#$ is bounded.
Finally, we can make the bound $r_{\eff}$ independent of $\#$ by taking the largest value among the close-packings (using Item \ref{asm:finitely_many_ground_states} of Assumption \ref{assumption} and Lemma \ref{lemma:periodic_ground_state}).
\end{proof}

This lemma has several useful consequences.
The first is that the local density of every ground state is constant.

\begin{lemma}\label{lemma:rholoc_cst}
  The local density (see Definition \ref{def:local_density}) of a ground state is constant: for all $x\in\mathcal L^\#$,
  \begin{equation}
    \rho_{\mathcal L^\#}(x)=\rho_{\mathrm{max}}
    .
    \label{rho_max_cst}
  \end{equation}
  In particular,
  \begin{equation}
    \sum_{\lambda\in\sigma^\#_x}v^\#(\lambda)
    =\rho_{\mathrm{max}}^{-1}
    .
    \label{sumv_rhomax}
  \end{equation}
\end{lemma}

\begin{proof}
  By Definition \ref{def:local_density}, for any $X\Subset\mathcal L^\#$,
  \begin{equation}
    \sum_{x\in X}\rho_{\mathcal L^\#}(x)^{-1}=
    \sum_{x\in X}\sum_{\lambda\in \sigma_x^\#}\frac1{\abs{\set{z\in \mathcal L^{\#}\mid\lambda\in \sigma_z^\#}}},
  \end{equation}
  so, if $\Lambda_X:=\bigcup_{x\in X}\sigma_x^\#$,
  \begin{equation}
    \sum_{x\in X}\rho_{\mathcal L^\#}(x)^{-1}=
    \sum_{\lambda\in \Lambda_X}\frac{\abs{\set{x\in X\mid\lambda\in \sigma_x^\#}}}{\abs{\set{z\in \mathcal L^\#\mid\lambda\in \sigma_z^\#}}}
    =
    \sum_{\lambda\in \Lambda_X}\left(
      1-\frac{\abs{\set{x\in \mathcal L^\#\setminus X\mid\lambda\in \sigma_x^\#}}}{\abs{\set{z\in \mathcal L^\#\mid\lambda\in \sigma_z^\#}}}
    \right),
  \end{equation}
  and thus
  \begin{equation}
    \frac1{\abs{X}}
    \sum_{x\in X}\rho_{\mathcal L^\#}(x)^{-1}=
    \frac{\abs{\Lambda_X}}{\abs{X}}
    -\sum_{\lambda\in \Lambda_X}\frac{\abs{\set{x\in \mathcal L^\#\setminus X\mid\lambda\in \sigma_x^\#}}}{\abs{\set{z\in \mathcal L^\#\mid\lambda\in \sigma_z^\#}}}
    .
    \label{rhomax_avg_loc}
  \end{equation}
  Now, taking a limit in which $X\to\mathcal L^\#$ such that $\Lambda_X\Uparrow\Lambda_{\infty}$ in the sense of Van Hove (there are many senses in which this limit can be taken; see, for instance, Definition \ref{def:distance-between-configurations}), we find that
  \begin{equation}
    \lim_{X\to\mathcal L^\#}\frac1{\abs{X}}
    \sum_{x\in X}\rho_{\mathcal L^\#}(x)^{-1}=
    \rho_{\mathrm{max}}^{-1}
    .
  \end{equation}
  However, by Item \ref{asm:density_max_local} of Assumption \ref{assumption}, $\rho_{\mathrm{max}}=\rho_{\mathrm{max}}^{\mathrm{loc}}$, so, for all $x\in X$ other than a set of size $o(\abs{X})$, $\rho_{\mathcal L^\#}(x)=\rho_{\mathrm{max}}$.
  Finally, by Lemma \ref{lemma:periodic_ground_state}, $\mathcal L^\#$ is periodic, if one $x\in X$ has a larger density, then a fraction of order $\abs{X}$ will have a larger density as well.
  This proves \eqref{rho_max_cst}.

  Having done this, \eqref{sumv_rhomax} then follows from \eqref{sumv_rholoc}.
\end{proof}

\begin{lemma}\label{lemma:V_local}
  Given a configuration $X\in\Omega$, if $x$ is $\#$-correct, then
  \begin{equation}
    V_{X}(\sigma_{x})=\sigma_{x}^{\#}
    .
  \end{equation}
  In particular, if $x$ is $\#$-correct, then $V_{X}(\sigma_x)$ has a radius of at most $r_{\mathrm{eff}}$.
\end{lemma}

\begin{proof}
  Since $x$ is $\#$-correct, its neighbors are all in $\mathcal{L}^{\#}$, but the discrete Voronoi cell $V_{X}(\sigma_{x})$ only depends on $x$ and its neighbors; see Definition \ref{def:voronoi}.
\end{proof}

\begin{lemma}\label{lem:pointwiselowerbound}
$\mu:=\min_{\lambda\in\sigma^{\#}_{x}} v^{\#}(\lambda)>0$.
\end{lemma}

\begin{proof}
This is a direct consequence of Lemma \ref{lem:finiteeffectiveparticle}: the maximum number of particles whose support (intersected with $\Lambda_{\infty}$) is at a given distance from a fixed site is finite.
\end{proof}

Finally, we prove that Item \ref{asm:imperfect_transition} of Assumption \ref{assumption} implies a coarse-grained version of the assumption.
In Definition \ref{def:gfc} below, this will ensure that the labeling function on a GFc is well-defined.

\begin{lemma}\label{lemma:unique_Rcorrect}
  If $\mathcal R_2>0$, then given a configuration $X\in\Omega(\Lambda)$ and $x,y\in X$ such that $\mathrm{d}_{\Lambda_{\infty}}(x,y)<\mathcal R_2$, if $x$ is $(\#,\mathcal R_2)$-correct and $y$ is $(\#',\mathcal R_2)$-correct, then $\#=\#'$.
\end{lemma}

\begin{proof}
  Since $y$ is $(\#',\mathcal R_2)$-correct, every particle inside the ball of radius $\mathcal R_2$ centered at $y$ is $\#'$-correct.
  Therefore, if $\mathrm{d}_{\Lambda_{\infty}}(x,y)<\mathcal R_2$, then $x$ is $\#'$-correct.
  Finally, since $x$ is $\#$-correct, by Item\-~\ref{asm:imperfect_transition} of Assumption \ref{assumption}, $\#=\#'$.
\end{proof}

\subsection{Gaunt-Fisher configurations}
\label{subsec:gfc}

We are now ready to construct GFc's.

\begin{definition}[Gaunt-Fisher configuration associated to a particle configuration]\label{def:gfc}
Consider a configuration $X\in\Omega^{\#}(\Lambda)$ for some $\#$. 
We associate a GFc to each connected component of the union of the Voronoi cells of the $\mathcal R_2$-incorrect particles in $X$ (recall Definition \ref{def:R_correct}), where $\mathcal R_2$ satisfies \eqref{ineq_R2_1}, \eqref{ineq_R2_2}, and \eqref{ineq_R2_3}:
\begin{equation}
  \bigcup_{x\in\mathcal{I}^{(\mathcal R_2)}(X)}V_{X}(\sigma_{x})
  =:\bigcup_{i}\bar\gamma_{i},
\end{equation}
where $\bar\gamma_{i}$ and $\bar\gamma_{j}$ are disconnected for $i\neq j$.
Each such GFc has a \emph{support}, an \emph{internal configuration}, and a \emph{labeling function}.

\begin{enumerate}
\item The support of a GFc is the set $\bar\gamma_{i}$ itself.

\item The internal configuration is $X_{\gamma_{i}}:= X\cap\bar{\gamma}_{i}$.
\item Let $\mathrm{ext}(\gamma_{i}),\mathrm{int}_{1}(\gamma_{i}),\dots,\mathrm{int}_{N}(\gamma_{i})$ be the connected components of $\bar{\gamma}_{i}^{c}:=\Lambda_{\infty}\setminus\bar{\gamma}_{i}$ with $\mathrm{ext}(\gamma_{i})$ being the unique unbounded component (so that $\mathrm{int}_{1}(\gamma_{i}),\cdots \mathrm{int}_{N}(\gamma_{i})$ are the \emph{holes} in $\bar\gamma_{i}$).
Now, given $A_{j}\in\set{\mathrm{int}_{1}(\gamma_{i}),\cdots,\mathrm{int}_{N}(\gamma_{i})}$, $A_j$ is simply connected, so, by Item \ref{asm:lattice} of Assumption \ref{assumption}, the interior boundary of $A_j$ is $\mathcal R_0$-connected (see Definition \ref{def:rconnected}).
Thus, provided that
\begin{equation}
  \mathcal R_2\ge \mathcal R_0,
  \label{ineq_R2_1}
\end{equation}
all the particles in $A_{j}$ whose Voronoi cell intersects $\partial^{\textrm{in}}A_{j}$ are $(\#_j,\mathcal R_2)$-correct for the same $\#_{j}\in\mathcal{G}$.
Similarly, if $A_j=\mathrm{ext}(\gamma_i)$, then we consider the exterior boundary of $A_j^c$, which is simply connected, and apply the same reasoning to find that the interior boundary of $A_j$ is lined with $(\#_j,\mathcal R_2)$-correct particles.
The labeling function $\mu_{\gamma_{i}}:\set{\mathrm{ext}(\gamma_{i}),\mathrm{int}_{1}(\gamma_{i}),\cdots,\mathrm{int}_{N}(\gamma_{i})}\rightarrow\mathcal{G}$ assigns this ground state $\#_{j}$ to $A_{j}$.
\end{enumerate}
\end{definition}

See Figure \ref{fig:GFc} for an example.

\begin{figure}
  \hfil\includegraphics[width=14cm]{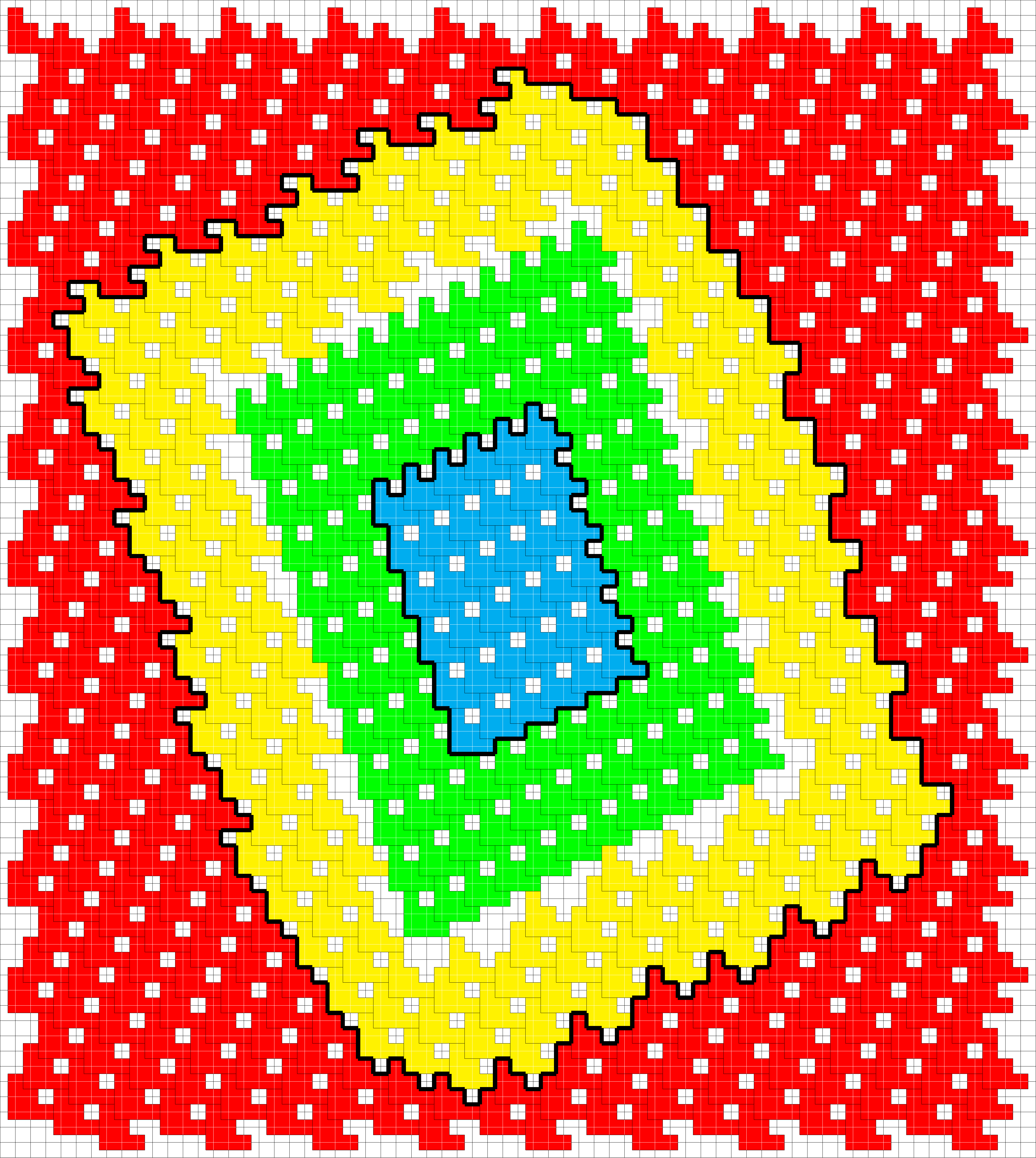}

  \caption{
    An example of a GFc associated to a particle configuration for the 3-staircase model; see Section \ref{subsec:n_staircases}.
    The red and yellow particles (the two outermost rings) are in one ground state and the green and cyan particles (the innermost) are in another ground state.
    For this model, $\mathcal R_2=3$.
    The green and yellow particles are $\mathcal R_2$-incorrect, and thus are in the GFc.
    Taking the union of the supports of their Voronoi cells, we find the support of the GFc, which is delineated by a thick line.
    \label{fig:GFc}
  }
\end{figure}

Hence, each GFc is a triplet $\gamma:=(\bar\gamma,X_{\gamma},\mu_{\gamma})$, but these are not arbitrary: there is a compatibility condition that ensures that a collection of GFc's corresponds to a particle configuration.
For one, the labels $\mu_{\gamma}$ must be compatible: if a GFc lies inside another, their labels must match up.
In addition, the hard-core repulsion imposes an a priori constraint on the $X_{\gamma}$, but we will now dispense with it by assuming that $\mathcal R_2$ is large enough, which prevents the $X_{\gamma}$ from interacting with each other.

\begin{lemma}
\label{lem:R_large_valid_configuration}
  Suppose that
  \begin{equation}
    \mathcal R_2>\max\set{\mathrm{d}_{\Lambda_{\infty}}(x,y)\mid x,y\in\Lambda,\ \omega_{x}\cap\omega_{y}\neq\emptyset}.
    \label{ineq_R2_2}
  \end{equation}
  Consider a configuration $X\in\Omega^{\#}(\Lambda)$ and its corresponding family of GFc's $\set{\gamma_{1},\cdots,\gamma_n}$.
  For each GFc $\gamma_{i}$, we construct a configuration $\xi_{\gamma_{i}}$ in which the holes of $\gamma_{i}$ are filled:
  \begin{equation}
  \label{eqn:canonical_configuration}
    \xi_{\gamma_{i}}:=X_{\gamma_{i}}\cup \left[\mathrm{ext}(\gamma_{i})\cap\mathcal{L}^{\mu_{\gamma_{i}}(\mathrm{ext}(\gamma_{i}))}\right]\cup\bigcup_{j}\left[\mathrm{int}_{j}(\gamma_{i})\cap\mathcal{L}^{\mu_{\gamma_{i}}(\mathrm{int}_{j}(\gamma_{i}))}\right]
    .
  \end{equation}
  This configuration is a valid particle configuration: $\xi_{\gamma_{i}}\in\Omega(\Lambda_{\infty})$.
\end{lemma}

\begin{proof}
Let $A_{j}\in\set{\mathrm{ext}(\gamma_{i}),\mathrm{int}_{1}(\gamma_{i}),\cdots,\mathrm{int}_{N}(\gamma_{i})}$ and consider the set $X\cap A_{j}$ of particles in $A_{j}$.
Let us remove the particles that are not in $\mathcal{L}^{\mu_{\gamma_{i}}(A_{j})}$, which will leave gaps, which we then fill with particles in $\mathcal{L}^{\mu_{\gamma_{i}}(A_{j})}$.
Proceeding in this way, we construct $\xi_{\gamma_{i}}$.
The proof then reduces to showing that the newly added particles do not overlap with any of the existing ones.
Obviously, particles in $\mathcal{L}^{\mu_{\gamma_{i}}(A_{j})}$ can only overlap with particles in $\mathcal{L}^{\#'}$ with $\#'\neq\mu_{\gamma_{i}}(A_{j})$.
But such particles in $\mathcal{L}^{\#'}$ must be a distance of at least $\mathcal R_2$ away (this follows straightforwardly from the fact that different GFc's are disconnected, so the extra particles are surrounded by $(\mu_{\gamma_{i}}(A_{j}),\mathcal R_2)$-correct particles, which implies that the nearest particles that are not in $\mathcal{L}^{\mu_{\gamma_{i}}(A_{j})}$ are at least a distance $\mathcal R_2$ away).
Therefore, as long as $\mathcal R_2> \max\set{\mathrm{d}_{\Lambda_{\infty}}(x,y)\mid x,y\in\Lambda,\ \omega_{x}\cap\omega_{y}\neq\emptyset}$, there can be no interaction.
\end{proof}

The converse to Lemma \ref{lem:R_large_valid_configuration} follows. 

\begin{proposition}
\label{prop:canonical_configuration_gfc_correspondence}
For each GFc $\gamma$, the canonical configuration $\xi_{\gamma}$ has $\gamma$ as its unique GFc.
\end{proposition}

The proof of this fact is simple, though writing it out is a touch tedious, so we postpone it to Appendix \ref{appx:canonical_configuration_gfc_correspondence}.

Let us now define the set of configurations of GFc's.

\begin{definition}
We define the set $\mathcal{C}^{\#}(\Lambda)$ of GFc's as the set of triplets $(\bar\gamma,X_{\gamma},\mu_{\gamma})$, where $\bar\gamma$ is a connected subset of $\Lambda$, $X_{\gamma}\in\Omega(\bar\gamma)$, and $\mu_{\gamma}$ is a map from $\set{\mathrm{ext}(\gamma),\mathrm{int}_{1}(\gamma),\cdots,\mathrm{int}_{N}(\gamma)}\to\mathcal{G}$, in such a way that $\xi_{\gamma}\in\Omega^{\#}(\Lambda)$.
\end{definition}

Note that $\#$ is the exterior close-packing index of the GFc {\it regardless of whether the GFc is surrounded by another one}.
This does not cause problems, due to our using the Minlos-Sinai trick in Section \ref{sec:Pirogov_Sinai}: we only ever need to study GFc's that have $\#$ as their external index; see Section \ref{subsec:GFc_model}.

\subsection{Relations to species-preserving isometries}

We briefly comment on the relation between GFc's and species-preserving isometries.
The basic observation is that the latter preserve not only the continuous support $\omega_{x}$ but also the discrete support $\sigma_{x}$, in the sense that $\psi(\sigma_{x})=\sigma_{\psi(x)}$ for any species-preserving isometry $\psi$.
Since the Voronoi cells are constructed in a choice-free manner, a multitude of convenient properties follow without effort:
\begin{enumerate}
\item preservation of Voronoi cells: given $X\in\Omega(\Lambda_{\infty})$ containing $x$, $\psi(V_{X}(\sigma_{x}))=V_{\psi(X)}(\sigma_{\psi(x)})$;
\item preservation of $\#$- and $(\#,\mathcal R_2)$-correctness: $x\in X$ is $\#$-correct (resp. $(\#,\mathcal R_2)$-correct) if and only if $\psi(x)\in\psi(X)$ is $\psi(\#)$-correct (resp. $(\psi(\#),\mathcal R_2)$-correct);
\item preservation of reference Voronoi cells: if $x\in\mathcal{L}^{\#}$, then $\psi(\sigma^{\#}_{x})=\sigma^{\psi(\#)}_{\psi(x)}$;
\item action on GFc's: $\psi$ acts on the set of all GFc's by
\begin{equation}\label{eqn:action_on_gfc}
\psi\cdot\gamma
=\psi\cdot(\bar{\gamma},X_{\gamma},\mu_{\gamma})
:=(\psi(\bar{\gamma}),\psi(X_{\gamma}),\psi\circ\mu_{\gamma}).
\end{equation}
\end{enumerate}
This will be useful in the proof of Proposition \ref{prop:central_estimates} below.

\section{Peierls condition}\label{sec:peierls}

We now formulate the Peierls condition in the context of HCLP systems. 
Since the GFc's are defined from $\mathcal R_2$-incorrect particles, we can prove that they come at a large \emph{cost} by showing that \emph{somewhere near} an $\mathcal R_2$-incorrect particle (in a sense made precise below), the local density must be bounded away from $\rho_{\max}=\rho_{\max}^{\loc}$.
When the fugacity is sufficiently large, the probability of having such a dip in the local density is low, which in turn induces a cost for each GFc that is proportional to its volume.

\subsection{Cost of an $\mathcal R_2$-incorrect particle}

Thus, we must prove that any $\mathcal R_2$-incorrect particle induces a dip in the local density.
This is a similar statement to Item \ref{asm:density_local_density} of Assumption \ref{assumption}, except for the fact that the dip in the density must hold for $\mathcal R_2$-incorrect particles, instead of $\mathcal R_1$-incorrect ones, whenever $\mathcal R_2\ge \mathcal R_1$.

\begin{lemma}\label{lem:cost}
  For any $\mathcal R_2$ such that
  \begin{equation}
    \mathcal R_2\ge\mathcal R_1,
    \label{ineq_R2_3}
  \end{equation}
  we define
  \begin{equation}
    \mathcal S_0:=\mathcal S_1+\mathcal R_2+4 r_{\mathrm{eff}}
    \label{calS0}
  \end{equation}
  (recall that $r_{\mathrm{eff}}$ was defined in Lemma \ref{lem:finiteeffectiveparticle}).
  We have that, for all $X\in\Omega(\Lambda)$ and $x\in X$, if $x$ is $\mathcal R_2$-incorrect, then there exists $y\in X$ such that $\mathrm d_{\Lambda_{\infty}}(x,y)\le \mathcal S_0$ and
  \begin{equation}
    \rho_{X}^{-1}(y)\ge\rho_{\mathrm{max}}^{-1}+\epsilon,
    \label{ineqrho}
  \end{equation}
  where $\mathcal R_1$, $\mathcal S_1$ and $\epsilon$ are the quantities appearing in Item \ref{asm:density_local_density} of Assumption \ref{assumption}.
\end{lemma}

\begin{proof}
  Consider $x\in X$ that is $\mathcal R_2$-incorrect.
  We first prove that there exists $z\in \mathcal N_{X}^{(\mathcal R_2)}(x)$ (see Definition \ref{def:R_correct}) such that $z$ is incorrect.
  We proceed by contradiction: suppose that for every $z\in\mathcal N_{X}^{(\mathcal R_2)}(x)$ there exists a $\#_z\in\mathcal G$ such that $z$ is $\#_z$-correct.
  A priori, $\#_z$ may depend on $z$, but we will now show that $\#_z$ has to be the same for all $z\in\mathcal N_{X}^{(\mathcal R_2)}$.
  Indeed, let us assume that $z\neq x$, and consider the shortest path through $\Lambda_{\infty}$ that goes from $z$ to $x$.
  By the definition of $\mathcal N_{X}^{(\mathcal R_2)}(x)$, all of the Voronoi cells that this path goes through belong to particles in $\mathcal N_{X}^{(\mathcal R_2)}(x)$.
  In particular, this path goes through at least one neighbor $z'\in\mathcal N_{X}(z)$ satisfying $\mathrm d_{\Lambda_{\infty}}(z',x)<\mathrm d_{\Lambda_{\infty}}(z,x)$.
  Now, by Item \ref{asm:imperfect_transition} of Assumption \ref{assumption},
  \begin{equation}
    \#_{z'}=\#_z
    .
  \end{equation}
  By induction, this shows that $\#_z=\#_x$.
  And, because $z$ was chosen arbitrarily, this implies that, for all $z\in\mathcal N_{X}^{(\mathcal R_2)}(x)$, $\#_z=\#_x$, and, therefore, that $x$ is $(\#_x,\mathcal R_2)$-correct, a contradiction.

  Thus, there exists $z\in\mathcal N_{X}^{(\mathcal R_2)}(x)$ such that $z$ is incorrect.
  We choose the incorrect $z$ such that $\mathrm d_{\Lambda_{\infty}}(V_{X}(z),V_{X}(x))$ is minimal.
  Since $z$ is incorrect, it is also $\mathcal R_1$-incorrect.
  By Item \ref{asm:density_local_density} of Assumption \ref{assumption}, there exists $y$ such that $\mathrm d_{\Lambda_{\infty}}(z,y)\le\mathcal S_1$ and $\rho_{X}^{-1}(y)\ge\rho_{\mathrm{max}}^{-1}+\epsilon$.

  Finally, we prove that $\mathrm d_{\mathrm{\Lambda_{\infty}}}(x,z)\le\mathcal R_2+4r_{\mathrm{eff}}$.
  Indeed, if $z=x$, then this is obvious.
  Otherwise, consider the shortest path from $x$ to $z$.
  If this path crosses the Voronoi cell of a particle $z'$, then $z'\in\mathcal N_{X}^{(\mathcal R_2)}(x)$.
  Among the possible choices of $z'$, we choose the one that is closest to the particle at $z$, which must therefore neighbor the particle at $z$.
  In addition, since $z$ minimizes the distance between $V_{X}(z)$ and $V_{X}(x)$ among incorrect particles, $z'$ must be $\#_{z'}$-correct for some $\#_{z'}\in\mathcal G$.
  Therefore, $z$ is the neighbor of $z'$ in $\mathcal L^{\#_z'}$, and, by Lemma \ref{lemma:V_local},
  \begin{equation}
    \mathrm d_{\Lambda_{\infty}}(z,z')\le 2r_{\mathrm{eff}}
    .
  \end{equation}
  In addition, since $z\neq x$, $x$ is also $\#_{z'}$-correct, so
  \begin{equation}
    \mathrm d_{\Lambda_{\infty}}(x,z')\le \mathcal R_2+2r_{\mathrm{eff}}
  \end{equation}
  and thus
  \begin{equation}
    \mathrm d_{\Lambda_{\infty}}(x,z)\le\mathcal R_2+4r_{\mathrm{eff}}
    .
  \end{equation}
\end{proof}

\subsection{Peierls condition}

To prove the Peierls condition from Lemma \ref{lem:cost}, we first need to define the \emph{effective volume} of a GFc.

\begin{definition}[effective volume]
\label{def:effective_volume}
Define the effective volume of a GFc $\gamma=(\bar{\gamma},X_{\gamma},\mu_{\gamma})$ by starting with $\bar\gamma$ and removing the weights due to the particles that are outside the GFc but whose Voronoi cell intersects its support (see Definitions \ref{def:R_correct} and \ref{def:effective_particle} for notation):
\begin{equation}\label{eqn:effective_volume}
\norm{\bar{\gamma}}:=\sum_{\lambda\in\bar{\gamma}}\left[1-\sum_{\#}\sum_{x\in \mathcal{C}_{\#}^{(\mathcal R_2)}(X):\ \sigma^{\#}_{x}\ni\lambda}v^{\#}(\lambda)\right],
\end{equation}
where $X$ is any configuration of which $\gamma$ is a GFc.
We note that it is not difficult to verify that \eqref{eqn:effective_volume} is independent of the choice of $X$.
\end{definition}

For an example, see the GFc in Figure \ref{fig:GFc}, in which the sites on $\bar\gamma$ that contribute less than 1 are the uncovered sites that neighbor the thick black line.

The following simple bound on the effective volume will be useful.

\begin{lemma}\label{lem:effective_volume_naive_bound}
For any GFc $\gamma$, $\norm{\bar{\gamma}}\ge\mu\abs{\bar{\gamma}}$, where $\mu$ is as in Lemma \ref{lem:pointwiselowerbound}.
\end{lemma}

\begin{proof}
Recall the definition of the configuration $\xi_{\gamma}$ in \eqref{eqn:canonical_configuration}.
It suffices to show that the summand in 
\begin{equation}
\norm{\bar{\gamma}}=\sum_{\lambda\in\bar{\gamma}}\left[1-\sum_{\#}\sum_{x\in \mathcal{C}_{\#}^{(\mathcal R_2)}(\xi_{\gamma}):\ \sigma^{\#}_{x}\ni\lambda}v^{\#}(\lambda)\right],
\end{equation} 
is bounded below by $\mu$ for each $\lambda\in\bar{\gamma}$. 
If the double sum does not vanish, then exactly one ground state $\#$ contributes by Item \ref{asm:imperfect_transition} of Assumption \ref{assumption}.
In this case, the summand reduces to 
\begin{equation}
1-\sum_{x\in \mathcal{C}_{\#}^{(\mathcal R_2)}(\xi_{\gamma}):\ \sigma^{\#}_{x}\ni\lambda}v^{\#}(\lambda)
=1-\frac{\abs{\set{x\in \mathcal{C}_{\#}^{(\mathcal R_2)}(\xi_{\gamma})\mid\lambda\in \sigma_{x}^{\#}}}}{\abs{\set{z\in \mathcal{L}^{\#}\mid\lambda\in\sigma^{\#}_{z}}}}.
\end{equation}
Here, $\lambda$ is in the support of the reference Voronoi cells associated to a particle outside the GFc.
However, since $\lambda\in\bar\gamma$, $\lambda$ is also in the Voronoi cell of some particle in $\mathcal{I}^{(\mathcal R_2)}(\xi_{\gamma})$ located on the boundary of the GFc.
By Definition \ref{def:R_correct}, this particle is also in $\mathcal L^\#$.
Therefore,
\begin{equation}
  \abs{\set{z\in \mathcal{L}^{\#}\mid\lambda\in\sigma^{\#}_{z}}}\ge
  1+\abs{\set{x\in \mathcal{C}_{\#}^{(\mathcal R_2)}(\xi_{\gamma})\mid\lambda\in \sigma_{x}^{\#}}}
\end{equation}
and so
\begin{equation}
1-\frac{\abs{\set{x\in \mathcal{C}_{\#}^{(\mathcal R_2)}(\xi_{\gamma})\mid\lambda\in \sigma_{x}^{\#}}}}{\abs{\set{z\in \mathcal{L}^{\#}\mid\lambda\in\sigma^{\#}_{z}}}}
\ge\frac{1}{\abs{\set{z\in \mathcal{L}^{\#}\mid\lambda\in\sigma^{\#}_{z}}}}
\ge\mu.
\end{equation}
\end{proof}

We are now ready to state the Peierls condition.

\begin{lemma}[Peierls condition]\label{lem:Peierls}
The Peierls condition is satisfied: defining
\begin{equation}
  \rho_{0}:=\frac1{\rho_{\mathrm{max}}^{-1}+\frac\epsilon{\mathcal N}}
  \label{rho0}
\end{equation}
where $\epsilon$ appears in Item \ref{asm:density_local_density} of Assumption \ref{assumption} and $\mathcal N$ is the number of particles inside the ball of radius $\mathcal S_0$ (see Lemma \ref{lem:cost}):
\begin{equation}
  \mathcal N:=\max_{X\in\Omega(\Lambda_{\infty})}\abs{
  \set{x\in X\mid\mathrm d_{\Lambda_{\infty}}(x,\mathbf 0)\le\mathcal S_0}
  },
  \label{calN}
\end{equation}
we have, for any GFc $\gamma$,
\begin{equation}\label{eqn:Peierls_condition}
\abs{X_{\gamma}}\le\rho_{0}\norm{\bar{\gamma}}.
\end{equation}
\end{lemma}

\begin{proof}
Given a GFc $\gamma$, consider the associated configuration $\xi_{\gamma}$ defined in \eqref{eqn:canonical_configuration}.
By Proposition \ref{prop:canonical_configuration_gfc_correspondence},$X_{\gamma}\subset \xi_{\gamma}$ coincides with the set $\mathcal{I}^{(\mathcal R_2)}(\xi_{\gamma})$ of $\mathcal R_2$-incorrect particles in $\xi_{\gamma}$. 
Hence, by \eqref{eqn:local_density},
\begin{equation}
\norm{\bar{\gamma}}=\sum_{x\in X_{\gamma}}\rho_{\xi_{\gamma}}(x)^{-1}.
\end{equation}
By Lemma \ref{lem:cost}, for every $x\in X_\gamma$, there exists a ball of radius $\mathcal S_0$ where at least one of the particles satisfies \eqref{ineqrho}.
In addition, $X_\gamma$ will contain at least $\abs{X_\gamma}/\mathcal N$ such balls.
Therefore,
\begin{equation}
\norm{\bar{\gamma}}\ge\frac{\abs{X_{\gamma}}}{\mathcal{N}}(\rho_{\max}^{-1}+\epsilon)+\left(\abs{X_{\gamma}}-\frac{\abs{X_{\gamma}}}{\mathcal{N}}\right)\rho_{\max}^{-1}=\left(\rho_{\max}^{-1}+\frac{\epsilon}{\mathcal{N}}\right)\abs{X_{\gamma}}.
\end{equation}
Therefore, the Peierls condition is satisfied with 
\begin{equation}
\rho_{0}:=\frac{1}{\rho_{\max}^{-1}+\frac{\epsilon}{\mathcal{N}}}<\rho_{\max}.\qedhere
\end{equation}
\end{proof}

\subsection{Proof of Lemma \ref{lem:equiv_assum}}\label{sec:equiv_assumption}
To close this section, let us prove Lemma \ref{lem:equiv_assum}, which provides an alternative to Item \ref{asm:density_local_density} of Assumption \ref{assumption}.
First, we introduce a few definitions.

\begin{definition}
\label{def:distance-between-configurations}
Denote by $\triangle$ the symmetric difference of two sets.
Define a metric $d$ on the space $\Omega(\Lambda_{\infty})$ of all particle configurations by
\begin{equation}
d(X,Y):=\sum_{\lambda\in X\triangle Y}2^{-\mathrm{d}_{\Lambda_{\infty}}(\lambda,\vec{0})}.
\end{equation}
\end{definition}

Our choice of the metric $d$ is fairly standard: under this metric, \emph{two configurations are close to each other if they coincide on a large region containing the origin}  \cite[\S6.4.1]{Friedli2017}.
Formally, we characterize convergence under this metric in the following manner, which is straightforward to verify:

\begin{lemma}\label{lem:pointwiseconvergence}
Given $X\in\Omega(\Lambda_{\infty})$ and a sequence $\set{X_{n}}\subset\Omega(\Lambda_{\infty})$, the following are equivalent:
\begin{enumerate}
\item $d(X_{n},X)\rightarrow 0$ (we also write $X_{n}\rightarrow X$);
\item $\indicator{X_{n}}(\lambda)\rightarrow\indicator{X}(\lambda)$ for all $\lambda\in\Lambda_{\infty}$;
\item $\indicator{X_{n}}\rightarrow\indicator{X}$ on all finite $\Lambda\Subset\Lambda_{\infty}$.
\end{enumerate}
In addition,
$\Omega(\Lambda_{\infty})$ is sequentially compact under the metric topology induced by $d$.
\end{lemma}

We can now prove Lemma \ref{lem:equiv_assum}.

\begin{proof}[Proof of Lemma \ref{lem:equiv_assum}]
First of all, let us prove that Item \ref{asm:density_local_density} implies \eqref{Ggm}.
Given $\#\in\mathcal G$, by Lemma \ref{lemma:rholoc_cst}, $\rho_{\mathcal L^\#}(x)=\rho_{\mathrm{max}}$, so $\mathcal L^\#\in g_m$.
Now, given $X$ in $g_m$, we prove that all particles are $(\#,\mathcal R_1)$-correct for the same $\#$.
By Item \ref{asm:density_local_density}, no particle can be $\mathcal R_1$-incorrect, otherwise $X\not\in g_m$.
Therefore, every $x\in X$ is $(\#_x,\mathcal R_1)$-correct.
However, by Item \ref{asm:imperfect_transition}, all $\#_x$ must be the same, and therefore, all particles are $(\#,\mathcal R_1)$-correct for some $\#$.
In particular, every particle is in $\mathcal L^\#$, and, since the local density is maximal, $X=\mathcal L^\#$.

We now turn to the more difficult direction, namely that \eqref{Ggm} implies Item \ref{asm:density_local_density}.

We proceed by contradiction, assuming that Item \ref{asm:density_local_density} does not hold.
In other words, there exists a sequence $\set{X_{n}}\subset\Omega(\Lambda_{\infty})$ such that, for each $n$, there exists an $\mathcal{R}_{0}$-incorrect $x_{n}\in X_{n}$ for which all $y\in X_{n}$ with $\mathrm{d}_{\Lambda_{\infty}}(y,x_{n})\le n$ are such that 
\begin{equation}\label{eqn:inverselocaldensitybound}
\rho_{X_{n}}(y)^{-1}\le\rho_{\max}^{\loc\ -1}+\frac{1}{n}.
\end{equation}
By translation invariance, we can assume that $x_{n}=\vec{0}$ for all $n$. 
By Lemma \ref{lem:pointwiseconvergence}, there exists a subsequence $\set{X_{n_{k}}}$ and a configuration $X\in\Omega(\Lambda_{\infty})$ such that $X_{n_{k}}\rightarrow X$.
\medskip

Let us first prove that $X\in g_{m}$.
Let $x\in X$.
For all sufficiently large $k$, we have that $x\in X_{n_{k}}$ by Lemma \ref{lem:pointwiseconvergence}, and
\begin{equation}\label{eqn:inverselocaldensityboundspecialized}
\rho_{X_{n_{k}}}(x)^{-1}\le\rho_{\max}^{\loc\ -1}+\frac{1}{n_{k}}
\end{equation}
by setting $(y,x_{n_{k}})=(x,\vec{0})$ in \eqref{eqn:inverselocaldensitybound}. Decomposing 
\begin{equation}
\Lambda_{\infty}=\bigsqcup_{m=0}^{\infty} S_{m},\text{ where }S_{m}:=\set{\lambda\in\Lambda_{\infty}\mid\mathrm{d}_{\Lambda_{\infty}}(\lambda,\vec{0})=m},
\end{equation}
we have, by Fatou's lemma,
\begin{equation}\label{eqn:fatou}
\sum_{m=0}^{\infty}\left(\liminf_{k\rightarrow\infty}\sum_{\lambda\in S_{m}}\frac{\indicator{V_{X_{n_{k}}}(\sigma_{x})}(\lambda)}{\abs{\set{z\in X_{n_{k}}\mid\lambda\in V_{X_{n_{k}}}(\sigma_{z})}}}\right)\le\liminf_{k\rightarrow\infty}\left(\sum_{m=0}^{\infty}\sum_{\lambda\in S_{m}}\frac{\indicator{V_{X_{n_{k}}}(\sigma_{x})}(\lambda)}{\abs{\set{z\in X_{n_{k}}\mid\lambda\in V_{X_{n_{k}}}(\sigma_{z})}}}\right).
\end{equation}
By \eqref{eqn:inverselocaldensityboundspecialized}, we have that
\begin{equation}\label{eqn:rhs}
\RHS=\liminf_{k\rightarrow\infty}\rho_{X_{n_{k}}}(x)^{-1}\le\rho_{\max}^{\loc\ -1}.
\end{equation}
To study the LHS, notice first that, for a fixed $\lambda\in S_{m}$, the value of the indicator function is determined locally, that is, by the restriction of the configuration $X_{n_{k}}$ to the finite region 
\begin{equation}
\Lambda_{1}:=\set{\mu\in\Lambda_{\infty}\mid\mathrm{d}_{\Lambda_{\infty}}(\lambda,\sigma_{\mu})<\mathrm{d}_{\Lambda_{\infty}}(\lambda,\sigma_{x})}.
\end{equation}
Hence, by Lemma \ref{lem:pointwiseconvergence}, the value of the indicator function converges to $\indicator{V_{X}(\sigma_{x})}(\lambda)$ as $k\rightarrow\infty$. 
If $\lambda\in V_{X}(\sigma_{x})$, notice again that the denominator is determined locally, this time by the restriction of $X_{n_{k}}$ to the finite region 
\begin{equation}
\Lambda_{2}:=\set{\mu\in\Lambda_{\infty}\mid\mathrm{d}_{\Lambda_{\infty}}(\lambda,\sigma_{\mu})=\mathrm{d}_{\Lambda_{\infty}}(\lambda,\sigma_{x})}.
\end{equation}
Thus, by Lemma \ref{lem:pointwiseconvergence}, the denominator converges to $\abs{\set{z\in X\mid \lambda\in V_{X}(\sigma_{z})}}$ as $k\rightarrow\infty$. 
Since $S_{m}$ is finite for each $m$, 
\begin{equation}\label{eqn:lhs}
\LHS=\sum_{m=0}^{\infty}\sum_{\lambda\in S_{m}}\frac{\indicator{V_{X}(\sigma_{x})}(\lambda)}{\abs{\set{z\in X\mid\lambda\in V_{X}(\sigma_{z})}}}=\rho_{X}(x)^{-1}.
\end{equation}
Combining \eqref{eqn:fatou}, \eqref{eqn:rhs}, and \eqref{eqn:lhs} with the obvious bound $\rho_{X}(x)^{-1}\ge\rho_{\max}^{\loc\ -1}$, we deduce that
\begin{equation}
  \rho_{X}(x)=\rho_{\max}^{\loc},
\end{equation}
and, therefore, that $X\in g_{m}$.
\medskip

Since $\vec{0}\in X_{n_{k}}$ for all $k$, we have $\vec{0}\in X$ by Lemma \ref{lem:pointwiseconvergence}. 
By \eqref{Ggm}, $X$ is a ground state, so $\vec{0}\in X$ is $(X,\mathcal{R}_{0})$-correct.
Since $(X,\mathcal{R}_{0})$-correctness is a local property, we conclude from Lemma \ref{lem:pointwiseconvergence} that $\vec{0}\in X_{n_{k}}$ is $(X,\mathcal{R}_{0})$-correct for large enough $k$, a contradiction.
\end{proof}

\section{High-fugacity expansion}\label{sec:Pirogov_Sinai}

In this section, we prove Theorems \ref{thm:analyticity} and \ref{thm:crystallization}.
Our argument consists of two main steps: converting the particle model into a GFc model, and subsequently evaluating the partition functions using cluster expansion techniques within the framework of Pirogov-Sinai theory.

\subsection{Passing to a GFc model}\label{subsec:GFc_model}

We begin by restating the boundary condition (Definition \ref{def:boundary_condition_correctness}) in terms of GFc's. 
That the two formulations of boundary conditions are equivalent is clear from the construction.

\begin{definition}[boundary condition]
Given a ground state $\#$ and a finite region $\Lambda\Subset\Lambda_{\infty}$, define the set of configurations in $\Lambda$ subject to the boundary condition $\#$ as 
\begin{equation}
\begin{multlined}
\Omega^{\#}(\Lambda):=\{X\in\Omega(\Lambda_{\infty})\mid X\setminus\Lambda=\mathcal{L}^{\#}\setminus\Lambda,\text{ all GFc's associated to $X$ are contained in $\Lambda$}
\\
\text{and at distance $\mathrm{d}_{\Lambda_{\infty}}\ge1$ from $\partial^{\text{in}}\Lambda$}\}.
\end{multlined}
\end{equation}
\end{definition}

We will convert our particle model into a GFc model by rewriting the partition functions in terms of certain \emph{weights} of the GFc's.
During the process, we will make use of the following standard notions in Pirogov-Sinai-type arguments:
\begin{enumerate}
\item for each ground state $\#$, define $\mathcal{C}^{\#}:=\cup_{\Lambda\Subset\Lambda_{\infty}}\mathcal{C}^{\#}(\Lambda)$ as the set of all GFc's \emph{of type $\#$};
\item two GFc's $\gamma_{1},\gamma_{2}$ of the same type are said to be \emph{compatible} if and only if $\mathrm{d}_{\Lambda_{\infty}}(\bar{\gamma}_{1},\bar{\gamma}_{2})>1$;
\item define the \emph{support} of a collection $\Gamma$ of GFc's as $\bar{\Gamma}:=\cup_{\gamma\in\Gamma}\bar{\gamma}$, that is, the union of the supports of all the GFc's contained in $\Gamma$;
\item given a GFc $\gamma$, define $\interior_{\#'}\gamma$, called the interior of $\gamma$ of type $\#'$, as the union of the holes in $\bar{\gamma}$ (cf. Definition \ref{def:gfc}) that are assigned the label $\#'$ by $\mu_{\gamma}$.
\end{enumerate}

As we will see, to evaluate the correlation functions as in Theorem \ref{thm:crystallization} requires us to differentiate the partition functions with respect to the fugacity $z$.
To this end, it will be convenient to allow variable fugacities, which we encode in a single \emph{fugacity function} $\mathbf{z}:\Lambda_{\infty}\rightarrow(0,\infty)$.
Accordingly, for each ground state $\#$, we define the (grand canonical) partition function with boundary condition $\#$ and fugacity function $\mathbf{z}$ as
\begin{equation}\label{eqn:partitionfunctiondefinition}
\Xi^{\#}_{\mathbf{z}}(\Lambda):=\sum_{X\in\Omega^{\#}(\Lambda)}\prod_{x\in X\cap\Lambda}\mathbf{z}(x).
\end{equation}
Finally, given a ground state $\#$ and a finite region $\Lambda\Subset\Lambda_{\infty}$, we write $\mathbf{z}^{\#}(\Lambda):=\prod_{x\in \mathcal{L}^{\#}\cap\Lambda}\mathbf{z}(x)$.

\begin{proposition}[GFc model]\label{prop:GFc}
Define the weight $w^{\#}_{\mathbf{z}}(\gamma)$ of a GFc $\gamma\in \mathcal{C}^{\#}$ as
\begin{equation}
\label{eqn:weightdefinition}
w^{\#}_{\mathbf{z}}(\gamma)
:=\frac{\prod_{x\in X_{\gamma}}\mathbf{z}(x)}{\prod_{x\in (\mathcal{L}^{\#}\cap\bar{\gamma})}\mathbf{z}(x)}\prod_{\#'}\frac{\Xi_{\mathbf{z}}^{\#'}(\interior_{\#'}\gamma)}{\Xi^{\#}_{\mathbf{z}}(\interior_{\#'}\gamma)}.
\end{equation} 
Then,
\begin{equation} \label{eqn:gfc_partition_function}
\frac{\Xi^{\#}_{\mathbf{z}}(\Lambda)}{\mathbf{z}^{\#}(\Lambda)}=\sum_{\substack{\Gamma\subseteq \mathcal{C}^{\#}(\Lambda):\\\text{compatible}}}\prod_{\gamma\in\Gamma}w^{\#}_{\mathbf{z}}(\gamma).
\end{equation}
\end{proposition}

\begin{proof}
Following a standard strategy in Pirogov-Sinai theory, we will first map the particle model into a model of external GFc's and then apply a recursion argument.

Let $\gamma_{1},\dots,\gamma_n$ be the GFc's associated to a configuration $X\in\Omega^{\#}(\Lambda)$.
By construction, these GFc's are mutually disconnected: $\mathrm{d}_{\Lambda_{\infty}}(\bar{\gamma}_{i},\bar{\gamma}_{j})>1$ whenever $i\ne j$.
Define the \emph{interior} of a GFc $\gamma_{i}$ as $\interior\gamma_{i}:=\cup_{\#'}\interior_{\#'}\gamma_{i}$.
We say that a GFc $\gamma_{i}$ is \emph{external} (relative to the collection $\set{\gamma_{1},\dots,\gamma_n}$) if $\gamma_{i}$ is not contained in the interior of any other GFc $\gamma_{j}$. 
Notice that any external GFc in the above collection is necessarily of type $\#$.

Conversely, given any compatible collection $\Gamma\subseteq \mathcal{C}^{\#}(\Lambda)$ of external GFc's (relative to $\Gamma$), there exist configurations $X\in\Omega^{\#}(\Lambda)$ from which the GFc's in $\Gamma$ are the only external GFc's that arise.
The general form of such configurations is given by 
\begin{equation}\label{eqn:configurationform}
X=\left(\mathcal{L}^{\#}\setminus\bigcup_{\gamma\in\Gamma}\Int\gamma\right)\cup\bigcup_{\gamma\in\Gamma}\left[X_{\gamma}\cup\bigcup_{\#'}\left(X^{\#'}_{\interior_{\#'}\gamma}\cap\interior_{\#'}\gamma\right)\right],
\end{equation} 
where we use the notation $\Int\gamma:=\bar{\gamma}\cup\bigcup_{\#'}\interior_{\#'}\gamma$.

Therefore, we can write
\begin{equation}
\Xi_{\mathbf{z}}^{\#}(\Lambda)=\sum_{\substack{\Gamma\subseteq \mathcal{C}^{\#}(\Lambda):\\\text{compatible}\\\text{external}}}\left(\prod_{x\in(\mathcal{L}^{\#}\cap\Lambda)\setminus\bigcup_{\gamma\in\Gamma}\Int\gamma}\mathbf{z}(x)\right)\left[\prod_{\gamma\in\Gamma}\left(\prod_{x\in X_{\gamma}}\mathbf{z}(x)\right)\left(\prod_{\#'}\Xi^{\#'}_{\mathbf{z}}(\interior_{\#'}\gamma)\right)\right].
\end{equation}
Dividing both sides by $\mathbf{z}^{\#}(\Lambda)$, we get
\begin{equation}
\frac{\Xi_{\mathbf{z}}^{\#}(\Lambda)}{\mathbf{z}^{\#}(\Lambda)}=\sum_{\substack{\Gamma\subseteq \mathcal{C}^{\#}(\Lambda):\\\text{compatible}\\\text{external}}}\left[\prod_{\gamma\in\Gamma}\left(\frac{\prod_{x\in X_{\gamma}}\mathbf{z}(x)}{\prod_{x\in (\mathcal{L}^{\#}\cap\bar{\gamma})}\mathbf{z}(x)}\right)\left(\prod_{\#'}\frac{\Xi^{\#'}_{\mathbf{z}}(\interior_{\#'}\gamma)}{\mathbf{z}^{\#}(\interior_{\#'}\gamma)}\right)\right].
\end{equation}
Applying the Minlos-Sinai trick \cite{Friedli2017}, we get
\begin{equation}
\frac{\Xi_{\mathbf{z}}^{\#}(\Lambda)}{\mathbf{z}^{\#}(\Lambda)}=\sum_{\substack{\Gamma\subseteq \mathcal{C}^{\#}(\Lambda):\\\text{compatible}\\\text{external}\\}}\prod_{\gamma\in\Gamma}\left(w^{\#}_{\mathbf{z}}(\gamma)\prod_{\#'}\frac{\Xi^{\#}_{\mathbf{z}}(\interior_{\#'}\gamma)}{\mathbf{z}^{\#}(\interior_{\#'}\gamma)}\right).
\end{equation} 

The computation can then be repeated for the inner ratios.
The recursion will terminate when $\interior_{\#'}\gamma$ is too small to accommodate any GFc's of type $\#$, in the sense that $\mathcal{C}^{\#}(\interior_{\#'}\gamma)=\emptyset$.
In this case, $\Xi^{\#}_{\mathbf{z}}(\interior_{\#'}\gamma)=\mathbf{z}^{\#}(\interior_{\#'}\gamma)$.
The proposition follows.
\end{proof}

\subsection{Technical estimates}

Here, we prepare several technical estimates for the proof of the main theorems.
We will use the following condition for the convergence of the cluster expansion, quoted directly from \cite{Friedli2017,Ueltschi2003}.

\begin{lemma}[cluster expansion] \label{lem:clusterexpansionfv}
Suppose that there exists a function $a:\cup_\# \mathcal{C}^{\#}\rightarrow\R_{>0}$ such that, given any ground state $\#$ and a GFc $\gamma_\ast\in \mathcal{C}^{\#}$,
\begin{equation}
\label{eqn:convergence_condition}
\sum_{\substack{\gamma\in \mathcal{C}^{\#}\\\mathrm{d}_{\Lambda_{\infty}}(\bar{\gamma},\bar{\gamma}_\ast)\le 1}}\abs{w^{\#}_{\mathbf{z}}(\gamma)}e^{a(\gamma)}\le a(\gamma_\ast).
\end{equation}
Then, given any ground state $\#$ and $\gamma_{1}\in \mathcal{C}^{\#}$,
\begin{equation}\label{eqn:velenik5-11}
1+\sum_{k=2}^{\infty} k\sum_{\gamma_{2}\in \mathcal{C}^{\#}}\dots\sum_{\gamma_{k}\in \mathcal{C}^{\#}}\abs{\varphi(\gamma_{1},\gamma_{2},\dots,\gamma_{k})}\prod_{j=2}^{k}\abs{w^{\#}_{\mathbf{z}}(\gamma_{j})}\le e^{a(\gamma_{1})},
\end{equation}
where $\varphi$ is the Ursell function. 
In addition, for any finite subregion $\Lambda\Subset\Lambda_{\infty}$, the series expansion
\begin{equation}\label{eqn:clusterexpansionform1}
\log\frac{\Xi_{\mathbf{z}}^{\#}(\Lambda)}{\mathbf{z}^{\#}(\Lambda)}
=\sum_{m=1}^{\infty}\sum_{\gamma_{1}\in \mathcal{C}^{\#}(\Lambda)}\dots\sum_{\gamma_{m}\in \mathcal{C}^{\#}(\Lambda)}\varphi(\gamma_{1},\dots,\gamma_{m})\prod_{i=1}^{m}w^{\#}_{\mathbf{z}}(\gamma_{i})
\end{equation} 
converges absolutely. 
\end{lemma}

\begin{lemma} \label{lem:clusterexpansion}
There exists a constant $\tau_{0}>0$ such that the following holds: if $\tau\ge\tau_{0}$ and $\abs{w_{\mathbf{z}}^{\#}(\gamma)}\le e^{-\tau\abs{\bar{\gamma}}}$ uniformly in $\#$ and $\gamma$, then \eqref{eqn:convergence_condition} holds with $a(\gamma):=\chi\abs{\bar{\gamma}}$, where $\chi$ is the maximal coordination number of $\Lambda_{\infty}$.
\end{lemma}

\begin{proof}
Fix $\#$ and $\gamma_\ast\in \mathcal{C}^{\#}$. 
Notice that any $\gamma\in \mathcal{C}^{\#}$ with $\mathrm{d}_{\Lambda_{\infty}}(\bar{\gamma},\bar{\gamma}_\ast)\le 1$ must intersect the set $\set{\lambda\in\Lambda_{\infty}\mid\mathrm{d}_{\Lambda_{\infty}}(\lambda,\bar{\gamma}_\ast)\le 1}$.
Bounding the size of the latter by $\chi\abs{\bar{\gamma}_\ast}$, we estimate
\begin{equation}\label{eqn:gfcentropybound}
\sum_{\substack{\gamma\in \mathcal{C}^{\#}\\\mathrm{d}_{\Lambda_{\infty}}(\bar{\gamma},\bar{\gamma}_\ast)\le 1}}\abs{w^{\#}_{\mathbf{z}}(\gamma)}e^{\chi\abs{\bar{\gamma}}}
\le\chi\abs{\bar{\gamma}_\ast}\sup_{\lambda\in\Lambda_{\infty}}\sum_{\substack{\gamma\in \mathcal{C}^{\#}\\\bar{\gamma}\ni\lambda}}e^{-(\tau-\chi)\abs{\bar{\gamma}}}.
\end{equation}

Recalling that a GFc $\gamma\in \mathcal{C}^{\#}$ is a triple $(\bar{\gamma},X_{\gamma},\mu_{\gamma})$, we bound the RHS of \eqref{eqn:gfcentropybound} as follows.
Suppose that $\abs{\bar{\gamma}}=n$.
We bound the number of $\bar{\gamma}$ with $\bar{\gamma}\ni\lambda$ and $\abs{\bar{\gamma}}=n$ by the number of walks on $\Lambda_{\infty}$ of length $2n$ starting at $\lambda$, which does not exceed $\chi^{2n}$.
Then, fixing $\bar{\gamma}$, each site therein is either occupied by a particle or not, so there are at most $2^{n}$ distinct configurations on $\bar{\gamma}$.
Finally, each hole in $\bar{\gamma}$ is adjacent to some point in $\bar{\gamma}$, so there are at most $\chi n$ such holes, each of which is assigned a label from the finite set $\mathcal{G}$.
Hence, for a fixed $\bar{\gamma}$, there are at most $\abs{\mathcal{G}}^{\chi n}$ possibilities for $\mu_{\gamma}$.
Therefore, uniformly in $\lambda\in\Lambda_{\infty}$, 
\begin{equation}\label{eqn:velenik5-28}
\sum_{\substack{\gamma\in \mathcal{C}^{\#}\\\bar{\gamma}\ni\lambda}}e^{-(\tau-\chi)\abs{\bar{\gamma}}}
\le\sum_{n=1}^{\infty} e^{-(\tau-\chi)n}\left(2\chi^{2}\abs{\mathcal{G}}^\chi\right)^{n}
\le 1
\end{equation}
for all sufficiently large $\tau$. 
In this case, from \eqref{eqn:gfcentropybound}, we get
\begin{equation}
\sum_{\substack{\gamma\in \mathcal{C}^{\#}\\\mathrm{d}_{\Lambda_{\infty}}(\bar{\gamma},\bar{\gamma}_\ast)\le 1}}\abs{w^{\#}_{\mathbf{z}}(\gamma)}e^{\chi\abs{\bar{\gamma}}}\le\chi\abs{\bar{\gamma}_\ast}.\qedhere
\end{equation}
\end{proof}

Next, we prove that the cardinality difference between the intersections of two ground states with the same region $\Lambda\Subset\Lambda_{\infty}$ is a boundary term.

\begin{lemma}\label{lem:intersection_estimate}
  For every finite region $\Lambda\Subset\Lambda_{\infty}$ and any ground state $\#\in\mathcal G$,
  \begin{equation}
    \abs{\abs{\Lambda\cap\mathcal L^\#}-\rho_{\max}\abs{\Lambda}}\le \mu^{-1}\abs{\partial^{\mathrm{in}}\Lambda}
  \end{equation}
  where $\partial^{\mathrm{in}}\Lambda$ is the {\it interior} boundary of $\Lambda$ (see Definition \ref{def:boundaries}) and $\mu$ was defined in Lemma \ref{lem:pointwiselowerbound}.
  In particular, for all $\#'\in\mathcal G$,
  \begin{equation}
    \abs{\abs{\Lambda\cap\mathcal L^\#}-\abs{\Lambda\cap\mathcal L^{\#'}}}\le 2\mu^{-1}\abs{\partial^{\mathrm{in}}\Lambda}
    .
  \end{equation}
\end{lemma}

\begin{proof}
  Recalling Definition \ref{def:effective_particle},
  \begin{equation}
    \abs{\Lambda}=\sum_{x\in\mathcal L^\#}\sum_{\lambda\in\sigma_x^\#\cap\Lambda}v^\#(\lambda)
  \end{equation}
  which we split as
  \begin{equation}
    \abs{\Lambda}=
    \sum_{x\in\mathcal L^\#:\ \sigma_x^\#\subset\Lambda}\sum_{\lambda\in\sigma_x^\#}v^\#(\lambda)
    +
    \sum_{x\in\mathcal L^\#:\ \sigma_x^\#\cap\Lambda^c\neq\emptyset}\sum_{\lambda\in\sigma_x^\#\cap\Lambda}v^\#(\lambda)
    .
  \end{equation}
  By \eqref{sumv_rhomax},
  \begin{equation}
    \abs{\Lambda}=
    \rho_{\mathrm{max}}^{-1}\abs{\set{x\in\mathcal L^\#\mid\sigma_x^\#\subset\Lambda}}
    +
    \sum_{x\in\mathcal L^\#\cap\Lambda:\ \sigma_x^\#\cap\Lambda^c\neq\emptyset}\sum_{\lambda\in\sigma_x^\#\cap\Lambda}v^\#(\lambda)
    .
  \end{equation}
  Therefore, splitting
  \begin{equation}
    \abs{\set{x\in\mathcal L^\#\mid\sigma_x^\#\subset\Lambda}}
    =
    \abs{\Lambda\cap\mathcal L^\#}
    -
    \abs{\set{x\in\mathcal L^\#\cap\Lambda\mid\sigma_x^\#\cap\Lambda^c\neq\emptyset}}
  \end{equation}
  we find that
  \begin{equation}
  \label{eqn:region-intersection-difference}
    \abs{\Lambda}-\rho_{\mathrm{max}}^{-1}\abs{\Lambda\cap\mathcal L^\#}
    =
    \sum_{x\in\mathcal L^\#\cap\Lambda:\ \sigma_x^\#\cap\Lambda^c\neq\emptyset}\sum_{\lambda\in\sigma_x^\#\cap\Lambda}v^\#(\lambda)
    -\rho_{\mathrm{max}}^{-1}\abs{\set{x\in\mathcal L^\#\cap\Lambda\mid\sigma_x^\#\cap\Lambda^c\neq\emptyset}}
    .
  \end{equation}
  By\-~\eqref{sumv_rhomax}, the RHS of \eqref{eqn:region-intersection-difference} is nonpositive, so
  \begin{equation}
    \abs{\abs{\Lambda}-\rho_{\mathrm{max}}^{-1}\abs{\Lambda\cap\mathcal L^\#}}
    \le
    \rho_{\mathrm{max}}^{-1}\abs{\set{x\in\mathcal L^\#\cap\Lambda\mid\sigma_x^\#\cap\Lambda^c\neq\emptyset}}
    .
  \end{equation}
  Since $\sigma_x^\#$ is connected,
  \begin{equation}
    \abs{\set{x\in\mathcal L^\#\cap\Lambda\mid\sigma_x^\#\cap\Lambda^c\neq\emptyset}}
    \le
    \abs{\set{x\in\mathcal L^\#\mid\sigma_x^\#\cap\partial^{\mathrm{in}}\Lambda\neq\emptyset}}
    .
  \end{equation}
  Now, recalling the definition of $\mu$ in Lemma \ref{lem:pointwiselowerbound}, we have that each point in $\sigma_x^\#$ can belong to the reference Voronoi cell of at most $\mu^{-1}$ particles.
  Therefore,
  \begin{equation}
    \abs{\set{x\in\mathcal L^\#\mid\sigma_x^\#\cap\partial^{\mathrm{in}}\Lambda\neq\emptyset}}
    \le
    \mu^{-1}\abs{\partial^{\mathrm{in}}\Lambda}.
  \end{equation}
\end{proof}

We now derive the central estimates of this subsection.

\begin{proposition} \label{prop:central_estimates}
Suppose that $\mathbf{z}(x)\equiv z$ for all but $n$ sites $x_{1},\dots,x_{n}\in\Lambda_{\infty}$, for which there exists a constant $c>0$ such that $e^{-\frac{c}{n}}\abs{z}\le\abs{\mathbf{z}(x_{i})}\le e^{\frac{c}{n}}\abs{z}$ for all $i$.
There exist constants $\tau\ge\tau_{0}$, $\varsigma$, $\eta>0$ and $z_{0}>1$ such that, whenever $\abs{z}\ge z_{0}$, the following hold for all finite regions $\Lambda\Subset\Lambda_{\infty}$, ground states $\#,\#'\in\mathcal G$, and GFc's $\gamma\in\mathcal{C}^{\#}(\Lambda)$:
\begin{equation}\label{eqn:weightestimate}
\abs{w_{\mathbf{z}}^{\#}(\gamma)}\le e^{-\tau\abs{\bar{\gamma}}},
\end{equation}
\begin{equation}\label{eqn:derivative_estimate}
\abs{\frac{\partial}{\partial\log \mathbf{z}(x_{i})}\log\frac{\Xi^{\#}_{\mathbf{z}}(\Lambda)}{\mathbf{z}^{\#}(\Lambda)}}\le \eta\indicator{x_{i}\in\Lambda},
\end{equation} 
\begin{equation}\label{eqn:ratio_estimate}
\abs{\frac{\Xi^{\#'}_{\mathbf{z}}(\Lambda)}{\Xi^{\#}_{\mathbf{z}}(\Lambda)}}\le\frac{\abs{z}^{\abs{\Lambda\cap \mathcal{L}^{\#'}}}}{\abs{z}^{\abs{\Lambda\cap \mathcal{L}^{\#}}}}e^{\varsigma\abs{\partial^{\text{in}}\Lambda}}.
\end{equation}
\end{proposition}

\begin{remark}
\label{rem:fugacity_real_vs_complex}
To derive the analyticity of the pressure, we will allow the fugacity function $\mathbf{z}$ to take complex values, but in this case we require that $\mathbf{z}$ be a constant function on $\Lambda_{\infty}$.
In the case of non-constant $\mathbf{z}$ as is needed in \eqref{eqn:derivative_estimate}, the function is required to take only real values.
\end{remark}

\begin{proof}
Following the version of Pirogov-Sinai theory proposed by Zahradn\'ik \cite{Zahradnk1984}, we assign a \emph{cutoff weight} to each GFc $\gamma$ in such a way that
\begin{equation}\label{eqn:cutoff_weight}
\hat{w}_{\mathbf{z}}^{\#}(\gamma):=
\min\set{\frac{\prod_{x\in X_{\gamma}}\abs{\mathbf{z}(x)}}{\prod_{x\in (\mathcal{L}^{\#}\cap\bar{\gamma})}\abs{\mathbf{z}(x)}}\prod_{\#'}\abs{\frac{\hat{\Xi}_{\mathbf{z}}^{\#'}(\interior_{\#'}\gamma)}{\hat{\Xi}^{\#}_{\mathbf{z}}(\interior_{\#'}\gamma)}},e^{-\tau_{0}\abs{\bar{\gamma}}}},
\end{equation}
where $\hat{\Xi}_{\mathbf{z}}^{\#}(\Lambda)$ is defined as in \eqref{eqn:gfc_partition_function} but with the true weight $w_{\mathbf{z}}^{\#}(\gamma)$ replaced by the cutoff weight $\hat{w}_{\mathbf{z}}^{\#}(\gamma)$.
It is a standard result that there exists a unique way to execute this assignment; see for instance \cite[Theorem 10.5.1.2]{Presutti2008}.
The benefit of using the cutoff weights is suggested by the presence of $\tau_{0}$ in \eqref{eqn:cutoff_weight}: since these weights are sufficiently small, Lemma \ref{lem:clusterexpansion} guarantees that we can use the cluster expansion to evaluate the associated partition functions.

In what follows, we prove \eqref{eqn:derivative_estimate} and \eqref{eqn:ratio_estimate} inductively, in addition to the following bound on the cutoff weights
\begin{equation}
\label{eqn:cutoff_weight_estimate}
\hat{w}_{\mathbf{z}}^{\#}(\gamma)\le e^{-\tau\abs{\bar{\gamma}}},
\end{equation}
which immediately implies \eqref{eqn:weightestimate} as $\tau\ge\tau_{0}$; see for instance \cite[Theorem 10.5.2.1]{Presutti2008}.

If $\abs{\Lambda}$ is too small to accommodate any GFc of any type, then there is nothing to prove about \eqref{eqn:cutoff_weight_estimate}.
On the other hand, $\Xi^{\#}_{\mathbf{z}}(\Lambda)=\mathbf{z}^{\#}(\Lambda)$ for each $\#$, so \eqref{eqn:derivative_estimate} holds trivially. As for \eqref{eqn:ratio_estimate}, given $\#$ and $\#'$, we estimate
\begin{equation}
\abs{\frac{\Xi^{\#'}_{\mathbf{z}}(\Lambda)}{\Xi^{\#}_{\mathbf{z}}(\Lambda)}}=\frac{\prod_{x\in\Lambda\cap \mathcal{L}^{\#'}}\abs{\mathbf{z}(x)}}{\prod_{x\in\Lambda\cap \mathcal{L}^{\#}}\abs{\mathbf{z}(x)}}\le e^{2c}\frac{\abs{z}^{\abs{\Lambda\cap \mathcal{L}^{\#'}}}}{\abs{z}^{\abs{\Lambda\cap \mathcal{L}^{\#}}}}\le e^{\varsigma\abs{\partial^{\text{in}}\Lambda}}\frac{\abs{z}^{\abs{\Lambda\cap \mathcal{L}^{\#'}}}}{\abs{z}^{\abs{\Lambda\cap \mathcal{L}^{\#}}}},
\end{equation} 
where the last inequality holds as long as
\begin{equation}
\label{eqn:boundary-coefficient-condition-1}
\varsigma\ge 2c.
\end{equation}

Assume henceforth that \eqref{eqn:cutoff_weight_estimate}, \eqref{eqn:derivative_estimate}, and \eqref{eqn:ratio_estimate} hold for all proper subregions of $\Lambda$.

We begin with \eqref{eqn:cutoff_weight_estimate}. 
Fix a ground state $\#$, and consider $\gamma\in \mathcal{C}^{\#}(\Lambda)$. 
Applying \eqref{eqn:cutoff_weight_estimate} and \eqref{eqn:ratio_estimate} inductively, we get 
\begin{equation}
\begin{split}
\abs{\hat{w}_{\mathbf{z}}^{\#}(\gamma)}&=\frac{\prod_{x\in X_{\gamma}}\abs{\mathbf{z}(x)}}{\prod_{x\in (\mathcal{L}^{\#}\cap\bar{\gamma})}\abs{\mathbf{z}(x)}}\prod_{\#'}\abs{\frac{\Xi_{\mathbf{z}}^{\#'}(\interior_{\#'}\gamma)}{\Xi^{\#}_\mathbf{z}(\interior_{\#'}\gamma)}}\\
&\le e^{2c}\frac{\abs{z}^{\abs{X_{\gamma}}}}{\abs{z}^{\abs{(\mathcal{L}^{\#}\cap\bar{\gamma})}}}\prod_{\#'}\left(\frac{\abs{z}^{\abs{\interior_{\#'}\gamma\cap \mathcal{L}^{\#'}}}}{\abs{z}^{\abs{\interior_{\#'}\gamma\cap \mathcal{L}^{\#}}}}e^{\varsigma\abs{\partial^{\text{in}}\interior_{\#'}\gamma}}\right)\\
&=e^{2c+\varsigma\sum_{\#'}\abs{\partial^{\text{in}}\interior_{\#'}\gamma}}\abs{z}^{\abs{X_{\gamma}}+\sum_{\#'}\abs{\interior_{\#'}\gamma\cap \mathcal{L}^{\#'}}-\abs{\Int\gamma\cap \mathcal{L}^{\#}}}.
\end{split}
\end{equation}
Notice that, by Lemma \ref{lemma:rholoc_cst},
\begin{equation}
\begin{split}
&\abs{X_{\gamma}}+\sum_{\#'}\abs{\interior_{\#'}\gamma\cap \mathcal{L}^{\#'}}-\abs{\Int\gamma\cap \mathcal{L}^{\#}}\\
=&\abs{X_{\gamma}}-\rho_{\max}\left(\sum_{x\in\Int\gamma\cap \mathcal{L}^{\#}}\norm{\sigma^{\#}_{x}}-\sum_{\#'}\sum_{x\in\interior_{\#'}\gamma\cap \mathcal{L}^{\#'}}\norm{\sigma^{\#'}_{x}}\right)\\
=&\abs{X_{\gamma}}-\rho_{\max}\norm{\bar{\gamma}}\le-(\rho_{\max}-\rho_{0})\norm{\bar{\gamma}}\le-\mu(\rho_{\max}-\rho_{0})\abs{\bar{\gamma}},
\end{split}
\end{equation}
where, in an abuse of notation, we write, for $x\in\mathcal{L}^{\#}$,
\begin{equation}
\norm{\sigma_{x}^{\#}}:=\sum_{\lambda\in\sigma^\#_x}v^\#(\lambda),
\end{equation}
and we use the Peierls condition in the first inequality and Lemma \ref{lem:effective_volume_naive_bound} in the second.
Hence,
\begin{equation}
\abs{\hat{w}_{\mathbf{z}}^{\#}(\gamma)}\le e^{2c+\varsigma\sum_{\#'}\abs{\partial^{\text{in}}\interior_{\#'}\gamma}}\abs{z}^{-\mu(\rho_{\max}-\rho_{0})\abs{\bar{\gamma}}}\le e^{-[\mu(\rho_{\max}-\rho_{0})\log z_{0}-2c-\varsigma\chi]\abs{\bar{\gamma}}}\le e^{-\tau\abs{\bar{\gamma}}},
\end{equation} 
where the last inequality holds as long as
\begin{equation}
\label{eqn:high-fugacity}
\mu(\rho_{\max}-\rho_{0})\log z_{0}-2c-\varsigma\chi\ge\tau.
\end{equation}

We now prove \eqref{eqn:derivative_estimate}. 
If $x_{i}\not\in\Lambda$, then the inequality holds trivially, so we assume otherwise. As long as
\begin{equation}
\tau\ge\tau_{0},
\end{equation} 
the cluster expansion
\begin{equation}
\log\frac{\Xi_{\mathbf{z}}^{\#}(\Lambda)}{\mathbf{z}^{\#}(\Lambda)}=\sum_{m=1}^{\infty}\sum_{\gamma_{1}\in \mathcal{C}^{\#}(\Lambda)}\dots\sum_{\gamma_{m}\in \mathcal{C}^{\#}(\Lambda)}\varphi(\gamma_{1},\dots,\gamma_{m})\prod_{j=1}^{m}w^{\#}_{\mathbf{z}}(\gamma_{j})
\end{equation}
converges absolutely by Lemma \ref{lem:clusterexpansion}. Differentiating the series term by term (which, by a corollary of the dominated convergence theorem, will be justified as soon as we show that the series of derivatives is bounded absolutely and uniformly for $\mathbf{z}(x_{i})\in [e^{-\frac{c}{n}}\abs{z},e^{\frac{c}{n}}\abs{z}]$; indeed, the latter is guaranteed by \eqref{eqn:derivativecondition1} and \eqref{eqn:derivativecondition2} as we compute below) and using Lemma \ref{lem:clusterexpansionfv}, we get 
\begin{equation}\label{eqn:derivativebeginning}
\begin{split}
&\abs{\frac{\partial}{\partial\log \mathbf{z}(x_{i})}\log\frac{\Xi^{\#}_{\mathbf{z}}(\Lambda)}{\mathbf{z}^{\#}(\Lambda)}}\\
\le&\sum_{m=1}^{\infty}\sum_{\gamma_{1}\in\mathcal{C}^{\#}(\Lambda)}\dots\sum_{\gamma_{m}\in\mathcal{C}^{\#}(\Lambda)}\abs{\varphi(\gamma_{1},\dots,\gamma_{m})}\sum_{j=1}^{m}\abs{\frac{\partial w^{\#}_{\mathbf{z}}(\gamma_{j})}{\partial\log \mathbf{z}(x_{i})}}\prod_{k\ne j}\abs{w^{\#}_{\mathbf{z}}(\gamma_{k})}\\
\le&\sum_{\gamma_{1}\in\mathcal{C}^{\#}(\Lambda)}\abs{\frac{\partial w^{\#}_{\mathbf{z}}(\gamma_{1})}{\partial\log \mathbf{z}(x_{i})}}\left[1+\sum_{m=2}^{\infty} m\sum_{\gamma_{2}\in\mathcal{C}^{\#}(\Lambda)}\dots\sum_{\gamma_{m}\in\mathcal{C}^{\#}(\Lambda)}\abs{\varphi(\gamma_{1},\dots,\gamma_{m})}\prod_{j=2}^{m}\abs{w^{\#}_{\mathbf{z}}(\gamma_{j})}\right]\\
\le&\sum_{\gamma\in\mathcal{C}^{\#}(\Lambda)}\abs{\frac{\partial w^{\#}_{\mathbf{z}}(\gamma)}{\partial\log \mathbf{z}(x_{i})}}e^{\chi\abs{\bar{\gamma}}}.
\end{split}
\end{equation}
Since
\begin{equation}
\frac{\partial w^{\#}_{\mathbf{z}}(\gamma)}{\partial\log \mathbf{z}(x_{i})}=w_{\mathbf{z}}^{\#}(\gamma)\frac{\partial\log w^{\#}_{\mathbf{z}}(\gamma)}{\partial\log \mathbf{z}(x_{i})},
\end{equation} 
we need to study 
\begin{equation}
\begin{multlined}
\log w_{\mathbf{z}}^{\#}(\gamma)=\left(\sum_{x\in X_{\gamma}}\log \mathbf{z}(x)-\sum_{x\in(\mathcal{L}^{\#}\cap\bar{\gamma})}\log \mathbf{z}(x)\right)\\
+\sum_{\#'}\left(\log\frac{\Xi^{\#'}_{\mathbf{z}}(\interior_{\#'}\gamma)}{\mathbf{z}^{\#'}(\interior_{\#'}\gamma)}-\log\frac{\Xi^{\#}_{\mathbf{z}}(\interior_{\#'}\gamma)}{\mathbf{z}^{\#}(\interior_{\#'}\gamma)}\right)\\
+\sum_{\#'}\left(\sum_{x\in\interior_{\#'}\gamma\cap \mathcal{L}^{\#'}}\log \mathbf{z}(x)-\sum_{x\in\interior_{\#'}\gamma\cap \mathcal{L}^{\#}}\log \mathbf{z}(x)\right).
\end{multlined}
\end{equation}
Differentiating the above and applying \eqref{eqn:derivative_estimate} inductively, we get 
\begin{equation}\label{eqn:firstderivativeindicatorbound}
\begin{split}
\abs{\frac{\partial\log w^{\#}_{\mathbf{z}}(\gamma)}{\partial\log \mathbf{z}(x_{i})}}\le&\abs{\indicator{x_{1}\in X_{\gamma}}-\indicator{x_{i}\in(\mathcal{L}^{\#}\cap\bar{\gamma})}}\\&+2\eta\sum_{\#'}\indicator{x_{i}\in\interior_{\#'}\gamma\cap \mathcal{L}^{\#'}}+\sum_{\#'}\abs{\indicator{x_{i}\in\interior_{\#'}\gamma\cap \mathcal{L}^{\#'}}-\indicator{x_{i}\in\interior_{\#'}\gamma\cap \mathcal{L}^{\#}}}\\
\le&3\indicator{x_{i}\in\Int\gamma},
\end{split}
\end{equation}
where the last inequality holds as long as 
\begin{equation}\label{eqn:derivativecondition1}
\eta\le 1.
\end{equation}
Therefore, 
\begin{equation}
\abs{\frac{\partial}{\partial\log \mathbf{z}(x_{i})}\log\frac{\Xi^{\#}_{\mathbf{z}}(\Lambda)}{\mathbf{z}^{\#}(\Lambda)}}\le3\sum_{\substack{\gamma\in\mathcal{C}^{\#}(\Lambda)\\x_{i}\in\Int\gamma}}e^{-(\tau-\chi)\abs{\bar{\gamma}}}.
\end{equation}
Notice that every $\gamma\in \mathcal{C}^{\#}$ with $x_{i}\in\bar{\gamma}$ has at most $\abs{\Int_{\gamma}}$ distinct translates $\gamma'\in \mathcal{C}^{\#}(\Lambda)$ such that $x_{1}\in\Int\gamma'$. 
Using the isoperimetric inequality, 
\begin{equation}\label{eqn:isoperimetric}
\abs{\Int\gamma}\le I_d\abs{\bar{\gamma}}^{d},
\end{equation}
we get 
\begin{equation}
\sum_{\substack{\gamma\in\mathcal{C}^{\#}(\Lambda)\\x_{i}\in\Int\gamma}}e^{-(\tau-\chi)\abs{\bar{\gamma}}}\le\sum_{\substack{\gamma\in\mathcal{C}^{\#}\\x_{i}\in\bar{\gamma}}}I_d\abs{\bar{\gamma}}^{d}e^{-(\tau-\chi)\abs{\bar{\gamma}}}\le I_d d!\sum_{\substack{\gamma\in\mathcal{C}^{\#}\\x_{i}\in\bar{\gamma}}}e^{-(\tau-\chi-1)\abs{\bar{\gamma}}}.
\end{equation} 
Bounding the series the same way as in \eqref{eqn:velenik5-28}, we get 
\begin{equation}\label{eqn:derivativeend}
\abs{\frac{\partial}{\partial\log \mathbf{z}(x_{i})}\log\frac{\Xi^{\#}_{\mathbf{z}}(\Lambda)}{\mathbf{z}^{\#}(\Lambda)}}\le3I_dd!\sum_{s=1}^{\infty} e^{-(\tau-\chi-1)s}\left(2\chi^{2}\abs{\mathcal{G}}^\chi\right)^s\le \eta,
\end{equation} 
where the last inequality holds as long as 
$\tau\ge\tau_{1}$, where $\tau_{1}$ satisfies $\tau_{1}\ge\tau_{0}$ and 
\begin{equation}\label{eqn:derivativecondition2}
\sum_{s=1}^{\infty} e^{-(\tau_{1}-\chi-1)s}\left(2\chi^{2}\abs{\mathcal{G}}^\chi\right)^s\le\frac{\eta}{3I_dd!}.
\end{equation}

It remains to prove \eqref{eqn:ratio_estimate}. 
If $\mathbf{z}$ is non-constant (in which case, recall from Remark \ref{rem:fugacity_real_vs_complex} that it can only take real values), then we first get rid of the non-constancy as follows.
By the mean-value theorem, there exists $\widetilde{\mathbf{z}}(x_{i})\in[z,\mathbf{z}(x_{i})]$ (or possibly $[\mathbf{z}(x_{i}),z]$) for each $i$ such that
\begin{equation}
\label{eqn:bye_variable_z}
\log\frac{\Xi^{\#'}_{\mathbf{z}}(\Lambda)}{\Xi^{\#}_{\mathbf{z}}(\Lambda)}-\log\frac{\Xi^{\#'}_{z}(\Lambda)}{\Xi^{\#}_{z}(\Lambda)}=\sum_{i=1}^{n}\left[\evaluate{\frac{\partial}{\partial \mathbf{z}(x_{i})}\log\frac{\Xi^{\#'}_{\mathbf{z}}(\Lambda)}{\Xi^{\#}_{\mathbf{z}}(\Lambda)}}{\widetilde{\mathbf{z}}}(\mathbf{z}(x_{i})-z)\right],
\end{equation}
where we extend $\widetilde{\mathbf{z}}(x):= z$ for all $x\ne x_{1},\dots,x_{n}$. 
Notice that \eqref{eqn:derivative_estimate} remains valid since nothing has been used about $\widetilde{\mathbf{z}}$ except that it satisfies the constraint stated in the proposition, so 
\begin{equation}
\label{eqn:bye_variable_z_error}
\begin{split}
&\sum_{i=1}^{n}\abs{\evaluate{\frac{\partial}{\partial \mathbf{z}(x_{i})}\log\frac{\Xi^{\#'}_{\mathbf{z}}(\Lambda)}{\Xi^{\#}_{\mathbf{z}}(\Lambda)}}{\widetilde{\mathbf{z}}}}\abs{\mathbf{z}(x_{i})-z}\\
\le&\sum_{i=1}^{n}\left(\abs{\evaluate{\frac{\partial}{\partial \mathbf{z}(x_{i})}\log\frac{\Xi^{\#'}_{\mathbf{z}}(\Lambda)}{\mathbf{z}^{\#'}(\Lambda)}}{\widetilde{\mathbf{z}}}}+\abs{\evaluate{\frac{\partial}{\partial \mathbf{z}(x_{i})}\log\frac{\Xi^{\#}_{\mathbf{z}}(\Lambda)}{\mathbf{z}^{\#}(\Lambda)}}{\widetilde{\mathbf{z}}}}\right.\\
&\hspace{1.5in}+\left.\abs{\evaluate{\frac{\partial\log \mathbf{z}^{\#'}(\Lambda)}{\partial \mathbf{z}(x_{i})}}{\widetilde{\mathbf{z}}}-\evaluate{\frac{\partial\log \mathbf{z}^{\#}(\Lambda)}{\partial \mathbf{z}(x_{i})}}{\widetilde{\mathbf{z}}}}\right)\abs{\mathbf{z}(x_{i})-z}\\
\le&\sum_{i=1}^{n}\frac{\abs{\mathbf{z}(x_{i})-z}}{\abs{\widetilde{\mathbf{z}}(x_{i})}}\left(2\eta\indicator{x_{i}\in\Lambda}+\abs{\indicator{x_{i}\in \mathcal{L}^{\#'}\cap\Lambda}-\indicator{x_{i}\in \mathcal{L}^{\#}\cap\Lambda}}\right)\\
\le&\sum_{i=1}^{n}(e^\frac{c}{n}+1)(2\eta+1).
\end{split}
\end{equation}
We now turn to
\begin{equation}
\label{eqn:ratio_of_partition_functions}
\log\frac{\Xi^{\#'}_{z}(\Lambda)}{\Xi^{\#}_{z}(\Lambda)}=\log\frac{\Xi^{\#'}_{z}(\Lambda)}{\abs{z}^{\abs{\mathcal{L}^{\#'}\cap\Lambda}}}-\log\frac{\Xi^{\#}_{z}(\Lambda)}{\abs{z}^{\abs{\mathcal{L}^{\#}\cap\Lambda}}}+\log\frac{\abs{z}^{\abs{\mathcal{L}^{\#'}\cap\Lambda}}}{\abs{z}^{\abs{\mathcal{L}^{\#}\cap\Lambda}}}.
\end{equation}
By Lemma \ref{lem:clusterexpansionfv}, we expand 
\begin{equation}
\label{eqn:cluster_expansion_cutoff_weight}
\log\frac{\Xi^{\#}_{z}(\Lambda)}{\abs{z}^{\abs{\mathcal{L}^{\#}\cap\Lambda}}}
=\sum_{m=1}^{\infty}\sum_{\gamma_{1}\in \mathcal{C}^{\#}(\Lambda)}\dots\sum_{\gamma_{m}\in \mathcal{C}^{\#}(\Lambda)}\varphi(\gamma_{1},\dots,\gamma_{m})\prod_{i=1}^{m}\hat{w}^{\#}_{z}(\gamma_{i}),
\end{equation}
where we substitute in the cutoff weights (which coincide with the true weights for the GFc's in $\mathcal{C}^{\#}(\Lambda)$ by \eqref{eqn:cutoff_weight} and \eqref{eqn:cutoff_weight_estimate}) in anticipation of extending the summation to be over tuples of GFc's in $\mathcal{C}^{\#}$.
Using the identity
\begin{equation}
1=\frac{1}{\abs{\bar{\Gamma}\cap \mathcal{L}^{\#}}}\sum_{\lambda\in\Lambda\cap \mathcal{L}^{\#}}\indicator{\lambda\in\bar{\Gamma}}
\end{equation}
for each tuple $\Gamma=(\gamma_{1},\dots,\gamma_{m})$ of GFc's in $\mathcal{C}^{\#}(\Lambda)$ contributing to \eqref{eqn:cluster_expansion_cutoff_weight} (where, by an abuse of notation, we write $\bar{\Gamma}:=\cup_{i}\bar{\gamma}_{i}$),
we get 
\begin{equation}
\label{eqn:cutoff_partition_function_evaluated}
\begin{split}
&\sum_{m=1}^{\infty}\sum_{\gamma_{1}\in \mathcal{C}^{\#}(\Lambda)}\dots\sum_{\gamma_{m}\in \mathcal{C}^{\#}(\Lambda)}\varphi(\gamma_{1},\dots,\gamma_{m})\prod_{i=1}^{m}\hat{w}^{\#}_{z}(\gamma_{i}) 
\\
=&\sum_{\lambda\in\Lambda\cap \mathcal{L}^{\#}}
\sum_{m=1}^{\infty}\sum_{\gamma_{1}\in \mathcal{C}^{\#}(\Lambda)}\dots\sum_{\gamma_{m}\in \mathcal{C}^{\#}(\Lambda)}
\indicator{\lambda\in\bar{\Gamma}}
\frac{1}{\abs{\bar{\Gamma}\cap\mathcal{L}^{\#}}}
\varphi(\gamma_{1},\dots,\gamma_{m})\prod_{i=1}^{m}\hat{w}^{\#}_{z}(\gamma_{i})
\\
=&
\begin{multlined}[t]
\sum_{\lambda\in\Lambda\cap \mathcal{L}^{\#}}
\left[
\sum_{m=1}^{\infty}\sum_{\gamma_{1}\in \mathcal{C}^{\#}}\dots\sum_{\gamma_{m}\in \mathcal{C}^{\#}}
\indicator{\lambda\in\bar{\Gamma}}
\frac{1}{\abs{\bar{\Gamma}\cap\mathcal{L}^{\#}}}
\varphi(\gamma_{1},\dots,\gamma_{m})\prod_{i=1}^{m}\hat{w}^{\#}_{z}(\gamma_{i})
\right.
\\
\left.
-\sum_{m=1}^{\infty}\sum_{\gamma_{1}\in \mathcal{C}^{\#}}\dots\sum_{\gamma_{m}\in \mathcal{C}^{\#}}
\indicator{\Gamma\not\sqsubset\mathcal{C}^{\#}(\Lambda)}
\indicator{\lambda\in\bar{\Gamma}}
\frac{1}{\abs{\bar{\Gamma}\cap\mathcal{L}^{\#}}}
\varphi(\gamma_{1},\dots,\gamma_{m})\prod_{i=1}^{m}\hat{w}^{\#}_{z}(\gamma_{i})
\right],
\end{multlined}
\end{split}
\end{equation}
where the involvement of GFc's in $\mathcal{C}^{\#}$ is justified by the estimate
\begin{equation}\label{eqn:largercluster}
\begin{split}
&\sum_{m=1}^{\infty}\sum_{\gamma_{1}\in \mathcal{C}^{\#}}\dots\sum_{\gamma_{m}\in \mathcal{C}^{\#}}
\indicator{\lambda\in\bar{\Gamma}}
\frac{1}{\abs{\bar{\Gamma}\cap\mathcal{L}^{\#}}}
\abs{\varphi(\gamma_{1},\dots,\gamma_{m})}
\prod_{i=1}^{m}\abs{\hat{w}^{\#}_{z}(\gamma_{i})}
\\
\le&
\sum_{\substack{\gamma_{1}\in \mathcal{C}^{\#}\\\bar{\gamma}_{1}\ni\lambda}}
\abs{\hat{w}^{\#}_{z}(\gamma_{1})}
\sum_{m=1}^{\infty}
m
\sum_{\gamma_{2}\in \mathcal{C}^{\#}}\dots\sum_{\gamma_{m}\in \mathcal{C}^{\#}}
\abs{\varphi(\gamma_{1},\dots,\gamma_{m})}
\prod_{i=2}^{m}\abs{\hat{w}^{\#}_{z}(\gamma_{i})}
\\
\le&\sum_{\substack{\gamma\in \mathcal{C}^{\#}\\\lambda\in\bar{\gamma}}}e^{-(\tau_{0}-\chi)\abs{\bar{\gamma}}}\le 1
\end{split}
\end{equation}
as a consequence of Lemmas \ref{lem:clusterexpansionfv} and \ref{lem:clusterexpansion} (in particular \eqref{eqn:velenik5-28}).
By \eqref{eqn:action_on_gfc}, the first series 
\begin{equation}
\sum_{m=1}^{\infty}\sum_{\gamma_{1}\in \mathcal{C}^{\#}}\dots\sum_{\gamma_{m}\in \mathcal{C}^{\#}}
\indicator{\lambda\in\bar{\Gamma}}
\frac{1}{\abs{\bar{\Gamma}\cap\mathcal{L}^{\#}}}
\varphi(\gamma_{1},\dots,\gamma_{m})\prod_{i=1}^{m}\hat{w}^{\#}_{z}(\gamma_{i})
\label{bulkterm}
\end{equation}
appearing in \eqref{eqn:cutoff_partition_function_evaluated} is independent of the ground state $\#$ and the site $\lambda\in \mathcal{L}^{\#}$.
On the other hand, the \emph{support} $\bar{\Gamma}$ of each tuple $\Gamma$ that contributes to the other series in \eqref{eqn:cutoff_partition_function_evaluated},
\begin{equation}
\sum_{\lambda\in\Lambda\cap \mathcal{L}^{\#}}
\sum_{m=1}^{\infty}\sum_{\gamma_{1}\in \mathcal{C}^{\#}}\dots\sum_{\gamma_{m}\in \mathcal{C}^{\#}}
\indicator{\Gamma\not\sqsubset\mathcal{C}^{\#}(\Lambda)}
\indicator{\lambda\in\bar{\Gamma}}
\frac{1}{\abs{\bar{\Gamma}\cap\mathcal{L}^{\#}}}
\varphi(\gamma_{1},\dots,\gamma_{m})\prod_{i=1}^{m}\hat{w}^{\#}_{z}(\gamma_{i})
\end{equation}
necessarily intersects $\partial^{\text{in}}\Lambda$ by the definition of the Ursell function $\varphi$.
Hence, this latter series is a boundary term:
\begin{equation}\label{eqn:clusterexpansionboundaryterm}
\begin{multlined}
\sum_{\lambda\in\Lambda\cap \mathcal{L}^{\#}}
\sum_{m=1}^{\infty}\sum_{\gamma_{1}\in \mathcal{C}^{\#}}\dots\sum_{\gamma_{m}\in \mathcal{C}^{\#}}
\indicator{\Gamma\not\sqsubset\mathcal{C}^{\#}(\Lambda)}
\indicator{\lambda\in\bar{\Gamma}}
\frac{1}{\abs{\bar{\Gamma}\cap\mathcal{L}^{\#}}}
\abs{\varphi(\gamma_{1},\dots,\gamma_{m})}
\prod_{i=1}^{m}\abs{\hat{w}^{\#}_{z}(\gamma_{i})}
\\
\le
\abs{\partial^{\text{in}}\Lambda}
\max_{\lambda\in\partial^{\text{in}}\Lambda}
\sum_{\substack{\gamma_{1}\in \mathcal{C}^{\#}\\\bar{\gamma}_{1}\ni\lambda}}
\abs{\hat{w}^{\#}_{z}(\gamma_{1})}
\sum_{m=1}^{\infty}
m
\sum_{\gamma_{2}\in \mathcal{C}^{\#}}\dots\sum_{\gamma_{m}\in \mathcal{C}^{\#}}
\abs{\varphi(\gamma_{1},\dots,\gamma_{m})}
\prod_{i=2}^{m}\abs{\hat{w}^{\#}_{z}(\gamma_{i})}
\le\abs{\partial^{\text{in}}\Lambda},
\end{multlined}
\end{equation}
again using \eqref{eqn:largercluster}. 
Therefore, 
\begin{equation}
\label{eqn:difference_of_partition_functions}
\abs{\log\frac{\Xi^{\#}_{z}(\Lambda)}{\abs{z}^{\abs{\mathcal{L}^{\#}\cap\Lambda}}}-\log\frac{\Xi^{\#'}_{z}(\Lambda)}{\abs{z}^{\abs{\mathcal{L}^{\#'}\cap\Lambda}}}}
\le\abs{\abs{\Lambda\cap \mathcal{L}^{\#}}-\abs{\Lambda\cap \mathcal{L}^{\#'}}}+2\abs{\partial^{\text{in}}\Lambda}
\le\left(2+2\mu^{-1}\right)\abs{\partial^{\text{in}}\Lambda}
\end{equation}
where we use Lemma \ref{lem:intersection_estimate} in the last inequality.
Putting together \eqref{eqn:bye_variable_z}, \eqref{eqn:bye_variable_z_error}, \eqref{eqn:ratio_of_partition_functions}, and \eqref{eqn:difference_of_partition_functions}, we get
\begin{equation}
\abs{\log\frac{\Xi^{\#'}_{\mathbf{z}}(\Lambda)}{\Xi^{\#}_{\mathbf{z}}(\Lambda)}-\log\frac{\abs{z}^{\abs{\mathcal{L}^{\#'}\cap\Lambda}}}{\abs{z}^{\abs{\mathcal{L}^{\#}\cap\Lambda}}}}\le\sum_{i=1}^{n}(e^\frac{c}{n}+1)(2\eta+1)+\left(2+2\mu^{-1}\right)\abs{\partial^{\text{in}}\Lambda}\le\varsigma\abs{\partial^{\text{in}}\Lambda},
\end{equation} 
where the last inequality holds as long as
\begin{equation}
\label{eqn:boundary-coefficient-condition-2}
\varsigma\ge 3n(e^{\frac{c}{n}}+1)+2+2\mu^{-1}.\qedhere
\end{equation}
\end{proof}

\subsection{Proof of the main theorems}

Our main theorems follow directly from the estimates in Proposition \ref{prop:central_estimates}.

\begin{proof}[Proof of Theorem \ref{thm:analyticity}]
By \eqref{eqn:weightestimate} and \eqref{eqn:cutoff_weight}, the cutoff weights coincide exactly with the true weights. 
Thus, it follows from \eqref{eqn:cutoff_partition_function_evaluated} that
\begin{equation}
\label{eqn:finite_volume_pressure_expanded}
\begin{multlined}
p^{\#}_{z}(\Lambda)
:=\frac{1}{\abs{\Lambda}}\log\Xi^{\#}_{z}(\Lambda)
\\
=\frac{\abs{\mathcal{L}^{\#}\cap\Lambda}}{\abs{\Lambda}}
\left[
\log z
+\sum_{m=1}^{\infty}\sum_{\gamma_{1}\in \mathcal{C}^{\#}}\dots\sum_{\gamma_{m}\in \mathcal{C}^{\#}}
\indicator{\lambda\in\bar{\Gamma}}
\frac{1}{\abs{\bar{\Gamma}\cap\mathcal{L}^{\#}}}
\varphi(\gamma_{1},\dots,\gamma_{m})\prod_{i=1}^{m}w^{\#}_{z}(\gamma_{i})
\right]
+\frac{O(\abs{\partial^{\text{in}}\Lambda})}{\abs{\Lambda}}.
\end{multlined}
\end{equation}
Taking the limit $\Lambda\Uparrow\Lambda_{\infty}$, we obtain the expansion
\begin{equation}\label{eqn:pressurehighfugacityexpansion}
p(z)-\rho_{\max}\log z
=\rho_{\max}\sum_{m=1}^{\infty}\sum_{\gamma_{1}\in \mathcal{C}^{\#}}\dots\sum_{\gamma_{m}\in \mathcal{C}^{\#}}
\indicator{\lambda\in\bar{\Gamma}}
\frac{1}{\abs{\bar{\Gamma}\cap\mathcal{L}^{\#}}}
\varphi(\gamma_{1},\dots,\gamma_{m})\prod_{i=1}^{m}w^{\#}_{z}(\gamma_{i}),
\end{equation}
where the series converges uniformly for $\abs{z}\ge z_{0}$ by \eqref{eqn:clusterexpansionboundaryterm}. 
Hence, to prove the analyticity of \eqref{eqn:pressurehighfugacityexpansion}, it suffices to check that the summands are analytic. 
Indeed, for each GFc $\gamma\in\mathcal{C}^{\#}$, the weight $w_{z}^{\#}(\gamma)$ is a rational function of $z$ (cf. \eqref{eqn:partitionfunctiondefinition} and \eqref{eqn:weightdefinition}) and bounded on $\abs{z}\ge z_{0}$ by \eqref{eqn:weightestimate}, hence analytic. 
The proof is now complete.
\end{proof}

\begin{proof}[Proof of Theorem \ref{thm:crystallization}]
First, we prove that the series
\begin{equation}\label{eqn:firstderivativestep1}
\sum_{m=1}^{\infty}\sum_{\gamma_{1}\in \mathcal{C}^{\#}}\dots\sum_{\gamma_{m}\in \mathcal{C}^{\#}}
\abs{\varphi(\gamma_{1},\dots,\gamma_{m})}
\abs{\evaluate{\frac{\partial}{\partial\log \mathbf{z}(x)}\prod_{i=1}^{m}w^{\#}_{\mathbf{z}}(\gamma_{i})}{z}}
\end{equation} 
converges for $z\ge z_{0}$.
Indeed, bounding \eqref{eqn:firstderivativestep1} by Line 2 of \eqref{eqn:derivativebeginning} with $\mathcal{C}^{\#}(\Lambda)$ replaced by $\mathcal{C}^{\#}$, the computations from \eqref{eqn:derivativebeginning} to \eqref{eqn:derivativeend} can be repeated with all instances of $\mathcal{C}^{\#}(\Lambda)$ replaced by $\mathcal{C}^{\#}$, which shows that \eqref{eqn:firstderivativestep1} is uniformly bounded by $1$.
Note that this justifies the interchange of differentiation and summation in 
\begin{equation}\label{eqn:firstderivativestep1.5}
\evaluate{\frac{\partial}{\partial\log \mathbf{z}(x)}\log\frac{\Xi^{\#}_{\mathbf{z}}(\Lambda)}{\mathbf{z}^{\#}(\Lambda)}}{z}
=\sum_{m=1}^{\infty}\sum_{\gamma_{1}\in \mathcal{C}^{\#}(\Lambda)}\dots\sum_{\gamma_{m}\in \mathcal{C}^{\#}(\Lambda)}\varphi(\gamma_{1},\dots,\gamma_{m})\evaluate{\frac{\partial}{\partial\log \mathbf{z}(x)}\prod_{i=1}^{m}w^{\#}_{\mathbf{z}}(\gamma_{i})}{z}.
\end{equation}

Second, we show that the RHS of \eqref{eqn:firstderivativestep1.5} converges to 
\begin{equation}\label{eqn:1pointhighfugacityexpansion}
\sum_{m=1}^{\infty}
\sum_{\gamma_{1}\in \mathcal{C}^{\#}}\dots\sum_{\gamma_{m}\in \mathcal{C}^{\#}}
\varphi(\gamma_{1},\dots,\gamma_{m})
\evaluate{\frac{\partial}{\partial\log \mathbf{z}(x)}\prod_{i=1}^{m}w^{\#}_{\mathbf{z}}(\gamma_{i})}{z}
\end{equation}
uniformly for $z\ge z_{0}$ in the limit $\Lambda\Uparrow\Lambda_{\infty}$. 
Indeed, their difference is bounded by
\begin{equation}\label{eqn:derivativestep2.1}
\sum_{m=1}^{\infty}\sum_{\gamma_{1}\in \mathcal{C}^{\#}}\dots\sum_{\gamma_{m}\in \mathcal{C}^{\#}}
\indicator{\Gamma\not\sqsubset\mathcal{C}^{\#}(\Lambda)}
\abs{\varphi(\gamma_{1},\dots,\gamma_{m})}
\abs{\evaluate{\frac{\partial}{\partial\log \mathbf{z}(x)}\prod_{i=1}^{m}w^{\#}_{\mathbf{z}}(\gamma_{i})}{z}}.
\end{equation}
Notice that the support of any tuple $\Gamma$ contributing to \eqref{eqn:derivativestep2.1} must intersect $\partial^{\text{in}}\Lambda$ and, by Line 2 of \eqref{eqn:derivativebeginning} and \eqref{eqn:firstderivativeindicatorbound}, enclose the point $x$ in the sense that $x\in\cup_{\gamma\in X'}\Int\gamma$. 
Hence, \eqref{eqn:derivativestep2.1} is bounded by
\begin{equation}\label{eqn:derivativestep2.2}
\sum_{m=1}^{\infty}\sum_{\gamma_{1}\in \mathcal{C}^{\#}}\dots\sum_{\gamma_{m}\in \mathcal{C}^{\#}}
\indicator{\operatorname{vol}(\Gamma)\ge\mathrm{d}_{\Lambda_{\infty}}(x,\partial^{\text{in}}\Lambda)}
\abs{\varphi(\gamma_{1},\dots,\gamma_{m})}
\abs{\evaluate{\frac{\partial}{\partial\log \mathbf{z}(x)}\prod_{i=1}^{m}w^{\#}_{\mathbf{z}}(\gamma_{i})}{z}},
\end{equation}
where we define the \emph{volume} of a tuple $\Gamma$ as $\operatorname{vol}(\Gamma):=\sum_{i}\abs{\bar{\gamma}_{i}}$.
Since $\mathrm{d}_{\Lambda_{\infty}}(x,\partial^{\text{in}}\Lambda)\rightarrow\infty$ as $\Lambda\Uparrow\Lambda_{\infty}$, by the first step, \eqref{eqn:derivativestep2.2} vanishes uniformly in the thermodynamic limit.

Finally, by the same argument as in the proof of Theorem \ref{thm:analyticity}, we conclude that \eqref{eqn:1pointhighfugacityexpansion} is (real) analytic in $z$ for $z\ge z_{0}$.
\end{proof}

\section{Examples}
\label{sec:examples}
Verifying Assumption \ref{assumption} boils down to computing and inspecting a set of \emph{local configurations} in the model that maximize the local density around a given particle.

\begin{definition}
  Without loss of generality, a {\it local configuration} is a configuration of particles that includes a particle at $\mathbf 0$.
  The {\it density} of a local configuration is the local density\-~\eqref{eqn:local_density} of the particle at $\mathbf 0$: $\rho_X(\mathbf 0)$.
  A {\it maximal-density local configuration} maximizes the local density at $\mathbf 0$.
\end{definition}

Below, we discuss two concrete models of interest: the 3- and 4-staircase model on $\Z^{2}$, and the disk model on $\Z^{2}$ of radius $5/2$.
The argument in both cases consists in computing the local configurations that maximize the local density, and extending them to close-packings on $\mathcal L_{\infty}$.
To prove that these are ground states, we use the following lemma.

\begin{lemma}\label{lemma:rholoc_bound_rho}
  If $\Lambda_{\infty}$ is a periodic graph, then
  \begin{equation}
    \rho_{\mathrm{max}}^{\mathrm{loc}}\ge
    \rho_{\mathrm{max}}
  \end{equation}
  (see \eqref{rhomax} and Definition \ref{def:local_density}).
\end{lemma}

\begin{proof}
  Given $X\in\Omega(\Lambda_{\infty})$ and $\Lambda\Subset\Lambda_{\infty}$ with $\Lambda\cap X\ne\emptyset$, we have, by Definition \ref{def:local_density},
  \begin{equation}
  \begin{multlined}
    (\rho_{\mathrm{max}}^{\mathrm{loc}})^{-1}
    \le
    \frac1{\abs{X\cap\Lambda}}\sum_{x\in X\cap\Lambda}\rho_X(x)^{-1}
    =
    \frac1{\abs{X\cap\Lambda}}\sum_{x\in X\cap\Lambda}\sum_{\lambda\in V_X(\sigma_x)}\frac1{\abs{\set{z\in X\mid\lambda\in V_X(\sigma_z)}}}
    \\
    =\frac1{\abs{X\cap\Lambda}}\sum_{\lambda\in \Lambda}\frac{\abs{\set{x\in X\cap\Lambda\mid\lambda\in V_X(\sigma_x)}}}{\abs{\set{z\in X\mid\lambda\in V_X(\sigma_z)}}}
    \le
    \frac{\abs{\Lambda}}{\abs{X\cap\Lambda}}.
    \end{multlined}
  \end{equation}
  Therefore,
  \begin{equation}
    (\rho_{\mathrm{max}}^{\mathrm{loc}})^{-1}
    \le\rho_{\mathrm{max}}^{-1}(\Lambda)
  \end{equation}
  after which we pass to the limit $\Lambda\Uparrow\Lambda_{\infty}$ (which exists because of the periodicity of $\Lambda_{\infty}$).
\end{proof}

\subsection{$3$- and $4$-staircases}
\label{subsec:n_staircases}

For $n\ge 3$, the $n$-staircase model on $\Z^{2}$ is defined as follows: the support of the particles is
\begin{equation}
  \omega_n:=\bigcup_{\displaystyle\mathop{\scriptstyle(x,y)\in\mathbb Z^{2}}_{x,y\ge0,\ x+y\le n-1}}\left(x-{\textstyle\frac12},x+{\textstyle\frac12}\right]\times\left(x-{\textstyle\frac12},x+{\textstyle\frac12}\right].
\end{equation}
See Figure \ref{fig:staircase}.

\begin{figure}
  \centering
  \begin{subfigure}{0.3\textwidth}
      \centering
      \includegraphics[width=1.6cm]{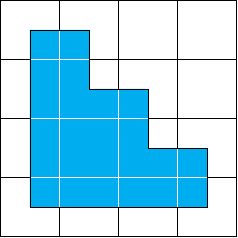}
  \end{subfigure}
  \begin{subfigure}{0.3\textwidth}
      \centering
      \includegraphics[width=2.0cm]{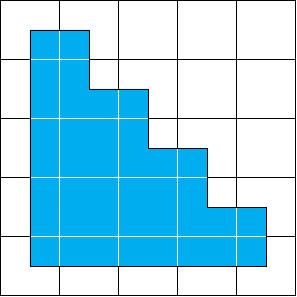}
  \end{subfigure}
  \caption{The 3- and 4-staircase.}
  \label{fig:staircase}
\end{figure}

We will prove that the $3$- and $4$-staircase models satisfy Assumption \ref{assumption}, and, therefore, crystallize at high fugacities.
The arguments can be extended to the general case of $n$-staircases.
In fact, the analog of Lemma \ref{lem:n_staircases_local_configurations} is proved in Appendix \ref{appx:n_staircases}.
Let us mention that in the general case, even and odd $n$ behave differently, which is why we discuss both examples $n=3,4$.

\begin{lemma}\label{lem:n_staircases_local_configurations}
If $n=4$, then the density of any local configuration $X$ is maximized if and only if
\begin{equation}
  \textstyle
  \set{\pm(2,2),\pm(-4,2),\pm(2,-4)}\subseteq X
  \label{n_staircase_local_even}
\end{equation}
and if $n=3$, the density is maximized if and only if
\begin{equation}
  \textstyle
  \set{\pm(2,1),\pm(-3,2),\pm(1,-3)}\subseteq X
  \quad\mathrm{or}\quad
  \set{\pm(1,2),\pm(-3,1),\pm(2,-3)}\subseteq X
  \label{n_staircase_local_odd}
\end{equation}
see Figure \ref{fig:staircase_close_packings}.
A configuration that does not include these has, for $n=3,4$,
\begin{equation}
  \rho_{X}^{-1}(\mathbf 0)\ge\rho_{\mathrm{max}}^{\mathrm{loc}}{}^{-1}+\epsilon_n
  ,\quad
  \epsilon_3=\frac13
  ,\quad
  \epsilon_4=\frac16
  .
  \label{epsilon_stairs}
\end{equation}
\end{lemma}

\begin{figure}[h]
	\centering
    \begin{subfigure}{0.4\textwidth}
        \centering
        \includegraphics[width=4cm]{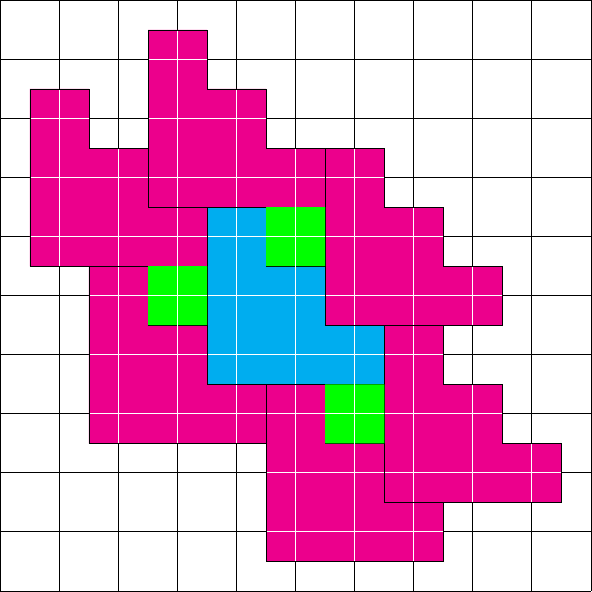}
        \caption{}
        \label{fig:staircase_close_packings_odd_n}
    \end{subfigure}
    \begin{subfigure}{0.4\textwidth}
        \centering
        \includegraphics[width=5.2cm]{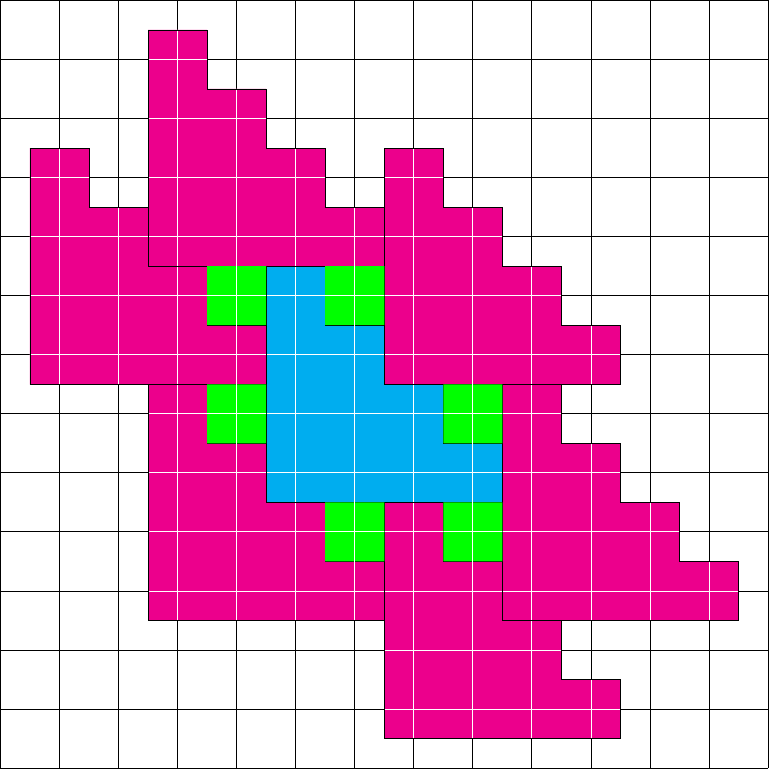}
        \caption{}
        \label{fig:staircase_close_packings_even_n}
    \end{subfigure}
    \captionsetup{font=small,width=0.8\textwidth,singlelinecheck=off}
    \caption[]{Maximal-density local configurations in the $3$- and $4$-staircase models on $\Z^{2}$.
    The Voronoi cell of the central (cyan) particle consists of the support of the particle along with the green sites.
    Each green site carries a weight of $\frac13$.
    \begin{enumerate}
    	\item If $n=3$, there are exactly two maximal-density local configurations; the one shown in Figure \ref{fig:staircase_close_packings_odd_n} and its reflection across the line $y=x$.
	The local density of the cyan particle is $\frac17$.
		\item If $n=4$, there is exactly one maximal-density local configuration as shown in Figure \ref{fig:staircase_close_packings_even_n}.
	The local density of the cyan particle is $\frac1{12}$.
    \end{enumerate}
    }
    \label{fig:staircase_close_packings}
\end{figure}

\begin{proof}
  First of all, we compute the local density for the configurations in Figure \ref{fig:staircase_close_packings}:
  \begin{equation}
  \rho_X(\mathbf 0)=\begin{cases}
      \frac17&\mathrm{for\ }n=3 \\
      \frac1{12}&\mathrm{for\ }n=4
      \end{cases}
      .
    \label{local_density_staircase}
  \end{equation}

  We seek to maximize the local density, in other words, to minimize
  \begin{equation}
  \rho_{X}^{-1}(\mathbf 0):=\sum_{\lambda\in V_{X}(\sigma_{\mathbf 0})}\frac{1}{\abs{\set{z\in X\mid\lambda\in V_{X}(\sigma_{z})}}}
  \end{equation}
  (see \eqref{eqn:local_density}).

  First, note that an uncovered site $\lambda$ that neighbors $\sigma_{\mathbf 0}$: $d_{\mathbb Z^2}(\lambda,\sigma_{\mathbf 0})=1$, can neighbor at most two other particles.
  Therefore, such a site can be in the Voronoi cell of at most 3 particles, so its contribution to $\rho_X^{-1}$ is $\ge\frac13$.

  Now, because of the hard-core repulsion, for $n=3$, one sees from Figure \ref{fig:3staircase_max} that at least 3 of the sites neighboring $\sigma_{\mathbf 0}$ must be left uncovered.
  Since each site contributes at least $\frac13$,
  \begin{equation}
    \rho_X^{-1}(\mathbf 0)\ge 7
    .
  \end{equation}
  For $n=4$, by Figure \ref{fig:4staircase_max}, at least 6 sites must be left uncovered.
  Since each site contributes at least $\frac13$,
  \begin{equation}
    \rho_X^{-1}(\mathbf 0)\ge 12
    .
  \end{equation}
  Thus, by \eqref{local_density_staircase}, the configurations \eqref{n_staircase_local_even}-\eqref{n_staircase_local_odd} maximize the local density.

  \begin{figure}[h]
	\centering
    \begin{subfigure}{0.4\textwidth}
        \centering
        \includegraphics[width=2.4cm]{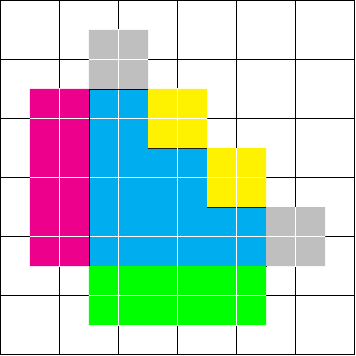}
        \caption{}
        \label{fig:3staircase_max}
    \end{subfigure}
    \begin{subfigure}{0.4\textwidth}
        \centering
        \includegraphics[width=2.8cm]{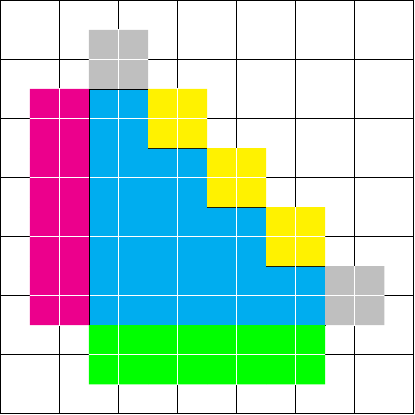}
        \caption{}
        \label{fig:4staircase_max}
    \end{subfigure}
    \caption{For $n=3$, at least one of the magenta, one of the green, and one of the yellow sites must be left uncovered.
    For $n=4$, at least {\it two} of the magenta, two of the green, and two of the yellow sites must be left uncovered.}
        \label{fig:staircase_max}
  \end{figure}

  Let us now check that these are the only ones.
  Consider $n=3$ first.
  (To follow this discussion, it may be helpful to draw the particles on a sheet of graph paper as they are added.)
  Choose which of the two yellow sites is covered.
  Without loss of generality, let us assume it is $(2,1)$.
  Since the gray sites are neither yellow, nor green, nor magenta, they must be covered to maximize the local density.
  Having placed a particle at $(2,1)$, the gray site $(3,0)$ can be covered in only one way.
  Having placed this latest particle, the green site $(2,-1)$ cannot be covered, so the other two green sites must be covered, which can only be done in one way.
  Once this is done, the magenta site $(-1,1)$ can no longer be covered, so the other magenta site $(-1,2)$ must be covered, which can only be done in one way.
  This then leaves a unique way of covering the remaining gray site $(0,3)$.
  Thus, the maximal-density local configuration is unique, once we have chosen which yellow site is to be covered.

  The argument for $n=4$ is similar.
  (Again, it is recommended to follow along with graphing paper.)
  Let us first try to place a particle in one of the yellow sites that is not $(2,2)$, say $(3,1)$.
  In this case, the other two yellow sites $(2,2)$ and $(1,3)$ are left unoccupied.
  However, $(2,3)$ will also be left unoccupied, so $(2,2)$ and $(1,3)$ will each only neighbor two particles, so their contribution to $\rho^{-1}$ will be at least $\frac12$ each, so the density will not be minimal.
  Therefore, the only yellow site that can be occupied in the maximal density configuration is $(2,2)$.
  Having fixed $(2,2)$, there are two ways of covering the gray site $(4,0)$, but one of these will leave the yellow site $(3,1)$ with only two neighbors, and will thus not be the maximal density.
  There is then just one possibility left to cover the gray site $(4,0)$.
  Having placed this particle, the green site $(3,-1)$ must be left uncovered.
  If we tried to cover the green site $(1,-1)$, then $(3,-1)$ would only neighbor two particles, and the density would not be maximal.
  Therefore, the green site $(1,-1)$ must be left uncovered, and so the green site $(2,-1)$ must be covered, which can only be done in one way.
  We repeat this argument to cover the gray site at $(0,4)$ and the magenta site $(-1,2)$.
  At this point, there is a unique way of covering the remaining sites at $(-1,0)$ and $(0,-1)$.

  Thus, for $n=4$, there is a unique way to maximize the local density.

  Finally, let us estimate $\epsilon$.
  For $n=3$, the only way to deviate from the construction above is if one had an extra empty site among the magenta, green, and yellow.
  This would increase $\rho_X^{-1}$ by $\ge\frac13$.
  (In fact, there exists a local configuration with $\rho_X^{-1}=7+\frac13$, (obtained by placing a particle at $(2,2)$) so this is optimal.)

  For $n=4$, one can deviate from the construction above in two ways: by adding an extra site, which would increase $\rho_X^{-1}$ by $\ge\frac13$, or by having one empty site neighbor $2$ particles instead of $3$, which would increase $\rho_X^{-1}$ by $\ge \frac12-\frac13=\frac16$.
  This is presumably not an optimal estimate.

\end{proof}

Using Lemma \ref{lem:n_staircases_local_configurations}, we construct configurations that have a constant maximal local density by extending the local configurations in Figure \ref{fig:staircase_close_packings}.
To do so, we apply Lemma \ref{lem:n_staircases_local_configurations} to a neighbor of $\mathbf 0$ and repeat.

We now verify Assumption \ref{assumption}.
\begin{enumerate}
  \item[\ref{asm:lattice}.] $\mathbb Z^2$ is a periodic graph with coordination number $4$ and it is easy to check (and this is also well known \cite{Timar2013}), the boundary of any connected set is $2$-connected.

  \item[\ref{asm:finitely_many_ground_states}.] By considering the possible translations of the packings in Lemma \ref{lem:n_staircases_local_configurations}, we find that there are $6\times2=12$ ground states for $n=3$ and $10$ ground states for $n=4$.

  \item[\ref{asm:isometry}.] If $n=4$, the ground states are related by translations.
  If $n=3$, they are related to each other by translations and the reflection $(x,y)\mapsto(y,x)$.
  These preserve the shapes of the particles.

  \item[\ref{asm:density_max_local}.] By Lemma \ref{lemma:rholoc_bound_rho}, $\rho_{\mathrm{max}}$ is smaller or equal to the local density of the configurations in Figure \ref{fig:staircase_close_packings} (which is the maximal local density).
  Since those can be extended to a configuration on all of $\Lambda_\infty$, $\rho_{\mathrm{max}}=\rho_{\mathrm{max}}^{\mathrm{loc}}$.

  \item[\ref{asm:imperfect_transition}.] Without loss of generality, let us consider a $\#$-correct particle at $\vec{0}$, which, by Lemma \ref{lem:n_staircases_local_configurations}, must then be in one of the local configurations (\ref{n_staircase_local_even})-(\ref{n_staircase_local_odd}).
  By Figure \ref{fig:staircase_close_packings}, no site at $\mathrm{d}_{\Lambda_{\infty}}=1$ from $\sigma_{\vec{0}}$ can neighbor any other particle than those already appearing in the local configuration. 
  In other words, the neighbors of $\vec{0}$ are precisely those specified in the local configuration.
  Now, if any neighbor of $\vec{0}$ is correct, it must too be in a local configuration specified in Lemma \ref{lem:n_staircases_local_configurations}, but it is straightforward to check that the only way this does not cause particles to overlap is if it is in the same local configuration as is $\vec{0}$, so they must be $\#$-correct for the same $\#\in\mathcal{G}$.

  \item[\ref{asm:density_local_density}.] By Lemma \ref{lem:n_staircases_local_configurations}, if a particle $x$ is incorrect, then
  \begin{equation}
    \rho_{X}^{-1}(\mathbf 0)\ge\rho_{\mathrm{max}}^{\mathrm{loc}}{}^{-1}+\epsilon_n
  \end{equation}
  where $\epsilon_n$ was defined in \eqref{epsilon_stairs}.
  Thus, Item \ref{asm:density_local_density} holds with $\mathcal R_1=\mathcal S_1=0$.
\end{enumerate}

Thus, for this model, Theorems \ref{thm:crystallization} and \ref{thm:analyticity} hold for $\abs{z}\ge z_0$.
We can get an explicit value for $z_0$ from Appendix \ref{appx:summary}; see (\ref{estimate_z}):
for $\Lambda_\infty=\mathbb Z^2$,
\begin{equation}
  \chi=4
  ,\quad
  I_d=\frac1{16}
\end{equation}
For $n=3$,
\begin{equation}
  \abs{\mathcal{G}}=12
  ,\quad
  \tau=5+\log2433024
  ,\quad
  \mu=\frac13
  ,\quad
  \varsigma=14
\end{equation}
\begin{equation}
  \mathcal R_0=2
  ,\quad
  \mathcal R_1=\mathcal S_1=0
  ,\quad
  \mathcal R_2=3
  ,\quad
  r_{\mathrm{eff}}=3
  ,\quad
  \mathcal S_0=15
\end{equation}
and since the set $\set{\lambda\in\mathbb Z^2\mid d_{\mathbb Z^2}(\lambda,0)\le \mathcal S_0}$ has area
\begin{equation}
  2\mathcal S_0^2+2\mathcal S_0-1
\end{equation}
the maximal number of particles in this set is bounded by
\begin{equation}
  \mathcal N\le\frac{2\mathcal S_0^2+2\mathcal S_0-1}{\abs{\sigma_x}}
  =\frac{479}6\le80
\end{equation}
\begin{equation}
  \rho_{\mathrm{max}}=\frac17
  ,\quad
  \rho_0\le\frac{479}{3355}
\end{equation}
and so
\begin{equation}
  z_0\ge\exp\left(\frac{70455}2(61+\log(2433024))\right)
  \approx\exp(2.67\times10^{6})
\end{equation}

For $n=4$,
\begin{equation}
  \abs{\mathcal{G}}=10
  ,\quad
  \tau=5+\log\frac{3520000}3
  ,\quad
  \mu=\frac13
  ,\quad
  \varsigma=14
\end{equation}
\begin{equation}
  \mathcal R_0=2
  ,\quad
  \mathcal R_1=\mathcal S_1=0
  ,\quad
  \mathcal R_2=4
  ,\quad
  r_{\mathrm{eff}}=4
  ,\quad
  \mathcal S_0=20
\end{equation}
\begin{equation}
  \mathcal N\le
  \frac{839}{10}
\end{equation}
\begin{equation}
  \rho_{\mathrm{max}}=\frac1{12}
  ,\quad
  \rho_0\le\frac{2517}{30209}
\end{equation}
so
\begin{equation}
  z_0\ge\exp\left(\frac{1087524}5\left(61+\log\left(\frac{3520000}3\right)\right)\right)
  \approx
  \exp(1.63\times10^{7})
  .
\end{equation}

Obviously, these values of $z_0$ are {\it far} from optimal, and could be improved with relatively little work.
Nevertheless, it is worth pointing out that our construction gives explicit values.

\subsection{Disks of radius $5/2$}

The hard-disk model on $\Z^{2}$ of radius $5/2$ has particles whose support is
\begin{equation}
  \omega:=\set{(x,y)\in\R^{2}\mid \sqrt{x^{2}+y^{2}}<5/2}.
\end{equation}
This model is the {\it 12th nearest neighbor exclusion}, and is equivalent to the {\it hard octagon model}; see Figure \ref{fig:octagon}.

\begin{figure}
  \hfil\includegraphics[width=2.4cm]{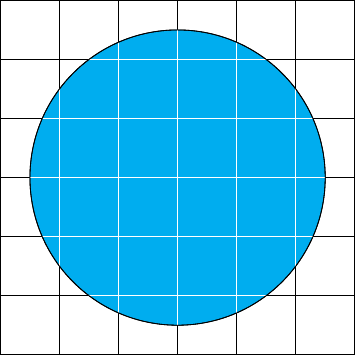}
  \hfil\includegraphics[width=2.4cm]{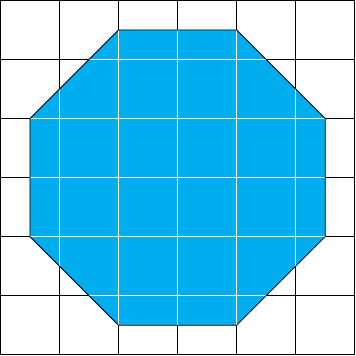}

  \caption{The hard-core model based on the disk of radius 2.5 is equivalent to that with hard-core octagons.}
  \label{fig:octagon}
\end{figure}

\begin{lemma}\label{lem:disk_local_configurations}
The density of a local configuration is maximal if and only if one of the following holds (see Figure \ref{fig:disk_close_packings}):
\begin{enumerate}
  \item
  \begin{equation}
    \Pi_1\Pi_2\set{(2,5),(-3,4),(-5,-1),(-2,-5),(3,-4),(5,1)}
    \subset X
  \end{equation}
  where $\Pi_1$ is one of the following operators: the identity, or the horizontal, or vertical reflection and $\Pi_2$ is one of the following operators: the identity, or the rotation by one of $\pi/2,\pi,-\pi/2$.

  \item
  \begin{equation}
    \Pi\set{(0,5),(-5,1),(-3,-4),(3,-4),(5,1)}
    \subset X
  \end{equation}
  where $\Pi$ is one of the following operators: the identity, or the rotation by one of $\pi/2,\pi,-\pi/2$.
\end{enumerate}
\end{lemma}

\begin{proof}
  We will prove this lemma by first reducing it to a finite (and somewhat small)  number of cases, and will then leave it up to the reader to check each case.

  First of all, the inverse of the local density of the configurations in Figure \ref{fig:disk_close_packings} is 23, that is, the 21 sites covered by the particle at $\mathbf 0$ plus 2.
  Thus, any maximal-density local configuration must leave empty sites near the particle at $\mathbf 0$ whose weight totals at most 2.

  Let us now emphasize a few properties.
  \begin{enumerate}
    \item\label{16sites} There are 16 sites that neighbor the particle at $\mathbf 0$ (that is, that are at distance 1 from $\sigma_{\mathbf 0}$).
    See Figure \ref{fig:disk_neighbors}.

    \item\label{cover3} A particle can cover at most 3 of these neighboring sites, and each empty site among these neighbors has a weight that is at least $1/3$ (since it can neighbor at most three particles).

    \item\label{leave2} Whenever a particle covers at least one neighbor of $\sigma_{\mathbf 0}$, there are at least 2 sites that must be left empty, and these sites neighbor $\sigma_{\mathbf 0}$ and either neighbor the particle or are at $(\pm1,\pm1)$ from the particle.
    If the particle covers 3 sites, then the two sites left uncovered have weight at least $1/2$.
    See Figure \ref{fig:disk_max}.
  \end{enumerate}

  \begin{figure}
    \centering
    \begin{subfigure}{0.3\textwidth}
      \hfil\includegraphics[width=3.2cm]{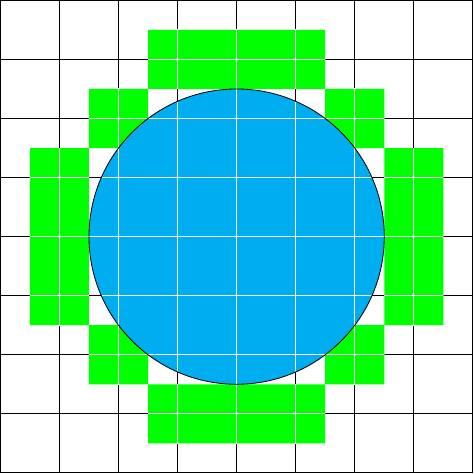}
      \caption{}
      \label{fig:disk_neighbors}
    \end{subfigure}
    \begin{subfigure}{0.3\textwidth}
      \hfil\includegraphics[width=2.4cm]{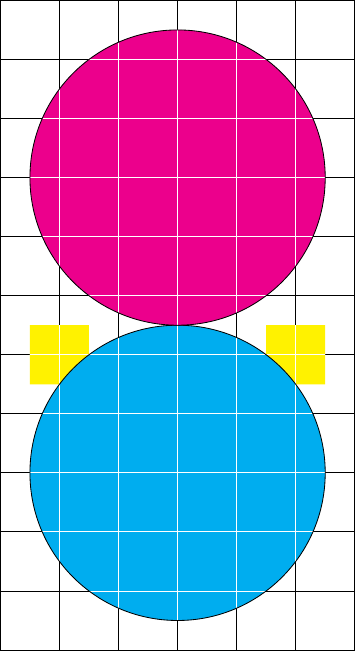}
      \caption{}
      \label{fig:disk_max_1}
    \end{subfigure}
    \begin{subfigure}{0.3\textwidth}
      \hfil\includegraphics[width=2.8cm]{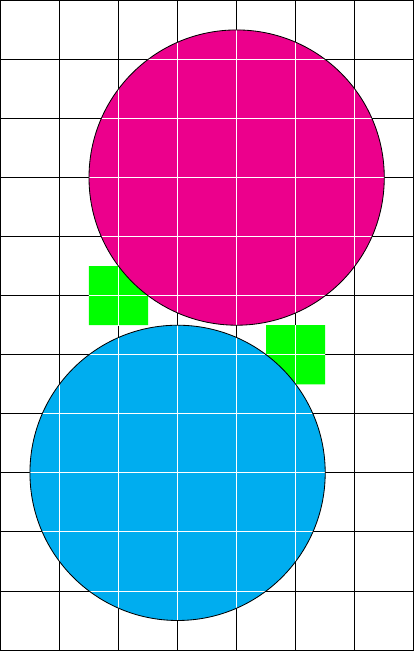}
      \caption{}
      \label{fig:disk_max_2}
    \end{subfigure}
    \caption{\ref{fig:disk_neighbors}: the 16 neighbors of the disk.
    \ref{fig:disk_max_1}: the magenta particle covers 3 neighbors of the cyan particle, and the two yellow sites must be left uncovered.
    In addition, the yellow sites cannot neighbor more than 2 particles, so their weight is at least $\frac12$.
    \ref{fig:disk_max_2}: the magenta particles cover 2 neighbors of the cyan particle, and the two green sites must be left uncovered.
    In this case, the green sites can neighbor up to 3 particles, so the weight is at least $\frac13$.}
    \label{fig:disk_max}
  \end{figure}

  By Items \ref{16sites} and \ref{cover3}, a close-packing configuration must involve at least 4 particles other than the one at $\mathbf 0$, as $(16-3\times3)/3>2$.

  Similarly, for a configuration with 4 particles other than the one at $\mathbf 0$ to be a close-packing, it must have at least 2 particles in a position in which they cover 3 neighbors, as $(16-2\times3-2\times2)/3=2$.
  Now, the only positions in which a particle covers 3 neighbors are $(0,5)$, $(0,-5)$, $(5,0)$ and $(-5,0)$.
  Let us first attempt to place a particle at $(0,5)$ and one at $(5,0)$: this leaves a site uncovered whose weight is 1, and at least two more with weight $1/2$.
  However, adding an extra particle will generate yet another empty site, so the density cannot be maximal.
  Let us now try to place a particle at $(0,5)$ and one at $(0,-5)$: by Item \ref{leave2}, this leaves four sites uncovered that each have a weight of at least $1/2$, so their total weight is 2.
  Because any extra particle that is placed will leave extra sites uncovered, this cannot be a maximal density configuration.

  Thus, maximal density configurations have to have at least 5 particles in addition to the one at $\mathbf 0$.
  The argument above still holds that we cannot have 2 particles covering 3 sites, so there must be at most 1 particle that covers 3 sites, and the others cover at most 2 sites.

  Let us first attempt the case in which 1 particle covers 3 sites (without loss of generality, assume it is at $(0,5)$ and use the symmetry of the model), and the others cover at most 2.
  By Item \ref{leave2}, this particle leaves two sites uncovered with a weight of at least $1/2$ each.
  By Item \ref{leave2} again, if there were 6 particles, there would also be at least 6 empty sites, and if two of these have weight $1/2$, the total weight of the empty sites is at least $7/3>2$, and so the density would not be maximal.
  Thus, when 1 particle covers 3 sites, we can only have 5 particles that neighbor $\sigma_{\mathbf 0}$.
  By Item \ref{leave2}, the total number of empty sites is at least 5, two of which have a weight $1/2$, so for the density to be maximal, the remaining 3 must have weight $1/3$, and so must neighbor 3 particles.
  At this stage, there are not very many possibilities: we have 5 particles other than that at $\mathbf 0$, one of which is at $(0,5)$, the other 4 must cover two neighbors of $\sigma_x$, and three of the empty sites that neighbor $\sigma_{\mathbf 0}$ must all neighbor three particles each.
  It is straightforward to check that there is only one way to do this, which is the one in Figure \ref{fig:disk_close_packing_2}.

  Now, if no particle covers 3 sites, then all cover at most 2.
  Thus, if there are 5 particles other than the one at $\mathbf 0$, then there must be at least $6$ empty sites that neighbor $\sigma_{\mathbf 0}$.
  If there are 6 particles, then by Item \ref{leave2}, there are at least 6 empty sites that neighbor $\sigma_{\mathbf 0}$ as well.
  Therefore, in either case, there must be exactly 6 empty sites (if there were more, then their total weight would be $>2$ and the density would not be maximal), whose weight is exactly $1/3$ each.
  By symmetry, we can assume without loss of generality that two of the sites $\set{(0,3),(1,3),(-1,3)}$ are covered (since there are $6$ vacancies, at least one of the 4 possible rotations of $\set{(0,3),(1,3),(-1,3)}$ must have two sites covered).
  There are 5 ways to cover two of those sites.
  \begin{itemize}
  \item
    First, consider placing particles at $(3,4)$ and $(-3,4)$.
    Following the prescription that every empty site must neighbor 3 particles, there are few possibilities for what goes next.
    In fact, one easily checks that there are no possible local configurations that follow the prescription.

  \item
    Next, consider placing particles at $(2,5)$ and $(-3,4)$.
    One then checks that there is only one way to follow the prescription: the one in Figure \ref{fig:disk_close_packing_1}.

  \item
    Otherwise, one can place the particles at $(-2,5)$ and $(3,4)$, which is the symmetric of the previous case.

  \item
    Next, one could place a particle at $(1,5)$.
    Following the prescription we find only one possibility: the $x\leftrightarrow y$ symmetric version of Figure \ref{fig:disk_close_packing_1} (that is, a $\pi/2$ rotation followed by a horizontal reflection).

  \item
    Finally, one could place a particle at $(-1,5)$, which is symmetric to the previous point.
  \end{itemize}
\end{proof}

Using Lemma \ref{lem:disk_local_configurations}, we construct configurations that have a constant maximal local density by extending those in Figure \ref{fig:disk_close_packings}.
We first note that the configuration in Figure \ref{fig:disk_close_packing_2} cannot be extended: consider the particle at $(3,-4)$, if its local density is to be maximal, then it must be part of a local configuration of the form of Figure \ref{fig:disk_close_packings} and its reflections and rotations, but because of the presence of the particle at $(-3,-4)$, it is clear that this cannot be the case.
Thus, Figure \ref{fig:disk_close_packing_2} cannot lead to a close-packing configuration.

On the other hand Figure \ref{fig:disk_close_packing_1} can be extended: each pink particle in the figure looks locally like the cyan one.
The completion of the close-packing is unique.
Indeed, consider the particle at $(2,5)$.
Because of the presence of the particle at $(5,1)$, it can only have a maximal density if it is surrounded in the same way as the cyan particle in Figure \ref{fig:disk_close_packing_1}, or its rotation by $\pi$.
But since Figure \ref{fig:disk_close_packing_1} is symmetric under rotations by $\pi$, this construction is unique.
We can make the same argument for each of the particles surrounding $\sigma_{\mathbf 0}$, which makes the extension of the figure unique.

Therefore, taking into account the horizontal and vertical reflections, as well as the $\pi/2$ rotation, there are 6 minimal density local configurations that can be extended to distinct close-packings.
Taking into account the translations of the particle at $\mathbf 0$, this yields a total of $18\times6=108$ close-packing configurations.

We now verify Assumption \ref{assumption}.
Items \ref{asm:lattice}, \ref{asm:density_max_local}, and \ref{asm:imperfect_transition} are the same arguments as for the staircases; see above.

\begin{enumerate}
  \item[\ref{asm:finitely_many_ground_states}.] By considering the possible symmetries of the packings in Lemma \ref{lem:disk_local_configurations}, we find that there are $18\times6=108$ ground states.

  \item[\ref{asm:isometry}.] The ground states are related by translations, rotations, and reflections.
  These preserve the shapes of the particles.

  \item[\ref{asm:density_local_density}.] We take $\mathcal R_1=0$ and consider a particle $x$ that is incorrect (and thus also $0$-incorrect).
  By the construction in the proof of Lemma \ref{lem:disk_local_configurations}, there are two ways the particle at $x$ can be incorrect.
  The first is that the local configuration around $x$ is not like those in Figure \ref{fig:disk_close_packings}, in which case
  \begin{equation}
    \rho_{X}^{-1}(x)\ge\rho_{\mathrm{max}}^{\mathrm{loc}}{}^{-1}+\frac16
  \end{equation}
  where $\frac16=\frac12-\frac13$ is the smallest possible defect contribution: instead of an empty site being surrounded by 3 particles, it is surrounded by 2.
  The second is that the local configuration around $x$ is Like those in Figure \ref{fig:disk_close_packing_2}, in which case, as was argued above, the particle at $x+(3,-4)$ cannot have a maximal density, so that particle has a density that is
  \begin{equation}
    \rho_{X}^{-1}(x+(3,-4))\ge\rho_{\mathrm{max}}^{\mathrm{loc}}{}^{-1}+\frac16
    .
  \end{equation}
  Thus, Item \ref{asm:density_local_density} holds with $\mathcal R_1=0$, $\mathcal S_1=d_{\mathbb Z^2}(\mathbf 0,(3,-4))=7$ and $\epsilon=\frac16$.
\end{enumerate}

Similarly to the staircases, we can, in principle, compute $z_0$ for this model.
The details are omitted here, as the computation is very similar to the staircases.

\begin{figure}[h]
	\centering
    \begin{subfigure}[t]{0.44\textwidth}
        \centering
        \includegraphics[width=6.4cm]{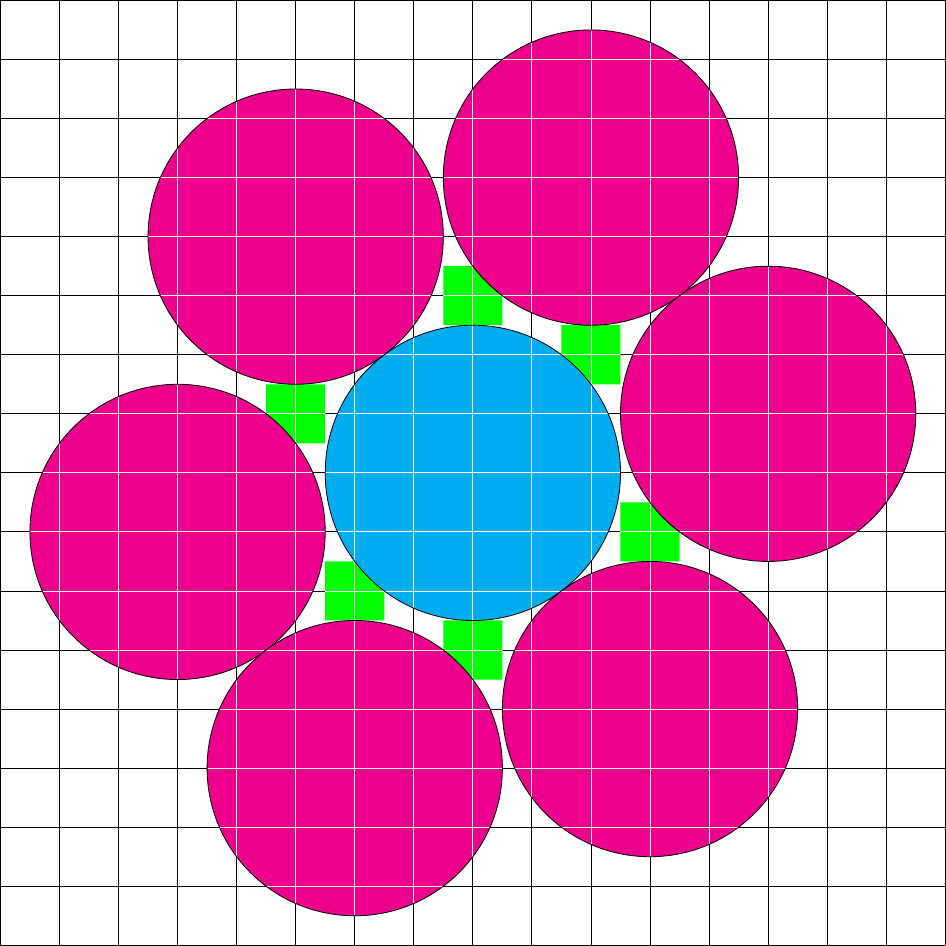}
        \caption{}
        \label{fig:disk_close_packing_1}
    \end{subfigure}
    \begin{subfigure}[t]{0.44\textwidth}
        \centering
        \includegraphics[width=6.4cm]{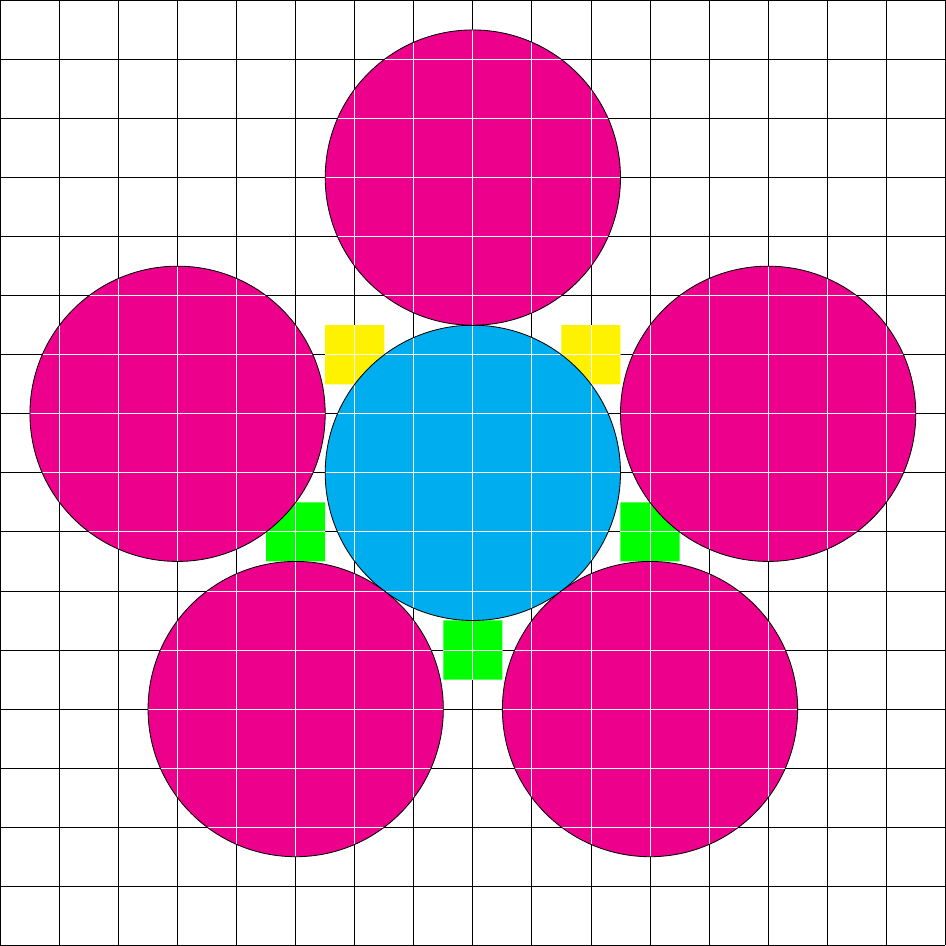}
        \caption{}
        \label{fig:disk_close_packing_2}
    \end{subfigure}
    \captionsetup{font=small,width=0.8\textwidth}
    \caption{Two possible maximal-density local configurations for the hard-disk model.
    The Voronoi cell of the central (cyan) particle consists of the support of the particle along with the green and yellow sites.
    The green sites have a weight $\frac13$ and the yellow sites $\frac12$.
    Thus, since the particle itself covers 21 sites, the local density in these configurations is $\frac1{23}$.}
    \label{fig:disk_close_packings}
\end{figure}

\newpage
\bibliographystyle{plain}
\bibliography{Bibliography.bib}

\newpage

\appendix

\section{Proof of Proposition \ref{prop:canonical_configuration_gfc_correspondence}}
\label{appx:canonical_configuration_gfc_correspondence}

The proof of Proposition \ref{prop:canonical_configuration_gfc_correspondence} requires several preliminary results.
Although their statements are intuitively obvious, the proofs are quite long and tedious.
Nevertheless, we present them in full detail in this appendix.

First, we prove that any particle in $X_{\gamma}$ has the same neighbors in any configuration having $\gamma$ as a GFc as in the canonical configuration $\xi_{\gamma}$.

\begin{lemma} \label{lem:GFc_particles_same_neighbors}
If a configuration $X$ has $\gamma$ as a GFc and $x\in X_{\gamma}$, then $\mathcal{N}_{X}(x)=\mathcal{N}_{\xi_{\gamma}}(x)$.
\end{lemma}

\begin{proof}

The GFc $\gamma$ is fully surrounded by $\mathcal R_2$-correct particles, in the sense that every particle $z\in X\setminus X_{\gamma}$ with $\mathrm{d}_{\Lambda_{\infty}}(V_{X}(\sigma_{z}),\bar{\gamma})\le 1$ is $(\#,\mathcal R_2)$-correct in $X$, where $\#\in\mathcal{G}$ is given by the labeling function $\mu_{\gamma}$ applied to the connected component of $\bar{\gamma}^{c}$ to which $z$ belongs.
We divide into two cases.
\begin{enumerate}

\item $x$ lies on the surface of the GFc: $\mathrm{d}_{\Lambda_{\infty}}(V_{X}(x),\bar{\gamma}^{c})=1$.
Since $\gamma$ is fully surrounded by $\mathcal R_2$-correct particles, by Definition \ref{def:R_correct}, $x$ is $\#$-correct in $X$ for some $\#$.
Thus, its neighbors in $X$ coincide with its $\#$-neighbors, which are either located in $X_{\gamma}$ or immediately surrounding $\gamma$.
Either way, they are contained in $\xi_{\gamma}$.
Moreover, all the other particles in $\mathcal{L}^{\#}\cap\xi_{\gamma}$ are strictly farther from $\sigma_{x}$, so none can neighbor $x$ in $\xi_{\gamma}$.
Finally, notice that the particles in $\xi_{\gamma}$ from holes with a different label than $\#$ cannot neighbor $x$ in $\xi_{\gamma}$, since differently labeled holes are well-separated by construction.

\item $x$ lies deep inside the GFc: $\mathrm{d}_{\Lambda_{\infty}}(V_{X}(x),\bar{\gamma}^{c})>1$.
Since each point on the interior boundary of $\bar{\gamma}$ is contained in $V_{X}(y)$ for some $y\in X_{\gamma}$, each such point is strictly closer to $\sigma_{y}$ than to $\sigma_{x}$.
Hence, a particle outside of $\bar{\gamma}$ cannot possibly be a neighbor of $x$.\qedhere

\end{enumerate}

\end{proof}

Second, we prove that any particle in $\xi_{\gamma}\setminus X_{\gamma}$ has the same neighbors in $\xi_{\gamma}$ and $\mathcal{L}^{\#}$, where $\#$ is the label that $\mu_{\gamma}$ assigns to the hole containing the particle.

\begin{lemma}\label{lem:canonical_non_GFc_particles_same_neighbors}
For each $x\in\xi_{\gamma}\setminus X_{\gamma}$, $\mathcal{N}_{\xi_{\gamma}}(x)=\mathcal{N}_{\mathcal{L}^{\#}}(x)$, where $\#$ is the label that $\mu_{\gamma}$ assigns to the hole containing $x$.
\end{lemma}

\begin{proof}
We divide into two cases.
\begin{enumerate}

\item $x$ lies on the surface of the hole: $\mathrm{d}_{\Lambda_{\infty}}(V_{\mathcal{L}^{\#}}(x),\bar{\gamma})\le 1$.
Consider any configuration $X$ containing $\gamma$ as a GFc.
As a particle immediately surrounding $\gamma$, $x$ is $(\#,\mathcal R_2)$-correct in $X$, so, in particular, the neighbors of $x$ in $X$ coincide with its $\#$-neighbors, which lie either in $X_{\gamma}$ or inside a hole of $\gamma$ with label $\#$, hence all contained in $\xi_{\gamma}$.
Again, the particles in holes with a different label than $\#$ cannot neighbor $x$ in $\xi_{\gamma}$.

\item $x$ lies deep inside the hole: $\mathrm{d}_{\Lambda_{\infty}}(V_{\mathcal{L}^{\#}}(x),\bar{\gamma})>1$.
Consider the neighbors of $x$ in $\xi_{\gamma}$.
By the assumption on $x$, none is in $X_{\gamma}$, so all have to be from within the hole.
However, the latter is filled with particles in $\mathcal{L}^{\#}$. \qedhere

\end{enumerate}
\end{proof}

Third, we prove that all the particles in $\xi_{\gamma}\setminus X_{\gamma}$ are $\mathcal R_2$-correct in $\xi_{\gamma}$.

\begin{lemma} \label{lem:canonical_non_GFc_particles_R_correct}
Every particle $x\in\xi_{\gamma}\setminus X_{\gamma}$ is $\mathcal R_2$-correct in $\xi_{\gamma}$.
\end{lemma}

\begin{proof}
By Lemma \ref{lem:canonical_non_GFc_particles_same_neighbors}, every particle $x\in\xi_{\gamma}\setminus X_{\gamma}$ is correct in $\xi_{\gamma}$, so we only need to check that all these particles have the right $\mathcal R_2$-neighbors.
Let $x\in\xi_{\gamma}\setminus X_{\gamma}$ be from a hole of $\gamma$ with label $\#$.
We divide into two cases.
\begin{enumerate}

\item $x$ lies on the surface of the hole: $\mathrm{d}_{\Lambda_{\infty}}(V_{\mathcal{L}^{\#}}(x),\bar{\gamma})\le 1$.
Consider any configuration $X$ containing $\gamma$ as a GFc.
As a particle immediately surrounding $\gamma$, $x$ is $(\#,\mathcal R_2)$-correct in $X$, so its $\mathcal R_2$-neighbors in $X$ coincide with its $\mathcal R_2$-neighbors in $\mathcal{L}^{\#}$, which are located either in $X_{\gamma}$ or in the holes of $\gamma$ with label $\#$.
Notice that all of these particles are in $\xi_{\gamma}$.
Moreover, since they are all $\#$-correct in $X$, by Lemmas \ref{lem:GFc_particles_same_neighbors} and \ref{lem:canonical_non_GFc_particles_same_neighbors}, their Voronoi cells in $X$ coincide with their Voronoi cells in $\xi_{\gamma}$, so they are all $\mathcal R_2$-neighbors of $x$ in $\xi_{\gamma}$.
By the same argument as before, $x$ has no additional $\mathcal R_2$-neighbors in $\xi_{\gamma}$.

\item $x$ lies deep inside the hole: $\mathrm{d}_{\Lambda_{\infty}}(V_{\mathcal{L}^{\#}}(x),\bar{\gamma})>1$.
Consider an $\mathcal R_2$-neighbor $y$ of $x$ in $\xi_{\gamma}$. 
It is located either in a hole of $\gamma$ labeled $\#$ or in $X_{\gamma}$.
In the former case, $y$ has the same Voronoi cell in $\xi_{\gamma}$ and $\mathcal{L}^{\#}$ by Lemma \ref{lem:canonical_non_GFc_particles_same_neighbors}.
In the latter case, $y$ is an $\mathcal R_2$-neighbor of a particle in the same hole as $x$ which surrounds $\gamma$.
Using Lemma \ref{lem:GFc_particles_same_neighbors}, we see that, again, $y$ has the same Voronoi cell in $\xi_{\gamma}$ and $\mathcal{L}^{\#}$.
Therefore, $y$ is an $\mathcal R_2$-neighbor of $x$ in $\mathcal{L}^{\#}$.

Conversely, let $z$ be an $\mathcal R_2$-neighbor of $x$ in $\mathcal{L}^{\#}$.
It is located either in a hole of $\gamma$ labeled $\#$ or in $\bar{\gamma}$.
In the former case, certainly $z\in\xi_{\gamma}$, and $z$ is an $\mathcal R_2$-neighbor of $x$ in $\xi_{\gamma}$ by Lemma \ref{lem:canonical_non_GFc_particles_same_neighbors}.
In the latter case, the same consideration as before shows that $z\in X_{\gamma}$ and has the same Voronoi cell in $\xi_{\gamma}$ and $\mathcal{L}^{\#}$, hence again an $\mathcal R_2$-neighbor of $x$ in $\xi_{\gamma}$.
\qedhere

\end{enumerate}
\end{proof}

Fourth, we prove that all the particles in $X_{\gamma}$ remain $\mathcal R_2$-incorrect in $\xi_{\gamma}$.

\begin{lemma} \label{lem:GFc_particles_R_incorrect}
Every particle $x\in X_{\gamma}$ is $\mathcal R_2$-incorrect in $\xi_{\gamma}$.
\end{lemma}

\begin{proof}
By contradiction, suppose that $x\in X_{\gamma}$ is $(\#,\mathcal R_2)$-correct in $\xi_{\gamma}$.
Consider a configuration $X$ containing $\gamma$ as a GFc.
By Lemma \ref{lem:GFc_particles_same_neighbors}, $x$ has the same neighbors in $\xi_{\gamma}$ and $X$, hence $\#$-correct in $X$.
We will take one step further and show that $x$ is $(\#,\mathcal R_2)$-correct in $X$, which will contradict the assumption that $\gamma$ is a GFc of $X$.

First, suppose that $y$ is an $\mathcal R_2$-neighbor of $x$ in $X$. 
There are three cases.
\begin{enumerate}

\item $y\in X_{\gamma}$.
By Lemma \ref{lem:GFc_particles_same_neighbors}, $y$ has the same Voronoi cell in $X$ and $\xi_{\gamma}$, hence an $\mathcal R_2$-neighbor of $x$ in $\xi_{\gamma}$.
Since $x$ is $(\#,\mathcal R_2)$-correct in $\xi_{\gamma}$, this implies that $y$ is $\#$-correct in $\xi_{\gamma}$. 
Again, by Lemma \ref{lem:GFc_particles_same_neighbors}, $y$ is also $\#$-correct in $X$.
In particular, $y$ has the same Voronoi cell in $X$ and $\mathcal{L}^{\#}$, hence an $\mathcal R_2$-neighbor of $x$ in $\mathcal{L}^{\#}$.

\item $y\in X_{\gamma'}$, where $\gamma'$ is another GFc of $X$.
By construction, $\gamma$ and $\gamma'$ are disconnected.
Hence, joining $V_{X}(x)$ and $V_{X}(y)$ by a path in $\Lambda_{\infty}$ with minimum length, the path necessarily contains a point that does not lie in the support of any GFc of $X$.
This point must then be contained in the Voronoi cell of a $(\#',\mathcal R_2)$-correct particle $z$ in $X$, so that $\mathrm{d}_{\Lambda_{\infty}}(V_{X}(x),V_{X}(z))\le\mathcal R_2$ and $\mathrm{d}_{\Lambda_{\infty}}(V_{X}(y),V_{X}(z))\le\mathcal R_2$.
Hence, $x$ and $y$ are both $\#'$-correct in $X$, which forces $\#'=\#$ because $x$ is also $\#$-correct.

\item $y\not\in X_{\gamma'}$ for any GFc $\gamma'$ of $X$. 
Then, $y$ must be $\mathcal R_2$-correct in $X$.
Since it has the $\#$-correct $\mathcal R_2$-neighbor $x$, $y$ must be $(\#,\mathcal R_2)$-correct.

\end{enumerate}

Conversely, suppose that $y$ is an $\mathcal R_2$-neighbor of $x$ in $\mathcal{L}^{\#}$. 
Since $x$ is $(\#,\mathcal R_2)$-correct in $\xi_{\gamma}$, this means that $y\in\xi_{\gamma}$ and is $\#$-correct in $\xi_{\gamma}$.
There are two cases.

\begin{enumerate}

\item $y\in X_{\gamma}$.
Since $\gamma$ is a GFc of $X$, $y\in X$. 
By Lemma \ref{lem:GFc_particles_same_neighbors}, the Voronoi cells of $x$ and of $y$ are the same in $X$ and $\xi_{\gamma}$, but by their $\#$-correctness in $\xi_{\gamma}$, they have the same Voronoi cells in $\xi_{\gamma}$ and $\mathcal{L}^{\#}$. 
Since $x$ and $y$ are $\mathcal R_2$-neighbor in $\mathcal{L}^{\#}$, so must they be in $X$.

\item $y\in \xi_{\gamma}\setminus X_{\gamma}$.
By Lemma \ref{lem:canonical_non_GFc_particles_R_correct}, $y$ has the same Voronoi cell in $\mathcal{L}^{\#}$ and $\xi_{\gamma}$. 
If $\mathrm{d}_{\Lambda_{\infty}}(V_{\mathcal{L}^{\#}}(y),V_{\mathcal{L}^{\#}}(x))\le 1$, then $y$ immediately surrounds $\gamma$, so necessarily $y\in X$. 
Else, joining $V_{\mathcal{L}^{\#}}(y)$ and $V_{\mathcal{L}^{\#}}(x)$ by a path in $\Lambda_{\infty}$ of minimum length ($\le\mathcal R_2$), the path contains a point in the exterior boundary of $\bar{\gamma}$ (see Definition \ref{def:boundaries}).
This point is then contained in the Voronoi cell of an $(\#,\mathcal R_2)$-correct particle $z$ in $X$.
Moreover, $\mathrm{d}_{\Lambda_{\infty}}(V_{\mathcal{L}^{\#}}(z),V_{\mathcal{L}^{\#}}(y))\le\mathcal R_2$, so $y$ is an $\mathcal R_2$-neighbor of $z$ in $\mathcal{L}^{\#}$.
Since $z$ is $(\#,\mathcal R_2)$-correct in $X$, this implies that $y\in X$.

\end{enumerate}

Having verified the definition of $(\#,\mathcal R_2)$-correctness, we conclude that $x$ is $(\#,\mathcal R_2)$-correct in $X$.
\end{proof}

\begin{proof}[Proof of Proposition \ref{prop:canonical_configuration_gfc_correspondence}]
The proposition follows immediately from Lemmas \ref{lem:canonical_non_GFc_particles_R_correct} and \ref{lem:GFc_particles_R_incorrect}.
\end{proof}

\section{Analysis of the $n$-staircase model} 
\label{appx:n_staircases}

In this appendix, we prove the following generalization of Lemma \ref{lem:n_staircases_local_configurations} to general $n$.

\begin{lemma}\label{thm:n_staircases_local_configurations}
If $n$ is even, then the density of any local configuration $X$ is maximized if and only if
\begin{equation}
  \textstyle
  \set{\pm(\frac{n}{2},\frac{n}{2}),\pm(-n,\frac{n}{2}),\pm(\frac{n}{2},-n)}\subseteq X,
  \label{n_staircase_local_even}
\end{equation}
and if $n$ is odd, the density is maximized if and only if
\begin{equation}
  \textstyle
  \set{\pm(\frac{n+1}{2},\frac{n-1}{2}),\pm(-n,\frac{n+1}{2}),\pm(\frac{n-1}{2},-n)}\subseteq X
  \quad\mathrm{or}\quad
  \set{\pm(\frac{n-1}{2},\frac{n+1}{2}),\pm(-n,\frac{n-1}{2}),\pm(\frac{n+1}{2},-n)}\subseteq X
  .
  \label{n_staircase_local_odd}
\end{equation}
\end{lemma}

Our strategy is to study the contribution to the inverse of the local density from three upper triangular regions surrounding the particle at $\vec{0}$.
Specifically, let
\begin{equation}
U_{(a,b)}^{N}:=\set{(x,y)\in\Z^{2}\mid x\le a,y\le b,x+y\ge a+b-N+1}.
\end{equation}
We will study the contribution to the local density \eqref{eqn:local_density} from
\begin{equation}\label{eqn:regions_of_interest}
U_{(n-1,n-1)}^{n-1},\ U_{(-1,n-1)}^{n-1}\text{, and }U_{(n-1,-1)}^{n-1},
\end{equation}
which are, respectively, the upper triangular regions to the northeast, northwest, and southeast of the staircase at $\vec{0}$; see Figure \ref{fig:5staircase_regions}.

  \begin{figure}[h]
	\centering
    \includegraphics[width=6cm]{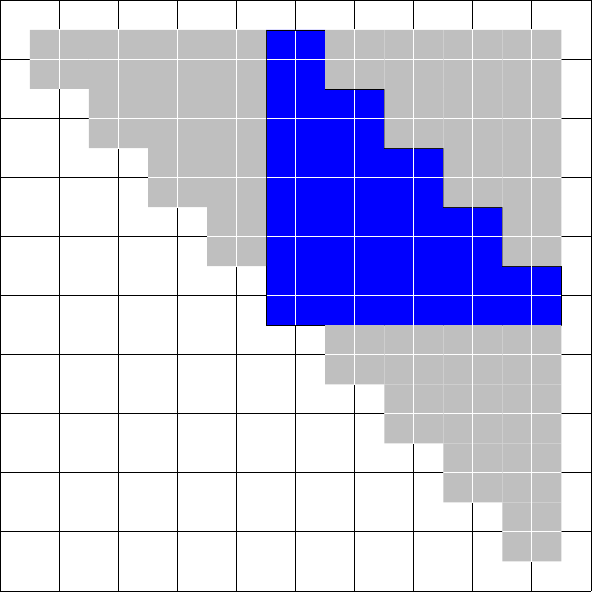}
    \caption{We study the contribution to the local density at the blue particle from the three upper triangular regions around it, which are shown in light gray.}
    \label{fig:5staircase_regions}
  \end{figure}

\begin{proof}[Proof of Lemma \ref{thm:n_staircases_local_configurations}]
We seek to minimize the quantity 
\begin{equation}\label{eqn:quantityofinterest}
\rho_{X}(\vec{0})^{-1}=\sum_{\lambda\in V_{X}(\sigma_{\vec{0}})}\frac{1}{\abs{\set{z\in X\mid\lambda\in V_{X}(\sigma_z)}}}
\end{equation}
over all configurations $X\in\Omega(\Lambda_{\infty})$ with $X\ni \vec{0}$. 
As discussed, we focus on the contribution to the RHS of \eqref{eqn:quantityofinterest} from the three upper triangular regions of \eqref{eqn:regions_of_interest}.
Notice that each of these regions intersects at most one $n$-staircase (which is clear by inspecting Figure \ref{fig:5staircase_regions}).
Depending on the nature of the intersection, we further restrict to the contribution to \eqref{eqn:quantityofinterest} from particular upper triangular subregions of empty sites.
The benefit of these subregions is that each contributes at least one-third of its volume to \eqref{eqn:quantityofinterest}; see Lemma \ref{lem:contribution} below.
Our choice of these subregions is formally described in Table \ref{tab:nstaircase} and illustrated in Figure \ref{fig:5staircase_cases}.
Informally:
\begin{enumerate}
\item Consider first $U^{n-1}_{(n-1,n-1)}$, the upper triangular region to the northeast of the staircase at $\vec{0}$. 
Notice that, if this region intersects another staircase, then the bottom left corner of that staircase must lie inside this region.
If this intersection does occur:
\begin{enumerate}
\item If the bottom left corner of the staircase is right adjacent to the staircase at $\vec{0}$, then the staircase cuts the upper triangular region $U^{n-1}_{(n-1,n-1)}$ into two smaller upper triangular regions (one of which might be empty) consisting entirely of empty sites.
We select these two regions; see Figure \ref{subfig:5staircase_case_1-1}.
\item Else, we select the largest upper triangular subregion of $U^{n-1}_{(n-1,n-1)}$ lying to the left of the intersecting staircase, as well as the largest upper triangular subregion of $U^{n-1}_{(n-1,n-1)}$ that lies strictly below the first; see Figure \ref{subfig:5staircase_case_1-2}.
\end{enumerate}
Finally, if no staircase intersects $U^{n-1}_{(n-1,n-1)}$, we select its top left and bottom right corners, as well as the largest upper triangular region sandwiched between them; see Figure \ref{subfig:5staircase_case_1-3}.

\item Consider next $U^{n-1}_{(-1,n-1)}$, the upper triangular region to the northwest of the staircase at $\vec{0}$. 
Notice that, if this region intersects another staircase, then the bottom right corner of that staircase must lie inside this region.
If this intersection does occur:
\begin{enumerate}
\item If the bottom right corner of the staircase is right adjacent to the staircase at $\vec{0}$, then the staircase cuts the upper triangular region $U^{n-1}_{(-1,n-1)}$ into two smaller upper triangular regions (one of which might be empty) consisting entirely of empty sites.
We select these two regions; see Figure \ref{subfig:5staircase_case_2-1}.
\item Else, we select the largest upper triangular subregion of $U^{n-1}_{(-1,n-1)}$ lying to the right of the intersecting staircase, as well as the largest upper triangular subregion of $U^{n-1}_{(-1,n-1)}$ that lies strictly below the first; see Figure \ref{subfig:5staircase_case_2-2}.
\end{enumerate}
Finally, if no staircase intersects $U^{n-1}_{(-1,n-1)}$, we select its top right and bottom right corners, as well as the largest upper triangular region sandwiched between them; see Figure \ref{subfig:5staircase_case_2-3}.

\item The choice of the subregions for $U^{n-1}_{(n-1,-1)}$ is completely analogous by the $x\leftrightarrow y$ reflection symmetry.
\end{enumerate}

\begin{table}[h]
\centering
\renewcommand{\arraystretch}{1.5}
\begin{tabular}{|c|M{0.4\linewidth}|M{0.3\linewidth}|}
\hline
Region & 
Case & 
Subregions \\ \hline

\multirow{3}{*}{$U_{(n-1,n-1)}^{n-1}$} & 
Intersects a particle located at $(a,b)$ with $a+b=n$ & 
$U_{(a-1,n-1)}^{a-1}$ and $U_{(n-1,b-1)}^{b-1}$ \\ \cline{2-3} 

& 
Intersects a particle located at $(a,b)$ with $a+b\ge n+1$ &
$U^{a-1}_{(a-1,n-1)}$ and $U^{n-a}_{(n-1,n-a)}$ \\ \cline{2-3} 

&
Does not intersect any particle & 
$\set{(1,n-1)}$, $\set{(n-1,1)}$, and $U_{(n-2,n-2)}^{n-3}$ \\ \hline
 
\multirow{3}{*}{$U_{(-1,n-1)}^{n-1}$} & 
Intersects a particle located at $(a,b)$ with $a=-n$ & 
$U_{(-1,n-1)}^{n-b-1}$ and $U_{(-1,b-1)}^{b-1}$ \\ \cline{2-3} 

& 
Intersects a particle located at $(a,b)$ with $a\le -n-1$ & 
$U_{(-1,n-1)}^{-a-b-1}$ and $U_{(-1,a+b+n)}^{a+b+n}$ \\ \cline{2-3} 

& 
Does not intersect any particle & 
$\set{(-1,1)}$, $\set{(-1,n-1)}$, and $U_{(-1,n-2)}^{n-3}$ \\ \hline

$U_{(n-1,-1)}^{n-1}$ & 
\multicolumn{2}{c|}{Analogous to $U_{(-1,n-1)}^{n-1}$ by the $x\leftrightarrow y$ reflection symmetry}  \\ \hline
\end{tabular}
\caption{We further restrict our attention to the contribution to the local density from particular upper triangular subregions around the particle, which consist entirely of empty sites.}
\label{tab:nstaircase}
\end{table}

\begin{figure}[]
\begin{subfigure}[t]{0.3\textwidth}
\centering
\includegraphics[width=0.7\linewidth]{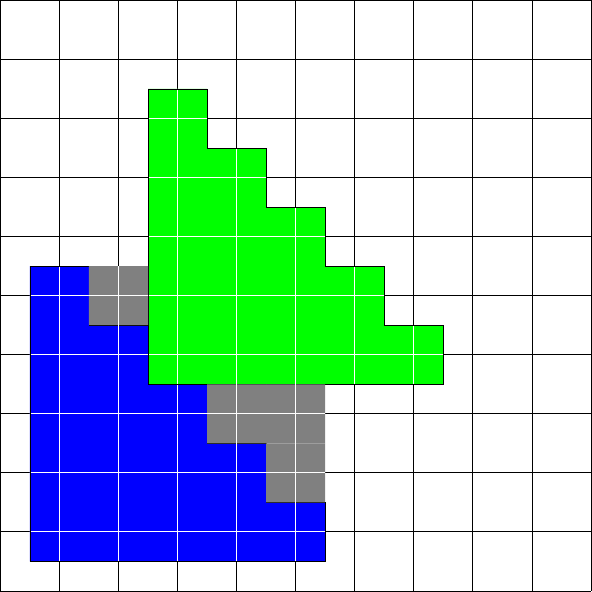}
\caption{$U_{(4,4)}^4$ intersects a particle located at $(2,3)$.}
\label{subfig:5staircase_case_1-1}
\end{subfigure}
\hspace{\fill}
\begin{subfigure}[t]{0.3\textwidth}
\centering
\includegraphics[width=0.7\linewidth]{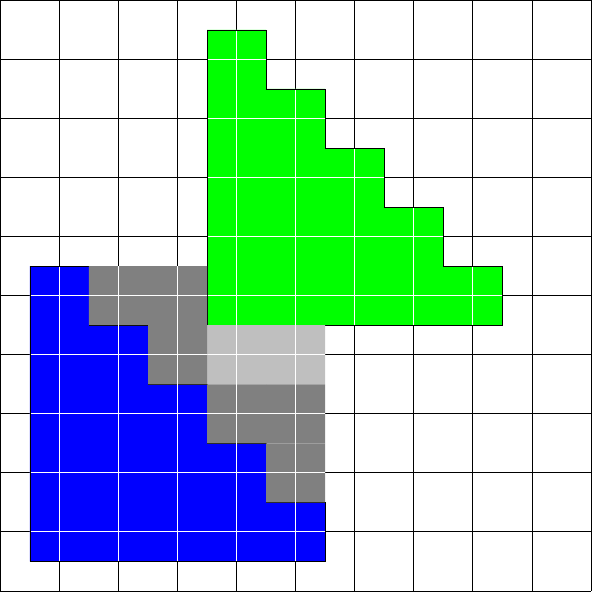}
\caption{$U_{(4,4)}^4$ intersects a particle located at $(3,4)$.}
\label{subfig:5staircase_case_1-2}
\end{subfigure}
\hspace{\fill}
\begin{subfigure}[t]{0.3\textwidth}
\centering
\includegraphics[width=0.7\linewidth]{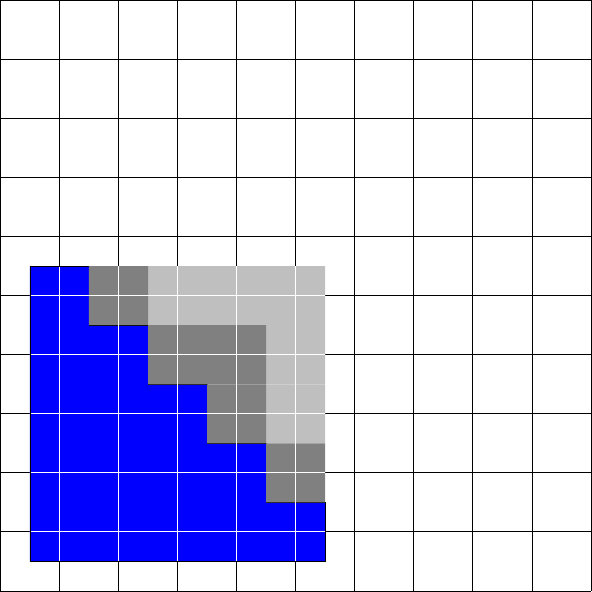}
\caption{$U_{(4,4)}^4$ does not intersect any particle.}
\label{subfig:5staircase_case_1-3}
\end{subfigure}

\bigskip

\begin{subfigure}[t]{0.3\textwidth}
\centering
\includegraphics[width=0.98\linewidth]{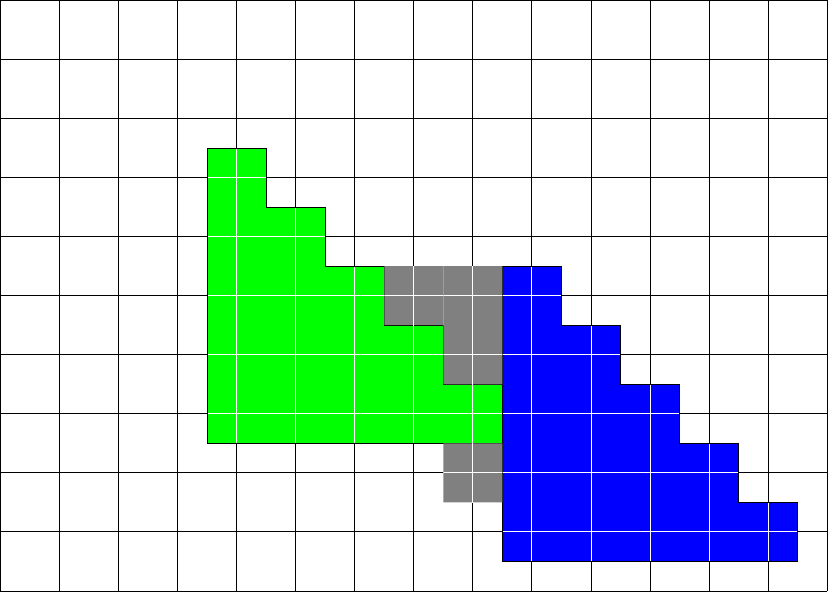}
\caption{$U_{(-1,4)}^4$ intersects a particle located at $(-5,2)$.}
\label{subfig:5staircase_case_2-1}
\end{subfigure}
\hspace{\fill}
\begin{subfigure}[t]{0.3\textwidth}
\centering
\includegraphics[width=0.98\linewidth]{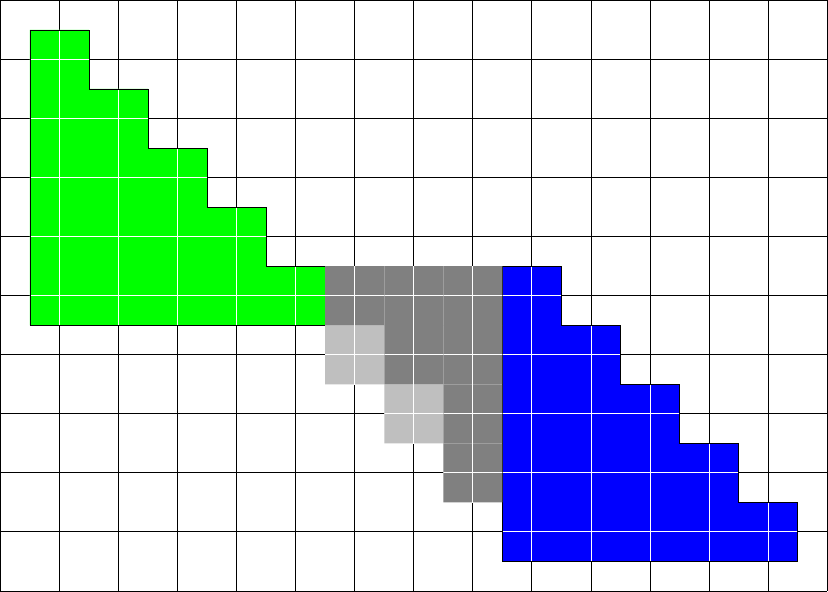}
\caption{$U_{(-1,4)}^4$ intersects a particle located at $(-8,4)$.}
\label{subfig:5staircase_case_2-2}
\end{subfigure}
\hspace{\fill}
\begin{subfigure}[t]{0.3\textwidth}
\centering
\includegraphics[width=0.98\linewidth]{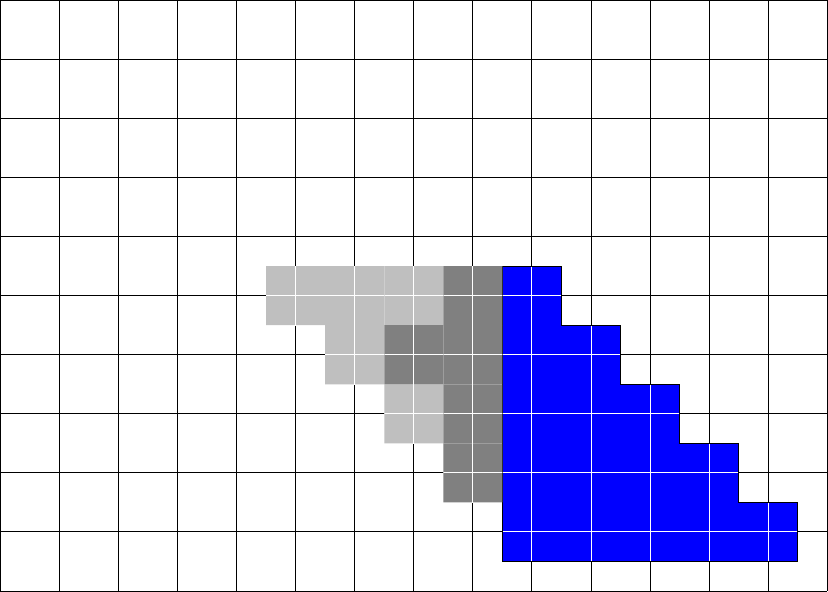}
\caption{$U_{(-1,4)}^4$ does not intersect any particle.}
\label{subfig:5staircase_case_2-3}
\end{subfigure}

\caption{Illustrations of the cases described in Table \ref{tab:nstaircase}. 
The particle at $\vec{0}$ is in blue. 
The particle that intersects one of the upper triangular regions around the particle at $\vec{0}$ is in green.
Light gray represents an empty site within the upper triangular region which we neglect. 
The subregions that we select are in dark gray.}
\label{fig:5staircase_cases}
\end{figure}

Comparing the sums of the volumes of the subregions in the last column of Table \ref{tab:nstaircase}, we bound the contribution from each of the original regions \eqref{eqn:regions_of_interest} by
\begin{equation}\label{eqn:single_ut_region_contribution_bound}
\begin{cases}
\frac{1}{12}n(n-2)&\mbox{$n$ even}
\\
\frac{1}{12}(n-1)^{2}&\mbox{$n$ odd}
\end{cases}.
\end{equation}
Table \ref{tab:nstaircasesummary} lists the necessary conditions for \eqref{eqn:single_ut_region_contribution_bound} to be attained.

\begin{table}[h]
\centering
\renewcommand{\arraystretch}{1.5}
\renewcommand\theadgape{\Gape[4pt]}
\renewcommand\cellgape{\Gape[4pt]}
\begin{tabular}{|M{0.1\linewidth}|M{0.46\linewidth}|}
\hline
Region &
Necessary condition for minimum contribution \\ \hline
$U_{(n-1,n-1)}^{n-1}$ & 
\makecell{$n$ even: $(\frac{n}{2},\frac{n}{2})\in X$ \\ $n$ odd: $(\frac{n+1}{2},\frac{n-1}{2})\in X$ or $(\frac{n-1}{2},\frac{n+1}{2})\in X$} \\ \cline{1-1} \cline{2-2} 
$U_{(-1,n-1)}^{n-1}$ & 
\makecell{$n$ even: $(-n,\frac{n}{2})\in X$ \\ $n$ odd: $(-n,\frac{n+1}{2})\in X$ or $(-n,\frac{n-1}{2})\in X$} \\ \cline{1-1} \cline{2-2} 
$U_{(-1,n-1)}^{n-1}$ & 
\makecell{$n$ even: $(\frac{n}{2},-n)\in X$ \\ $n$ odd: $(\frac{n-1}{2},-n)\in X$ or $(\frac{n+1}{2},-n)\in X$} \\ \hline
\end{tabular}
\caption{Necessary conditions for optimizing the contribution to the local density from each upper triangular region around the particle.}
\label{tab:nstaircasesummary}
\end{table}

Taking into account the volume of the particle $\sigma_{\vec{0}}$ itself, we obtain the following lower bound of \eqref{eqn:quantityofinterest}:
\begin{equation}\label{eqn:nstaircasefinalbound}
\frac{1}{2}n(n+1)+\begin{cases}
\frac{1}{4}n(n-2) & \mbox{$n$ even}\\
\frac{1}{4}(n-1)^{2} & \mbox{$n$ odd}
\end{cases},
\end{equation}
which is attained only if
\begin{enumerate}
\item all the necessary conditions in Table \ref{tab:nstaircasesummary} are satisfied, and
\item no site outside of $\sigma_{\vec{0}}\sqcup U_{(n-1,n-1)}^{n-1}\sqcup U_{(-1,n-1)}^{n-1}\sqcup U_{(n-1,-1)}^{n-1}$ contributes to \eqref{eqn:quantityofinterest}.
\end{enumerate}
It is easy to see that the local configurations in the statement of Lemma \ref{thm:n_staircases_local_configurations} are the only local configurations in which both of the above conditions are satisfied.
Moreover, in these local configurations, the bound \eqref{eqn:nstaircasefinalbound} is indeed attained.
\end{proof}

It remains to prove that each upper triangular subregion we selected in Table \ref{tab:nstaircase} contributes at least one-third of its volume to \eqref{eqn:quantityofinterest}.
To this end, we will exploit the following \emph{triangular} symmetry.
Very soon, we will partition each upper triangular subregion into  orbits under the symmetry, from which the factor of $1/3$ will appear naturally.

\begin{lemma}[triangular symmetry]\label{lem:utsymmetry}
Let $(a,b)\in\Z^{2}$ and $N\in\N$.
Restricted to $\vec{x}\in U^{N}_{(a,b)}$, the functions 
\begin{equation}
\Delta_{1}(x,y):= a+1-x,\ \Delta_{2}(x,y):= b+1-y,\ \text{and }\Delta_{3}(x,y):= (x+y)-(a+b-N),
\end{equation} 
give the (graph) distance from $\vec{x}$ to the sets of points on $\partial^{\text{ex}}U^{N}_{(a,b)}$,
\begin{equation}
\begin{split}
P_{1}(U^{N}_{(a,b)})&:=\set{(a+1,y)\in\Z^{2}\mid b-N+1\le y\le b},\\
P_{2}(U^{N}_{(a,b)})&:=\set{(x,b+1)\in\Z^{2}\mid a-N+1\le x\le a},\\
P_{3}(U^{N}_{(a,b)})&:=\set{(x,a+b-N-x)\in\Z^{2}\mid a-N\le x\le a},
\end{split}
\end{equation}
respectively; see Figure \ref{fig:upper_triangular_boundary}.
There exists a unique order-3 automorphism $T$ of $\Z^{2}$ that fixes $U_{(a,b)}^{N}$ and permutes the functions $\Delta_{i}$ by $\Delta_{i}\circ T=\Delta_{i+1}$, where the subscripts are understood modulo $3$.
\end{lemma}

  \begin{figure}[h]
	\centering
    \includegraphics[width=6cm]{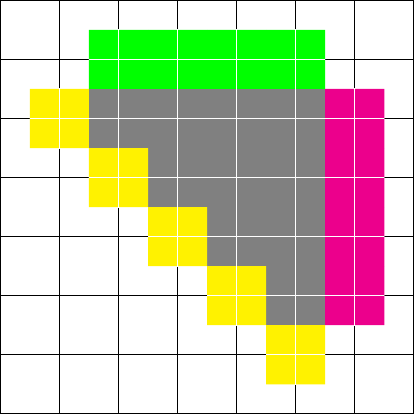}
    \caption{The upper triangular region $U^{4}_{(0,0)}$ is drawn in gray.
    The point sets $P_{1}(U^{4}_{(0,0)})$, $P_{2}(U^{4}_{(0,0)})$, and $P_{3}(U^{4}_{(0,0)})$ are respectively to the east (magenta), north (green), and southwest (yellow) of $U^{4}_{(0,0)}$.}
    \label{fig:upper_triangular_boundary}
  \end{figure}

\begin{proof}
Solving the system of equations $\Delta_{i}\circ T=\Delta_{i+1}$, $i=1,2,3$, yields the affine transformation
\begin{equation}
T(x,y):=(y+a-b,-x-y+a+2b-N+1),
\end{equation} 
which satisfies the remaining properties by direct computation.
\end{proof}

We are now ready to prove that each subregion we selected in Table \ref{tab:nstaircase} contributes at least $1/3$ of its volume to the local density.

\begin{lemma}[contribution from a subregion]\label{lem:contribution}
Consider an upper triangular region $U_{(a,b)}^{N}$. Let $(i,j,k)$ be any permutation of $(1,2,3)$. If a configuration $X\in\Omega(\Lambda_{\infty})$ with $X\ni\vec{0}$ is such that
\begin{enumerate}
\item $U_{(a,b)}^{N}$ consists only of empty sites,
\item $P_{i}(U^{N}_{(a,b)})$ is fully covered by the $n$-staircase located at $\vec{0}$, and
\item $P_{j}(U^{N}_{(a,b)}),P_{k}(U^{N}_{(a,b)})$ each intersects no more than one $n$-staircase,
\end{enumerate}
then
\begin{equation}\label{eqn:minimumcontribution}
\sum_{e\in U_{(a,b)}^{N}\cap V_{X}(\sigma_{\vec{0}})}\frac{1}{\abs{\set{z\in X\mid e\in V_{X}(\sigma_z)}}}\ge\frac{1}{3}\abs{U_{(a,b)}^{N}}.
\end{equation}
Notice that, by inspecting Figure \ref{fig:5staircase_cases}, each subregion we selected in Table \ref{tab:nstaircase} satisfies these conditions.
\end{lemma}

\begin{proof}
Let $T$ be the automorphism associated to $U_{(a,b)}^{N}$ by Lemma \ref{lem:utsymmetry}. 
Let $O\in U_{(a,b)}^{N}/\gen{T}$ be an orbit and $\vec{x}\in O$.
Consider the multiset 
\begin{equation}
\label{eqn:distance_multiset}
M:=\set{\Delta_{1}(\vec{x}),\Delta_{2}(\vec{x}),\Delta_{3}(\vec{x})}
=\set{\Delta_{i}(\vec{x}),\Delta_{i}(T\vec{x}),\Delta_{i}(T^{2}\vec{x})},
\end{equation}
which is independent of the choice of $\vec{x}$. 
If $O$ is a singleton, then \eqref{eqn:distance_multiset} implies that $\vec{x}$ is equidistant to $P_{\ell}(U^{N}_{(a,b)})$, $\ell=1,2,3$, so $\vec{x}$ contributes at least $1/3$ to the LHS of \eqref{eqn:minimumcontribution}.
Otherwise, $\abs{O}=3$, and there are three cases:
\begin{enumerate}
\item $M$ contains three equal numbers.
Then, each point in $O$ contributes at least $1/3$ to the LHS of \eqref{eqn:minimumcontribution}.
\item $M$ has a unique minimum $\Delta_{i}(T^{\ell}\vec{x})$.
Then, $T^{\ell}\vec{x}$ contributes $1$ to the LHS of \eqref{eqn:minimumcontribution}.
\item $M$ has exactly two minima: $\Delta_{i}(T^{\ell}\vec{x})$ and $\Delta_{i}(T^{\ell+1}\vec{x})$.
Then, $T^{\ell}\vec{x}$ and $T^{\ell+1}\vec{x}$ each contribute at least $1/2$ to the LHS of \eqref{eqn:minimumcontribution}.\qedhere
\end{enumerate}
\end{proof}

\section{Summary of key estimates}
\label{appx:summary}

Here, we recapitulate the important constants and estimates in our analysis of the high-density behavior of hard-core lattice particle models satisfying Assumption \ref{assumption}.

In the proof of Proposition \ref{prop:central_estimates}, we show that the weights of the GFc's that we construct in \S\ref{subsec:gfc} satisfy the Peierls condition
\begin{equation}\tag{\ref{eqn:weightestimate}}
\abs{w_{\mathbf{z}}^{\#}(\gamma)}\le e^{-\tau\abs{\bar{\gamma}}},
\end{equation}
where the Peierls constant $\tau$ satisfies 
\begin{equation}\tag{\ref{eqn:velenik5-28}}
\sum_{n=1}^{\infty} e^{-(\tau-\chi)n}\left(2\chi^{2}\abs{\mathcal{G}}^\chi\right)^{n}
\le 1
\end{equation}
to ensure the convergence of the cluster expansion on the GFc model, and
\begin{equation}\tag{\ref{eqn:derivativecondition2}}
\sum_{s=1}^{\infty} e^{-(\tau-\chi-1)s}\left(2\chi^{2}\abs{\mathcal{G}}^\chi\right)^s\le\frac{\eta}{3I_dd!}
\end{equation}
which ultimately gives us control on the one-point correlation functions as required by Theorem \ref{thm:crystallization}.
In \eqref{eqn:velenik5-28} and \eqref{eqn:derivativecondition2}, the maximal coordination number $\chi$ of the underlying graph $\Lambda_{\infty}$ and the number of ground states $\abs{\mathcal{G}}$ arise out of the need to control the entropy of the GFc's.
The constant $I_{d}$ appears in the $d$-dimensional isoperimetric inequality
\begin{equation}\tag{\ref{eqn:isoperimetric}}
\abs{\Int\gamma}\le I_d\abs{\bar{\gamma}}^{d}.
\end{equation}
Finally, $\eta$, which controls the following derivatives of the partition functions:
\begin{equation}\tag{\ref{eqn:derivative_estimate}}
\abs{\frac{\partial}{\partial\log \mathbf{z}(x_{i})}\log\frac{\Xi^{\#}_{\mathbf{z}}(\Lambda)}{\mathbf{z}^{\#}(\Lambda)}}\le \eta\indicator{x_{i}\in\Lambda},
\end{equation} 
satisfies simply that
\begin{equation}\tag{\ref{eqn:derivativecondition1}}
\eta\le 1.
\end{equation}

The proof of the Peierls bound for the weights of the GFc's requires control of the ratio of two partition functions with possibly different boundary conditions.
In the end, we prove that, up to certain powers of the fugacity $z$, such ratios are essentially boundary terms:
\begin{equation}\tag{\ref{eqn:ratio_estimate}}
\abs{\frac{\Xi^{\#'}_{\mathbf{z}}(\Lambda)}{\Xi^{\#}_{\mathbf{z}}(\Lambda)}}\le\frac{\abs{z}^{\abs{\Lambda\cap \mathcal{L}^{\#'}}}}{\abs{z}^{\abs{\Lambda\cap \mathcal{L}^{\#}}}}e^{\varsigma\abs{\partial^{\text{in}}\Lambda}},
\end{equation}
where the coefficient $\varsigma$ satisfies
\begin{equation}
\tag{\ref{eqn:boundary-coefficient-condition-1}}
\varsigma\ge 2c
\end{equation}
and
\begin{equation}
\tag{\ref{eqn:boundary-coefficient-condition-2}}
\varsigma\ge 3n(e^{\frac{c}{n}}+1)+2+2\mu^{-1}.
\end{equation}
Finally, the constants $n$ and $c$ appear out of the need to compute correlation functions by differentiating the partition functions.
Recall from the statement of Proposition \ref{prop:central_estimates} that we assume that the site-wise fugacity $\mathbf{z}(x)$ is equal to $z$ for all but $n$ sites, on which it is allowed to deviate only slightly from $z$: $e^{-\frac{c}{n}}\abs{z}\le\abs{\mathbf{z}}\le e^{\frac{c}{n}}\abs{z}$.
Note that $c$ can be taken to be arbitrarily small but must remain positive.

All the above computations take place at high fugacities:
\begin{equation}
\tag{\ref{eqn:high-fugacity}}
\mu(\rho_{\max}-\rho_{0})\log z-2c-\varsigma\chi\ge\tau,
\end{equation}
where $\mu$, defined in Lemma \ref{lem:pointwiselowerbound}, denotes the minimum fractional weight assigned to a point in the reference Voronoi cells of a ground state, and $\rho_{\max},\rho_{0}$ are the only places where the formulation of the Peierls condition in terms of the effective volume enters; see Lemma \ref{lem:Peierls}.

We note that, as made evident by \eqref{eqn:boundary-coefficient-condition-2} and \eqref{eqn:high-fugacity}, our analysis requires progressively higher fugacities to control higher-order correlation functions.
Theorem \ref{thm:crystallization}, however, requires only that we consider the case $n=1$.

\medskip

In summary, we prove analyticity in the domain (see \eqref{eqn:high-fugacity})
\begin{equation}
  \abs{z}>z_0:=\exp\left(\frac{\tau+\varsigma\chi}{\mu(\rho_{\mathrm{max}}-\rho_0)}\right)
  \label{estimate_z}
\end{equation}
where $\rho_{\mathrm{max}}$ is the maximum density \eqref{rhomax};
and (see \eqref{rho0}, \eqref{calN}, \eqref{calS0})
\begin{equation}
  \rho_0:=\frac1{\rho_{\mathrm{max}}^{-1}+\frac\epsilon{\mathcal N}}
  ,\quad
  \mathcal N:=\max_{X\in\Omega(\Lambda_{\infty})}\abs{
  \set{x\in X\mid\mathrm d_{\Lambda_{\infty}}(x,0)\le\mathcal S_0}
  },
  \quad
  \mathcal S_0:=\mathcal S_1+\mathcal R_2+4 r_{\mathrm{eff}},
\end{equation}
in which $\epsilon$, $\mathcal R_1$ and $\mathcal S_1$ appear in Item \ref{asm:density_local_density} of Assumption \ref{assumption}, $\mathcal R_2$ satisfies \eqref{ineq_R2_1}, \eqref{ineq_R2_2}, and \eqref{ineq_R2_3}:
\begin{equation}
  \mathcal R_2\ge \max\set{\mathcal R_0,\mathcal R_1}
  ,\quad
  \mathcal R_2>\max\set{\mathrm{d}_{\Lambda_{\infty}}(x,y)\mid x,y\in\Lambda,\ \omega_{x}\cap\omega_{y}\neq\emptyset},
\end{equation}
and $r_{\mathrm{eff}}$ appears in Lemma\-~\ref{lem:finiteeffectiveparticle};
$\mu$ is defined in Lemma \ref{lem:pointwiselowerbound}: (see also \eqref{sigmasharp} and \eqref{vsharp})
\begin{equation}
  \mu=\min_{\lambda\in V_{\mathcal L^\#}(\sigma_x)}\frac{1}{\abs{\set{z\in \mathcal{L}^{\#}\mid\lambda\in\sigma^{\#}_{z}}}};
\end{equation}
$\chi$ is the maximal coordination number of $\Lambda_{\infty}$ (the number of neighbors of each site);
$\tau$ satisfies \eqref{eqn:velenik5-28} and \eqref{eqn:derivativecondition2}:
\begin{equation}
  \tau=\chi+1+\log\left(2\chi^2\abs{\mathcal{G}}^\chi\left(1+\frac1{3I_dd!}\right)\right)
\end{equation}
where we used\-~\eqref{eqn:derivativecondition1}, $\abs{\mathcal{G}}$ is the number of ground states, and (see \eqref{eqn:isoperimetric})
\begin{equation}
  I_d:=\sup_{\Lambda\Subset\Lambda_{\infty},\ \mathrm{connected}}\frac{\abs{\Lambda}}{\abs{\partial^{\ex} \Lambda}^d}
\end{equation}
where $\partial^{\ex} \Lambda$ is the exterior boundary of $\Lambda$ (see Definition \ref{def:boundaries});
and (see \eqref{eqn:boundary-coefficient-condition-2})
\begin{equation}
  \varsigma=8+2\mu^{-1}
  .
\end{equation}

\end{document}